\documentclass[10pt,journal]{IEEEtran}

\usepackage[noadjust]{cite} 
\usepackage[pdftex]{graphicx}
\usepackage{amsmath}
\usepackage{hyperref}
\hypersetup{colorlinks=true, linkcolor=blue, citecolor=blue, urlcolor = blue}
\usepackage{amssymb}
\usepackage{amsthm}
\usepackage{mathtools}
\usepackage{mathrsfs}
\usepackage{enumitem}
\usepackage{mleftright}
\mleftright
\usepackage{algorithmic}

\usepackage{subfig}
\usepackage[font=small]{caption}
\usepackage{stfloats}
\usepackage{url}
\usepackage{xcolor}
\begin{document}	

\title{Lossless Source Coding in the Point-to-Point, \\ Multiple Access, and Random Access Scenarios}
\author{Shuqing~Chen,~\IEEEmembership{Graduate Student Member,~IEEE,}
            Michelle~Effros,~\IEEEmembership{Fellow,~IEEE,}
        and~Victoria~Kostina,~\IEEEmembership{Member,~IEEE}
      	\IEEEauthorrefmark{1}
\thanks{Manuscript received September 4, 2019; 
	revised May 27, 2020; accepted June 8, 2020}%
\thanks{This work is supported in part by the National Science Foundation under Grants CCF-1817241 and CCF-1956386. The work of S. Chen is supported in part by the Oringer Fellowship Fund in Information Science and Technology. This paper was presented in part at the 2019 IEEE International Symposium on Information Theory \cite{chen-e-k}. Matlab code for the computation of nonasymptotic bounds in this paper is available at github \cite{spectre}.}%
\thanks{Shuqing Chen was with the Department of Electrical Engineering, California
Institute of Technology, Pasadena, CA 91125 USA. She is now with Virtu
Financial Inc., New York, NY 10006 USA. (e-mail: {schen2@caltech.edu}).}
\thanks{Michelle Effros and Victoria Kostina are with the Department of Electrical Engineering, California Institute of Technology, Pasadena, CA 91125, USA. (e-mail: {effros@caltech.edu}, {vkostina@caltech.edu}).}%
\thanks{Communicated by I. Kontoyiannis, Associate Editor At Large. }
\thanks{Digital Object Identifier 10.1109/TIT.2020.3005155}
}



\maketitle

\begin{abstract}
This work studies point-to-point, multiple access, 
and random access lossless source coding 
in the finite-blocklength regime. 
In each scenario, a random coding technique is developed 
and used to analyze third-order coding performance.
Asymptotic results include a third-order characterization 
of the Slepian-Wolf rate region 
with an improved converse 
that relies on a connection to composite hypothesis testing. 
For dependent sources, 
the result implies that 
the independent encoders used by Slepian-Wolf codes 
can achieve the same third-order-optimal performance 
as a single joint encoder. 
The concept of random access source coding is introduced
to generalize multiple access (Slepian-Wolf) source coding 
to the case where encoders decide independently whether or not to participate 
and the set of participating encoders is unknown {\em a priori} 
to both the encoders and the decoder.
The proposed random access source coding 
strategy employs rateless coding with scheduled feedback.  
A random coding argument proves 
the existence of a single deterministic code of this structure 
that simultaneously achieves the third-order-optimal Slepian-Wolf performance 
for each possible active encoder set. 
\end{abstract}

\begin{IEEEkeywords}
Lossless source coding, Slepian-Wolf, random access, finite blocklength, random coding, non-asymptotic information theory, Gaussian approximation, hypothesis testing, meta-converse.
\end{IEEEkeywords}

\theoremstyle{plain}
\newtheorem{thm}{Theorem}
\newtheorem{cor}[thm]{Corollary}
\newtheorem{lem}[thm]{Lemma}
\newtheorem{prop}[thm]{Proposition}
\newtheorem{defn}{Definition}
\theoremstyle{remark}
\newtheorem{remark}{Remark}
\renewcommand\qedsymbol{$\blacksquare$}
\newcommand\mydots{\makebox[1em][c]{.\hfil.\hfil.}}


%

\section{Introduction}
%
%
%
%



\IEEEPARstart{W}{e} study the fundamental limits 
of fixed-length, finite-blocklength lossless source coding in three scenarios: 
\begin{enumerate}
	\item \emph{Point-to-point}: A single source 
	is compressed by a single encoder and decompressed by a single decoder.
	\item \emph{Multiple access}: 
	Each source in a fixed set of sources 
	is compressed by an independent encoder; 
	all sources are decompressed by a joint decoder. 
	\item \emph{Random access}: 
	Each active source from some set of possible sources 
	is compressed by an independent encoder; 
	all active sources are decompressed by a joint decoder.
\end{enumerate}

The information-theoretic limit 
in any lossless source coding scenario 
is the set of code sizes or rates 
at which a desired level of reconstruction error is achievable. 
Shannon's theory~\cite{shannon} analyzes this fundamental limit 
by allowing an arbitrarily long encoding blocklength 
in order to obtain a vanishing error probability. 
Finite-blocklength limits~\cite{strassen, pol-poo-ver, kontoyiannis-verdu, kostina-pol-verdu}, 
which are of particular interest 
in delay-sensitive and computationally-constrained coding environments, 
allow a non-vanishing error probability 
and study refined asymptotics of the rates achievable 
with encoding blocklength $n$.
Due to their non-vanishing error probability, 
the resulting codes are sometimes called ``almost-lossless'' source codes.  
We here use the term ``source coding'' 
to refer to this almost-lossless coding paradigm. 

In point-to-point source coding, 
non-asymptotic bounds and asymptotic expansions 
of the minimum achievable rate appear in~\cite{yushkevich,strassen,han,hayashi, kontoyiannis-verdu}. 
In \cite{kontoyiannis-verdu}, Kontoyiannis and Verd\'{u} 
analyze the optimal code 
to give a \emph{third-order} characterization 
of the minimum achievable rate $R^{*}(n,\epsilon)$ 
at blocklength $n$ and error probability $\epsilon$. 
For a finite-alphabet, stationary, memoryless source 
with single-letter distribution $P_X$, entropy $H(X)$, and varentropy $V(X)>0$, 
\begin{equation}
R^{*}(n,\epsilon) \approx H(X) + \sqrt{\frac{V(X)}{n}}Q^{-1}(\epsilon) - \frac{\log n}{2n}, \label{eq-intro-1}
\end{equation} 
where $Q^{-1}(\cdot)$ is 
the inverse complementary Gaussian distribution function, and 
any higher-order term is bounded by $O\big(\frac{1}{n}\big)$.  

For a multiple access source code (MASC), 
also known as a Slepian-Wolf (SW) source code~\cite{slepian-wolf}, 
the fundamental limit is the set of achievable rate tuples 
known as the rate region. 
The \emph{first-order} rate region 
for stationary, memoryless and general sources 
appears in~\cite{slepian-wolf} and~\cite{miyake-kanaya, han}, 
respectively. 
{\em Second-order} asymptotic expansions of the MASC rate region 
for stationary, memoryless sources 
appear in~\cite{tan-kosut, nomura-han}. 
Tan and Kosut's characterization~\cite{tan-kosut} 
is similar in form to the first two terms of \eqref{eq-intro-1}, 
with varentropy $V(X)$ replaced by the entropy dispersion matrix 
and third-order term bounded by $O\big(\frac{\log n}{n}\big)$.

For point-to-point source coding, 
our contributions include non-asymptotic characterizations 
of the performance of randomly designed codes using  
threshold and maximum-likelihood decoders. 
The former analysis demonstrates that 
combining random coding with the best possible threshold decoder 
cannot achieve $-\frac{\log n}{2n}$ in the third-order term in~\eqref{eq-intro-1}, and thus it is strictly sub-optimal. 
The latter shows that combining random coding 
with maximum likelihood decoding 
achieves the first three terms in \eqref{eq-intro-1}. 
We derive both bounds by deriving and analyzing 
a source coding analog to the random coding union (RCU) bound 
from channel coding~\cite[Th.~16]{pol-poo-ver}. 
Our asymptotic expansion is achieved by a random code 
rather than the optimal code from~\cite{kontoyiannis-verdu}.
Thus, there is no loss (up to the third-order term) 
due to random code design, 
which in turn shows that many codes have near-optimal performance; 
further, since our RCU bound 
holds when restricted to linear compressors, 
there are many good linear codes. 
The RCU bound is also important because it generalizes 
to the MASC and other scenarios 
where the optimal code is not known.  

Our MASC RCU bound 
yields a new MASC achievability bound (Theorem~\ref{thm-sw-rcu}). 
Establishing a link to composite hypothesis testing (HT) 
yields a new MASC HT converse (Theorem~\ref{thm-sw-cht-conv}), 
which extends the meta-converse for channel coding in~\cite{pol-poo-ver} 
to source coding with multiple encoders. 
This converse recovers and improves the previous converse 
due to Han \cite[Lemma~7.2.2]{han} 
and is equivalent to the LP-based converse 
of Jose and Kulkarni~\cite{jose-k}, 
which is the current best MASC converse.  
Our analysis of composite HT, 
including both non-asymptotic and asymptotic characterizations, 
develops tools with potential application 
in other multiple-terminal communication scenarios and beyond.
The MASC RCU bound and HT converse 
together yield the \emph{third-order} MASC rate region 
for stationary, memoryless sources (Theorem~\ref{thm-sw}), 
revealing a $-\frac{\log n}{2n}$ third-order term 
that is independent of the number of encoders. 
This tightens the $O\big(\frac{\log n}{n}\big)$ third-order bound 
from~\cite{tan-kosut}, 
which grows linearly with the source alphabet size 
and exponentially with the number of encoders. 
For dependent sources, 
the MASC's third-order-optimal sum rate 
equals the third-order-optimal rate 
achievable through joint encoding. 

While a MASC assumes a fixed, known collection of encoders, 
the set of transmitters communicating with a given access point in 
applications like sensor networks, the internet of things, 
and random access communication 
may be unknown or time-varying. 
The information theory literature treats 
the resulting \emph{channel coding} challenges 
in papers such as~\cite{minero-tse,ra-polyanskiy, recep}. 
We introduce the notion of a \emph{random access source code (RASC)} 
and tackle the resulting source coding challenges. 
The RASC  extends the MASC 
to scenarios where some encoders are inactive, 
and the decoder seeks to reliably reconstruct 
the sources associated with the active encoders 
assuming that the set of active encoders is unknown {\em a priori}.

We propose and analyze a robust RASC 
with rateless encoders that transmit codewords symbol by symbol 
until the receiver tells them to stop. 
Unlike typical rateless codes, 
which allow arbitrary decoding times~\cite{burnashev, tt, ppv2, draper}, 
our code employs a small set of decoding times. 
Single-bit feedback from the decoder 
to all encoders at each potential decoding time 
tells the encoders whether or not to continue transmitting.

We demonstrate (Theorem \ref{thm-rasc}) 
that there exists a single deterministic RASC 
that {\em simultaneously} achieves, 
for every possible set of active encoders, 
the third-order-optimal MASC performance 
for the active source set.  
Since traditional random coding arguments 
do not guarantee the existence of a single deterministic code 
that meets multiple independent constraints, 
prior code designs for multiple-constraint scenarios 
(e.g.,~\cite{ppv2}) 
employ a family of codes indexed using common randomness 
shared by all communicators. 
We develop an alternative approach, 
deriving a refined random coding argument 
(Lemma~\ref{lem-bad-code}) 
that demonstrates the existence of a single deterministic code 
that meets all our constraints simultaneously; 
this technique may eliminate the need for common randomness 
in other communication scenarios. 
For stationary, memoryless, permutation-invariant sources, 
employing identical encoders at all transmitters 
reduces RASC design complexity.  

Except where noted, 
all presented source coding results 
apply to both finite and countably infinite source alphabets. 

The organization of this paper is as follows. 
Section~\ref{sec-notation} defines notation. 
Section~\ref{sec-almost-lossless} treats (point-to-point) source coding. 
Section~\ref{sec-cht} studies composite HT, 
developing general tools for multiple-encoder communication scenarios. 
Section~\ref{sec-SW} treats the MASC. 
Section~\ref{sec-rasc} introduces and studies the RASC. 
Each of Sections~\ref{sec-almost-lossless}, \ref{sec-SW}, and \ref{sec-rasc} 
follows a similar flow:
\begin{enumerate}
\item For the (point-to-point) source code: 
Section~\ref{sec-def-almost-lossless} defines the problem. 
Section~\ref{sec-b-almost-lossless} provides historical background. 
Section~\ref{sec-result-almost-lossless} presents our new random coding 
achievability bounds and their asymptotic expansions.
\item For the MASC: 
Section~\ref{sec-def-sw} gives definitions.
Section~\ref{sec-b-SW} provides historical background.
Section~\ref{sec-sw-nonasymp} presents new non-asymptotic bounds.
Section~\ref{sec-result-SW} presents the third-order MASC characterization, 
comparing MASC and point-to-point source coding performance. 
Section~\ref{sec-sw-feedback} bounds the impact 
of limited feedback (and cooperation) on the third-order-optimal MASC region.
\item For the RASC: 
Section~\ref{sec-def-ra} defines the problem and describes our proposed code.
Section~\ref{sec-ra-b} highlights related work.
Section~\ref{sec-result-rasc} derives converse and achievability characterizations 
for our proposed code's finite-blocklength performance. 
Section~\ref{sec-rasc-perminv} 
treats the simplified code for permutation-invariant sources. 
\end{enumerate}
Section~\ref{sec-conclusion} contains concluding remarks. 
Proofs of auxiliary results appear in the appendices.

\section{Notation} \label{sec-notation}
For any positive integer $i$, let $[i]\triangleq\{1,\ldots,i\}$. 
We use uppercase letters (e.g., $X$) for random variables, 
lowercase letters (e.g., $x$) for scalar values, 
calligraphic uppercase letters (e.g., $\mathcal{E}$) 
for subsets of a sample space (events) or index sets, and 
script uppercase letters (e.g., $\mathscr{Q}$) 
for subsets of a Euclidean space. 
We use both bold face and superscripts for vectors 
(e.g., $\mathbf{x} = x^n$, $\mathbf{1} = (1,\ldots,1)$, 
and $\mathbf{0} = (0,\ldots,0)$). 
Given a sequence $(x_1,x_2,\ldots)$ 
with element $x_i$ in set $\mathcal{X}_i$ for each $i$ 
and given an ordered index set $\mathcal{T} \subseteq \mathbb{N}$, 
we define vector $\mathbf{x}_{\mathcal{T}} \triangleq (x_i, \; i \in \mathcal{T})$ 
and set $\mathcal{X}_{\mathcal{T}} \triangleq \prod_{i \in \mathcal{T}} \mathcal{X}_i$. 
Given a set $\mathcal{X}$, 
$\mathcal{X}^n$ is the $n$-fold Cartesian product of $\mathcal{X}$.
We denote matrices by sans serif uppercase letters (e.g., $\mathsf{V}$) 
and the $(i,j)$-th element of matrix $\mathsf{V}$ by $[\mathsf{V}]_{i,j}$. 
Inequalities between two vectors of the same dimension 
indicate elementwise inequalities. 
Given vector $\mathbf{u} \in \mathbb{R}^d$ 
and set $\mathscr{Q} \subset \mathbb{R}^d$, 
$\mathbf{u} + \mathscr{Q}$ denotes the Minkowski sum 
of $\{\mathbf{u}\}$ and $\mathscr{Q}$, 
giving $\mathbf{u} + \mathscr{Q} \triangleq 
\left\{\mathbf{u} + \mathbf{q}: \mathbf{q} \in \mathscr{Q} \right\}$. 
For two functions $u(n)$ and $f(n)$, 
$u(n) = O(f(n))$ if there exist $c, \, n_{0} \in \mathbb{R}_{+}$ 
such that $0 \leq u(n) \leq cf(n)$ for all $n > n_{0}$. 
For a $d$-dimensional function 
$\mathbf{u}: \mathbb{N} \rightarrow \mathbb{R}^d$, 
$\mathbf{u}(n) = O(f(n))\mathbf{1}$ 
if $u_i(n) = O(f(n))$ for all $i \in [d]$. 
For any finite set $\mathcal{A}$, 
$\mathcal{P}(\mathcal{A})$ represents the power set of $\mathcal{A}$ 
excluding the empty set, 
giving $\mathcal{P}(\mathcal{A})\triangleq
\{\mathcal{T}: \mathcal{T}\subseteq\mathcal{A}\}\setminus\emptyset$. 
We use $|\cdot|_+ \triangleq \max\{0,\cdot\}$. 
All uses of `$\log$' and `$\exp$', if not specified, 
employ an arbitrary common base, which determines the information unit.

Denote the standard and complementary 
Gaussian cumulative distribution functions (cdf) by $\Phi(z)$ and $Q(z)$, giving 
\begin{IEEEeqnarray}{rcl}
\Phi(z) &\, \triangleq\, & \frac{1}{\sqrt{2\pi}}\int_{-\infty}^{z} e^ {-\frac{u^2}{2}} du 
\label{def-gaussian-cdf} \\
Q(z) &\, \triangleq \,& 1- \Phi(z). \label{def-qfunc}
\end{IEEEeqnarray}
Function $Q^{-1}(\cdot)$ denotes the inverse of $Q(\cdot)$. 
The standard Gaussian probability density function is 
\begin{equation}
\phi(z) \triangleq \Phi'(z) = \frac{1}{\sqrt{2\pi}}e^ {-\frac{z^2}{2}}.
\end{equation} 
The $d$-dimensional generalization of the Gaussian cdf is 
\begin{IEEEeqnarray}{rCl}
\lefteqn{\Phi(\mathsf{V}; \mathbf{z})
\triangleq\Phi(\mathsf{V}; z_1,\ldots,z_d)} \label{eq-def-gaussian} \\ 
&\triangleq& \frac{1}{\sqrt{(2\pi)^{d}|\mathsf{V}|}} 
\int_{-\infty}^{z_1} \mydots \int_{-\infty}^{z_d} 
e^{-\frac{1}{2}\sum\limits_{i,j=1}^{d} u_{i}u_{j}[\mathsf{V}^{-1}]_{i,j}} du_{d}
\ldots du_{1}.  \nonumber 
\end{IEEEeqnarray}

Given an ordered index set $\mathcal{T} \subset \mathbb{N}$, 
let $P_{\mathbf{X}_\mathcal{T}}$ be a distribution 
defined on countable alphabet $\mathcal{X}_\mathcal{T}$. 
For any $\mathcal{A}, \mathcal{B} \subseteq \mathcal{T}$ 
with $\mathcal{A} \cap \mathcal{B} = \emptyset$ 
and any $(\mathbf{x}_{\mathcal{A}},\mathbf{x}_{\mathcal{B}})
\in \mathcal{X}_{\mathcal{A}} \times \mathcal{X}_{\mathcal{B}}$, 
the information and conditional information are defined as 
\begin{IEEEeqnarray}{rCl} 
\imath(\mathbf{x}_{\mathcal{A}}) 
&\triangleq& \log \frac{1}
	{P_{\mathbf{X}_{\mathcal{A}}}(\mathbf{x}_{\mathcal{A}})} \\
\imath(\mathbf{x}_{\mathcal{A}}|\mathbf{x}_{\mathcal{B}}) 
&\triangleq& \log \frac{1}
	{P_{\mathbf{X}_{\mathcal{A}}|\mathbf{X}_{\mathcal{B}}}
	(\mathbf{x}_{\mathcal{A}}|\mathbf{x}_{\mathcal{B}})}.
\end{IEEEeqnarray} 
The corresponding entropy, conditional entropy, varentropy, conditional varentropy, third centered moment of information, and third centered moment of conditional information are defined by, respectively, 
\begin{IEEEeqnarray}{rCl}
H(\mathbf{X}_{\mathcal{A}}) &\triangleq&  \mathbb{E}\left[\imath(\mathbf{X}_{\mathcal{A}})\right] \\
H(\mathbf{X}_{\mathcal{A}}|\mathbf{X}_{\mathcal{B}}) &\triangleq&  \mathbb{E}\left[\imath(\mathbf{X}_{\mathcal{A}}|\mathbf{X}_{\mathcal{B}})\right] \\
V(\mathbf{X}_{\mathcal{A}}) &\triangleq&  \text{Var}\left[\imath(\mathbf{X}_{\mathcal{A}})\right] \\
V(\mathbf{X}_{\mathcal{A}}|\mathbf{X}_{\mathcal{B}}) &\triangleq&  \text{Var}\left[\imath(\mathbf{X}_{\mathcal{A}}|\mathbf{X}_{\mathcal{B}})\right] \\
T(\mathbf{X}_{\mathcal{A}}) &\triangleq&  \mathbb{E}\left[|\imath(\mathbf{X}_{\mathcal{A}})-H(\mathbf{X}_{\mathcal{A}})|^3\right] \\
T(\mathbf{X}_{\mathcal{A}}|\mathbf{X}_{\mathcal{B}}) &\triangleq&  \mathbb{E}\left[|\imath(\mathbf{X}_{\mathcal{A}}|\mathbf{X}_{\mathcal{B}})-H(\mathbf{X}_{\mathcal{A}}|\mathbf{X}_{\mathcal{B}})|^3\right].
\end{IEEEeqnarray}

We also define random variables 
\begin{align}
 V_c(\mathbf{X}_{\mathcal{A}}|\mathbf{X}_{\mathcal{B}}) &\triangleq  \mathbb{E} \left[ \left( \imath(\mathbf{X}_{\mathcal{A}}|\mathbf{X}_{\mathcal{B}}) - \mathbb E \left[ \imath(\mathbf{X}_{\mathcal{A}}|\mathbf{X}_{\mathcal{B}}) | \mathbf{X}_{\mathcal{B}} \right] \right)^2 | \mathbf{X}_{\mathcal{B}} \right]\\
 T_c(\mathbf{X}_{\mathcal{A}}|\mathbf{X}_{\mathcal{B}}) &\triangleq  \mathbb{E} \left[ \left| \imath(\mathbf{X}_{\mathcal{A}}|\mathbf{X}_{\mathcal{B}}) - \mathbb E \left[ \imath(\mathbf{X}_{\mathcal{A}}|\mathbf{X}_{\mathcal{B}}) | \mathbf{X}_{\mathcal{B}} \right] \right|^3 | \mathbf{X}_{\mathcal{B}} \right].
\end{align}

\section{Point-to-Point Source Coding}
\label{sec-almost-lossless}
\subsection{Definitions}
\label{sec-def-almost-lossless}
In point-to-point source coding, 
the encoder maps a discrete random variable $X$ 
defined on finite or countably infinite alphabet $\mathcal{X}$ 
into a message from codebook $[M]$. 
The decoder reconstructs $X$ 
from the compressed description. 
Formal definitions of codes 
and their information-theoretic limits follow.
For prior definitions, see, for example,~\cite[Chapter~1]{han}.

\begin{defn}[Point-to-point source code] \label{def-almost-lossless}
An $(M, \epsilon)$ code for a random variable $X$ 
with discrete alphabet $\mathcal{X}$ 
comprises an encoding function 
$\mathsf{f}\colon \mathcal{X} \rightarrow [M]$ and 
a decoding function 
$\mathsf{g}\colon [M] \rightarrow \mathcal{X}$ 
with error probability 
$\mathbb{P}\left[\mathsf{g}(\mathsf{f}(X)) \neq X\right] \leq \epsilon$.
\end{defn}

\begin{defn}[Block point-to-point source code] \label{def-almost-lossless-block} 
An $(n,M,\epsilon)$ code is an $(M,\epsilon)$ code 
defined for a random vector $X^n$ 
with discrete vector alphabet $\mathcal{X}^n$. 
\end{defn}

\begin{defn}[Minimum achievable rate]
The minimum code size $M^*(n,\epsilon)$ and rate $R^*(n,\epsilon)$ 
achievable at blocklength $n$ and error probability $\epsilon$ 
are defined as 
\begin{IEEEeqnarray*}{rcl}
M^{*}(n,\epsilon) & \triangleq & \min \left\{M: \exists \, (n,M,\epsilon) \text{ code}\right\} \\
R^{*}(n,\epsilon) & \triangleq & \frac{1}{n}\log M^{*}(n,\epsilon).
\end{IEEEeqnarray*}
\end{defn}

A discrete information source 
is a sequence of discrete random variables, $X_1, X_2, \ldots$, 
specified by the transition probability kernels $P_{X_i | X^{i-1}}$, $i = 1, 2, \ldots$ 
While Definition~\ref{def-almost-lossless-block} applies to many classes of sources, 
including sources with memory and non-stationary sources, 
our asymptotic analysis focuses on stationary, memoryless sources, 
where $P_{X_i | X^{i-1}} = P_{X}$ 
for all $i = 1, 2, \ldots$ (i.e., $X_1, X_2, \ldots$ are i.i.d.). 

\subsection{Background}
\label{sec-b-almost-lossless}
Shannon's source coding theorem~\cite{shannon} 
describes the fundamental limit on the asymptotic performance 
for lossless source coding on a stationary, memoryless source, 
giving 
\begin{equation} 
\underset{n\rightarrow \infty}{\lim}R^{*}(n,\epsilon) = H(X),
\quad \forall\, \epsilon \in (0,1). 
\end{equation} 

In the finite-blocklength regime, 
Kontoyiannis and Verd\'{u} \cite{kontoyiannis-verdu} 
characterize $R^{*}(n,\epsilon)$ 
using upper and lower bounds
that match in their first three terms 
and show an $ O\left(\frac{1}{n}\right)$ fourth-order gap. 
\begin{thm}[Kontoyiannis and Verd\'{u} \cite{kontoyiannis-verdu}] \label{thm-k-v} 
Consider a stationary, memoryless source 
with finite alphabet $\mathcal{X}$, 
single-letter distribution $P_X$, 
and varentropy $V(X) > 0$. 
Then\footnote{These bounds, 
which are stated in a base-2 logarithmic scale in~\cite{kontoyiannis-verdu}, 
hold for any base. The base of the logarithm determines the information unit.} 
	(achievability) for all $0 < \epsilon \leq \frac{1}{2}$ 
	and all\footnote{According to \cite{kontoyiannis-verdu}, 
	the achievability bound holds for any $n \geq 1$. 
	Notice, however, that it only becomes meaningful 
	when $n > \left(\frac{T(X)}{V(X)^{3/2}\epsilon}\right)^2$.}  
	$n > \left(\frac{T(X)}{V(X)^{3/2}\epsilon}\right)^2$,	
\begin{IEEEeqnarray}{rCl}
	R^{*}(n,\epsilon) 
	&\leq& H(X) + \sqrt{\frac{V(X)}{n}}Q^{-1}(\epsilon) - \frac{\log_2 n}{2n} 
	\label{eq-achiev-k-v} \\ 
	&&+ \frac{1}{n}\log_2 \left(\frac{\log_2 e}{\sqrt{2\pi V(X)}} 
		+ \frac{T(X)}{V(X)^{3/2}} \right) \nonumber \\ 
	&&+ \frac{1}{n}\frac{T(X)}{V(X)\phi\left(\Phi^{-1}\left(\Phi(Q^{-1}(\epsilon))
		+\frac{T(X)}{V(X)^{3/2}\sqrt{n}}\right)\right)}; 
		\nonumber 
	\end{IEEEeqnarray}
	(converse) for all $0 < \epsilon \leq \frac{1}{2}$ and all 
	\begin{equation}
	n > \frac{1}{4}\left(1+\frac{T(X)}{2V(X)^{3/2}}\right)^{2}
	\frac{1}{\left(\phi(Q^{-1}(\epsilon))Q^{-1}(\epsilon)\right)^2},
	\end{equation} 
	\begin{IEEEeqnarray}{rCl}
	R^{*}(n,\epsilon) &\geq& H(X) 
	+ \sqrt{\frac{V(X)}{n}}Q^{-1}(\epsilon) - \frac{\log_2 n}{2n} \nonumber \\
	&& -\frac{1}{n} \frac{T(X)+2V(X)^{3/2}}{2V(X)\phi(Q^{-1}(\epsilon))}. 
	\label{eq-conv-k-v}
	\end{IEEEeqnarray}
\end{thm}
\begin{remark}
	Although \cite[Theorem~\ref{thm-k-v}]{kontoyiannis-verdu} 
	restricts attention to $0 < \epsilon \leq \frac{1}{2}$ and 
	$\mathcal{X}$ finite, 
	the proof in \cite{kontoyiannis-verdu} applies 
	for all $0 < \epsilon < 1$ and any countable source alphabet, 
	achieving the same first three terms 
	in~\eqref{eq-achiev-k-v} and~\eqref{eq-conv-k-v}
	and fourth-order term $\pm O\left(\frac{1}{n}\right)$ (which varies with $\epsilon$) 
	provided that the third centered moment $T(X)$ 
	of information random variable $X$ is finite.
\end{remark}
\begin{remark} \label{rem-zero-var}
	When $V(X) = 0$, 
	the source is uniformly distributed over a finite alphabet 
	(i.e., non-redundant), 
	and $H(X) = \log|\mathcal{X}|$.
	The optimal code maps any $1-\epsilon$ fraction 
	of possible source outcomes to unique codewords, 
	giving 
	\[
	1-\epsilon \leq \frac{M^*(n,\epsilon)}{|\mathcal{X}|^n} \leq 1 - \epsilon 
	+ \frac{1}{|\mathcal{X}|^n}. \label{eq-rem2-1}
	\]
	As a result, when $P_X$ is uniform, 
	\begin{align}
	&~ H(X) - \frac{1}{n}\log\frac{1}{1-\epsilon}\leq R^*(n,\epsilon) \nonumber \\
	\leq&~ H(X) - \frac{1}{n}\log\frac{1}{1-\epsilon} 
	+ \frac{\log e}{n(1-\epsilon)}\exp\left(-nH(X)\right), 
	\IEEEeqnarraynumspace \label{eq-rem-zero-var-1}
	\end{align} 
	which matches \eqref{eq-achiev-k-v} 
	up to the second order (since $V(X) = 0$) 
	but omits the $-\frac{\log n}{2n}$ third-order term.
\end{remark}
\begin{remark} \label{rem-zero-var-1}
While it is not captured by our notation, $R^*(n,\epsilon)$ is a function of $P_X$.  
Since the $-\frac{\log n}{2n}$ third-order term 
appears in~\eqref{eq-achiev-k-v} and~\eqref{eq-conv-k-v} 
but not in~\eqref{eq-rem-zero-var-1}, 
the {\em bound} on $R^{*}(n,\epsilon)$, when viewed as a function of $P_X$, 
is discontinuous at the point where 
$P_X$ equals the uniform distribution on $\mathcal{X}$.  
In contrast, $R^*(n,\epsilon)$, which is known and calculable, is continuous.
The problem arises because 
Berry-Esseen type bounds are loose for small $V(X)$. 
Thus for any finite $n$, the achievability bound in~\eqref{eq-achiev-k-v} 
blows up as $V(X) \rightarrow 0$.
See Figure~\ref{fig-p2p}. 
Theorem~\ref{thm-k-v} states that 
for any $V(X) > 0$ there exists some $n_0 = n_0(P_X, \epsilon)$ 
such that for all $n > n_0$, $R^{*}(n,\epsilon)$ 
behaves like $-\frac{\log n}{2n}$ in the third-order term; 
the smaller the value of $V(X)$, 
the larger $n_0$ must be. 
\end{remark}
\begin{figure}[!t]
\centering
\includegraphics[width=0.48\textwidth]{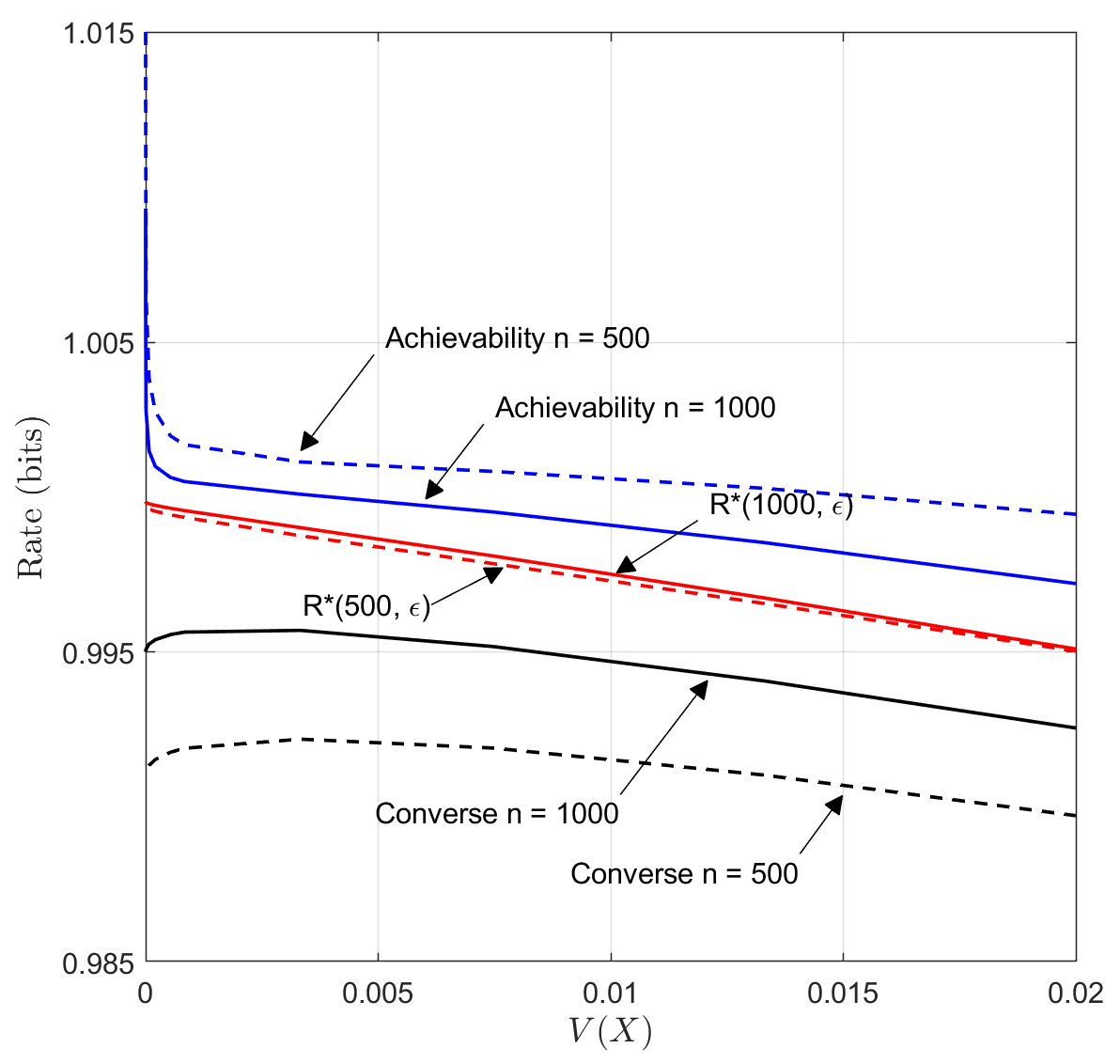}
\caption{Evaluations of the achievability bound in \eqref{eq-achiev-k-v}, 
the converse bound in \eqref{eq-conv-k-v}, 
and the optimum $R^{*}(n, \epsilon)$, 
all shown as a function of $V(X) = p(1-p)(\log \frac{1-p}{p})^2$ 
for a Bernoulli-$p$ source at $\epsilon = 0.1$.}
\label{fig-p2p}
\end{figure}
	
Achievability results 
based on Shannon's random coding argument~\cite{shannon} 
are important because they do not require knowledge of the optimal code, 
which is available only in a few special communication scenarios 
(e.g.,~\cite{kontoyiannis-verdu,kostina-pol-verdu}).
The following \emph{random coding} 
achievability bound\footnote{Tighter bounds 
based on the optimal code appear in~\cite[Lemma~1.3.1]{han} 
and~\cite[Remark~5]{kostina-verdu}.} 
is obtained by assigning source realizations to codewords 
independently and uniformly at random. 
The threshold decoder 
decodes to $x \in \mathcal{X}$ 
if and only if $x$ is a unique source realization that 
(i) is compatible with the observed codeword 
under the given (random) code design, and 
(ii) has information $\imath(x)$ 
below $\log M - \gamma$.

\begin{thm}[e.g. {\cite{verdu-notes}}, {\cite[Th.~9.4]{polyanskiy-notes}}] \label{thm-achiev-p} There exists an $(M,\epsilon)$ code for discrete random variable $X$ such that \begin{equation} \label{eq-achiev-p} \epsilon \leq \mathbb{P}\left[\imath(X) > \log M - \gamma\right] + \exp\left(-\gamma\right),\, \forall \, \gamma > 0. \end{equation}
\end{thm}

Particularizing \eqref{eq-achiev-p} to a stationary, memoryless source 
with single-letter distribution $P_X$ satisfying $V(X)>0$ and $T(X) < \infty$, 
choosing $\log M$ and $\gamma$ optimally, 
and applying the Berry-Esseen inequality 
(see Theorem~\ref{thm-berry-esseen} below)  
gives 
\begin{IEEEeqnarray}{rCl} \label{eq-achiev-p-1} 
R^{*}(n,\epsilon) 
\leq H(X) + \sqrt{\frac{V(X)}{n}}Q^{-1}(\epsilon) 
+ \frac{\log n}{2n} + O\left(\frac{1}{n}\right). \IEEEeqnarraynumspace 
\end{IEEEeqnarray} 

Since the optimal application of Theorem~\ref{thm-achiev-p} 
yields \eqref{eq-achiev-p-1}, 
which exceeds the bounds in Theorem~\ref{thm-k-v} 
by $+\frac{\log n}{n}$ in the third-order term, 
we are left to wonder whether 
random code design, threshold decoding, or both 
yield third-order performance penalties.
In \cite[Th.~8]{kontoyiannis-verdu}, 
Kontoyiannis and Verd\'{u} precisely characterize 
the performance of a code designed with 
i.i.d.\ uniform random codeword generation 
and an optimal (maximum likelihood) decoder.  
Unfortunately, that result is difficult to use in the asymptotic analysis. 
In Section~\ref{sec-result-almost-lossless} 
Theorem~\ref{thm-rcu-bound}, below, 
we derive a new random coding bound 
using a maximum likelihood decoder; 
this result demonstrates that random coding suffices 
to achieve the third-order optimal performance 
for a stationary, memoryless source.

\subsection{New Achievability Bounds Based on Random Coding}
\label{sec-result-almost-lossless}

We next use random code design to derive 
two new non-asymptotic achievability bounds 
for point-to-point source coding. 
We call these results 
the dependence testing (DT) bound and 
the random coding union (RCU) bound 
since they are the source coding analogues 
of the DT~\cite[Th.~17]{pol-poo-ver} 
and RCU~\cite[Th.~16]{pol-poo-ver} bounds 
in channel coding. 
The DT bound tightens Theorem~\ref{thm-achiev-p},  
which is also based on threshold decoding.  
\begin{thm}[DT bound] \label{thm-achiev-dt} 
Given a discrete random variable $X$, 
there exists an $(M,\epsilon)$ code with a threshold decoder 
for which 
\begin{equation} \label{eq-achiev-dt} 
\epsilon \leq 
\mathbb{E}\left[\exp \left\{-\left|\log M - \imath(X)\right]|_{+} \right\}\right]. \end{equation}	
\end{thm}

\begin{proof}
Appendix~\ref{append-thm-achiev-dt}.
\end{proof}

The proof of Theorem~\ref{thm-achiev-dt} 
bounds the random coding performance 
of a threshold decoder with threshold $\log \gamma$ as  
\begin{equation} 
\epsilon \leq \mathbb{P}\left[\imath(X) > \log \gamma \right] 
+ \frac{1}{M} \mathbb{U}\left[\imath(X) \leq \log \gamma\right], 
\label{eq-lem-dt-1-copy}
\end{equation} 
where $\mathbb{U}\left[\cdot\right]$ denotes a mass 
with respect to the counting measure $U_X$ on $\mathcal{X}$, 
which assigns unit weight to each $x \in \mathcal{X}$. 
As in a channel coding argument from~\cite{pol-poo-ver}, 
we apply the Neyman-Pearson lemma and 
find that the right-hand side of~\eqref{eq-lem-dt-1-copy} 
equals $\frac{M+1}{M}$ times the minimum measure of the error event 
in a Bayesian binary hypothesis test between $P_X$ 
with a priori probability $\frac{M}{M+1}$ and 
$U_X$ with a priori probability $\frac{1}{M+1}$. 
(The Neyman-Pearson lemma generalizes to $\sigma$-finite measures 
like $U_X$~\cite[Remark~5]{kostina-verdu}.) 
This error measure is minimized 
by the test that compares the log likelihood ratio 
$\log \frac{U_X(X)}{P_X(X)}$ 
to the log ratio of a priori probabilities $\log \frac{M/(M+1)}{1/(M+1)}$, giving
\begin{align*}
&H_{0}: P_X, \text{ selected if } \imath(X) \leq \log M\\
&H_{1}: U_X, \text{ selected if } \imath(X) > \log M.
\end{align*} 
Taking $\gamma = M$ minimizes the right-hand side of \eqref{eq-lem-dt-1-copy}, 
which implies that Theorem~\ref{thm-achiev-dt} is the tightest possible bound 
for random coding with threshold decoding. 

Particularizing Theorem~\ref{thm-achiev-dt} 
to a stationary, memoryless source with a single-letter distribution $P_X$ 
satisfying $V(X)>0$ and $T(X) < \infty$ 
and invoking the Berry-Esseen inequality 
(see Theorem~\ref{thm-berry-esseen} below), 
we obtain the asymptotic expansion
\begin{equation}
R^{*}(n,\epsilon) \leq H(X) + \sqrt{\frac{V(X)}{n}}Q^{-1}(\epsilon) + O\left(\frac{1}{n}\right). \label{eq-achiev-dt-7}
\end{equation} 
Unfortunately, \eqref{eq-achiev-dt-7} 
is sub-optimal in its third-order term. 
Thus, random code design with threshold-based decoding 
fails to achieve the optimal third-order performance. 

Next, we present the RCU bound, 
which employs random code design 
and maximum likelihood decoding.
\begin{thm}[RCU bound] \label{thm-rcu-bound} 
Given a discrete random variable $X$, there exists an $(M,\epsilon)$ code with a maximum likelihood decoder for which \begin{IEEEeqnarray}{rCl}\label{eq-rcu} 
\!\!\!\!\!\!\!\!\!		\epsilon \leq \mathbb{E} \left[\min \left\{1, \, \frac{1}{M}\mathbb{E}\left[\exp\left(\imath(\bar{X})\right)1\left\{\imath(\bar{X}) \leq \imath(X) \right\}|X\right] \right\} \right], 
	\end{IEEEeqnarray} 
	where $P_{X\bar{X}}(a,b) = P_{X}(a)P_{X}(b)$ for all $a,b \in \mathcal{X}$. 
\end{thm}

\begin{proof}
	Our random code design 
	randomly and independently draws 
	encoder output $\mathsf{F}(x)$ for each $x \in \mathcal{X}$ 
	from the uniform distribution on $[M]$. 
	We use the maximum likelihood decoder 
	\begin{equation} 
	\mathsf{g}(c) = \arg \underset{x \in \mathcal{X}:\,\mathsf{F}(x)=c}{\max} 
	P_X(x) = \arg \underset{x \in \mathcal{X}:\,\mathsf{F}(x)=c}{\min} \imath(x).
	\end{equation} 
	If multiple source symbols have the maximal probability mass, 
	the decoder design chooses among them uniformly at random. 
	
	Under this random code construction, 
	the expected error probability is bounded by the probability 
	$\mathbb{P}\left[\mathcal{E}\right]$ of event 
	\begin{equation}
	\mathcal{E}\triangleq \{\exists \, \bar{x} \in \mathcal{X}\backslash\{X\} 
	\text{ s.t. } \imath(\bar{x}) \leq \imath(X), \mathsf{F}(\bar{x}) = \mathsf{F}(X)\},
	\end{equation} 
	where probability measure $\mathbb{P}[\cdot]$ captures 
	both the random source output $X$ 
	and the random encoding map $\mathsf{F}$. 
	The resulting error bound is 
	\begin{IEEEeqnarray}{rCl}
	\lefteqn{\mathbb{E}\left[\left.\mathbb{P}\left(\left\{g(F(X))\neq X\right\}
		\right|F(\cdot)\right)\right]}\nonumber \\
	& \leq & \mathbb{P}\left[\mathcal{E}\right] \nonumber \\
	&=& \mathbb{E}\left[\mathbb{P}\left[\underset{\mathclap{\quad \;\; 
	\bar{x}\in \mathcal{X}\backslash\{X\}}}{\quad \bigcup\;} 
	\left\{\imath(\bar{x}) \leq \imath(X), 
	\mathsf{F}(\bar{x}) = \mathsf{F}(X)\right\}|X\right] \right] \label{eq-rcu-1} \\
	&\leq& \mathbb{E} \left[\min \left\{1,
	 \, \underset{\substack{{\bar{x}\in \mathcal{X}:} \\ {\bar{x} \neq X}}}{\sum} 
	\mathbb{P}\left[\{\imath(\bar{x}) \leq \imath(X),
	\mathsf{F}(\bar{x}) = \mathsf{F}(X)\}|X\right] \right\} \right] 
	\label{eq-rcu-2} \nonumber \\*\\
	&\leq& \mathbb{E} \left[\min \left\{1,
	 \,\frac{1}{M}\underset{\bar{x}\in \mathcal{X}}{\sum} 
	 1\{\imath(\bar{x}) \leq \imath(X)\}\right\}\right] 
	\label{eq-rcu-3} \\
	&=& \mathbb{E} \left[\min \left\{1,
	 \, \frac{1}{M} \mathbb{E}\left[\frac{1}{P_X(\bar{X})} 
	 1\{\imath(\bar{X}) \leq \imath(X)\}|X\right] \right\}\right] \label{eq-rcu-4},
	\end{IEEEeqnarray} 
	where \eqref{eq-rcu-1} applies the law of iterated expectation, 
	\eqref{eq-rcu-2} bounds the probability 
	by the minimum of the union bound and 1, 
	\eqref{eq-rcu-3} holds because the encoder outputs 
	are drawn i.i.d. uniformly at random and independently of $X$, 
	and \eqref{eq-rcu-4} rewrites \eqref{eq-rcu-3} 
	in terms of the distribution $P_{X\bar{X}} = P_{X}P_{X}$. 
	
	The existence of the desired $(M,\epsilon)$ code follows 
	since~\eqref{eq-rcu-4} equals the right-hand side of~\eqref{eq-rcu}.
\end{proof}
\begin{remark}
By the argument employed in the proof of \cite[Th.~9.5]{polyanskiy-notes}, 
we obtain the same RCU bound 
if we randomize only over linear encoding maps. 
Thus, there is no loss in performance when restricting to linear compressors.
\end{remark}

We next show that the RCU bound 
recovers the first three terms of the achievability result in Theorem~\ref{thm-k-v}. 
Thus, the sub-optimal third-order terms 
in~\eqref{eq-achiev-p-1} and~\eqref{eq-achiev-dt-7} 
result from the sub-optimal decoder 
rather than the random encoder design. 
This is important since optimal codes are not available 
for scenarios like the MASC studied in Section~\ref{sec-SW}, below.

Theorem~\ref{thm-rcu-asymp} 
focuses on a stationary, memoryless source with single-letter distribution $P_X$ 
satisfying 
\begin{IEEEeqnarray}{rCl}
V(X) &>& 0 \label{assump-a1} \\
T(X) &<& \infty. \label{assump-a2}
\end{IEEEeqnarray}
Define constants
\begin{align}
B &\triangleq C_0 \frac{T(X)}{V(X)^{3/2}} \\
C &\triangleq 2\left(\frac{\log 2}{\sqrt{2\pi V(X)}} + 2 B(X) \right) \label{eq-def-C},
\end{align}  
where $C_0$ is the absolute constant in the Berry-Esseen inequality 
for i.i.d. random variables. 
(See Theorem~\ref{thm-berry-esseen}, below.) 

\begin{thm}[Third-order-optimal achievability via random coding] \label{thm-rcu-asymp}
Consider a stationary, memoryless source 
satisfying the conditions in \eqref{assump-a1} and \eqref{assump-a2}. For all $0 < \epsilon < 1$,
\begin{equation}
R^{*}(n,\epsilon) \leq 
H(X) + \sqrt{\frac{V(X)}{n}}Q^{-1}(\epsilon) - \frac{\log n}{2n} + \xi(n), 
\label{eq-thm-rcu-asymp}
\end{equation}
where $\xi(n) = O\big(\frac{1}{n}\big)$ 
is bounded more precisely as follows.  
\begin{enumerate}[leftmargin=1.3\parindent]
\item[1)] For all $0 < \epsilon \leq \frac{1}{2}$ and $n > \left(\frac{B + C}{\epsilon}\right)^2$,
\begin{IEEEeqnarray}{rCl}
\!\!\!\!\!\!\!\!\!\!\!\!\!\! \xi(n)&\leq& \frac{1}{n} \log C\label{eq-thm-rcu-asymp-1} 
+\frac{1}{n} \frac{B + C}{\phi\left(\Phi^{-1}\left(\Phi(Q^{-1}(\epsilon))+\frac{B + C}{\sqrt{n}} \right)\right)}.
\end{IEEEeqnarray}
\item[2)] For all $\frac{1}{2} < \epsilon < 1$ and $n > \left(\frac{B + C}{\epsilon - \frac{1}{2}}\right)^2$,
\begin{IEEEeqnarray}{rCl}
	\xi(n) &\leq& \frac{1}{n} \log C 
	+ \frac{1}{n} \frac{B + C}{\phi\left(Q^{-1}(\epsilon)\right)}. \label{eq-thm-rcu-asymp-2}
\end{IEEEeqnarray} 
	\end{enumerate}
\end{thm}

Before we show our proof of the asymptotic expansion 
in Theorem~\ref{thm-rcu-asymp}, 
we state two auxiliary results used in our analysis. 
The first is the classical Berry-Esseen inequality 
(e.g.,~\cite[Chapter~XVI.5]{feller}), 
stated here with the best known absolute constant $C_0$ from \cite{shevtsova2013absolute}.

\begin{thm}[Berry-Esseen inequality] \label{thm-berry-esseen}
Let  $Z_1, \ldots, Z_n$ be independent random variables
such that $V \triangleq \frac 1 n \sum_{i = 1}^n \emph{Var}[Z_i] > 0$ 
and $T \triangleq \frac 1 n \sum_{i = 1}^n \mathbb{E}[|Z_i - \mathbb E[Z_i]|^3] < \infty$. 
Then for any real $t$ and $n \geq 1$,
\begin{align}
\left|\mathbb{P}\left[\frac 1 {\sqrt{nV}}\sum\limits_{i = 1}^{n}(Z_i - \mathbb E[Z_i]) \geq  t \right]-Q(t)\right| \leq \frac{C_0T}{V^{3/2}\sqrt{n}}, \label{eq-berry-esseen}
\end{align} where $0.4097 \leq C_0 \leq 0.5583$ ($0.4097 \leq C_0 < 0.4690$ 
for identically distributed $Z_i$) \cite{shevtsova2013absolute}. 
\end{thm}
We refer to $C_0 \cdot T/V^{3/2}$ as the Berry-Esseen constant.

The second result is from Polyanskiy et al.~\cite[Lemma~47]{pol-poo-ver}.
\begin{lem}[{\cite[Lemma~47]{pol-poo-ver}}] \label{lem-p}
	In the setting of Theorem~\ref{thm-berry-esseen}, it holds for any $A$ and $n \geq 1$ that
	\begin{IEEEeqnarray}{rCl} 
	\lefteqn{\mathbb{E}\left[\exp\left\{-\sum \limits_{i=1}^{n}Z_i\right\}
		1 \left\{\sum \limits_{i=1}^{n}Z_i \geq A\right\} \right]} \nonumber \\ 
	&\leq& 2\left(\frac{\log 2}{\sqrt{2\pi V}}
		+2C_0\frac{T}{V^{3/2}}\right)\frac{1}{\sqrt{n}}\exp\left(-A\right). 
	\label{eq-lem-p}
	\end{IEEEeqnarray}
\end{lem}

\begin{proof}[Proof of Theorem~\ref{thm-rcu-asymp}] 
We analyze the RCU bound of Theorem~\ref{thm-rcu-bound} 
for random variable $X^{n}$.  
For notational brevity, define
\begin{equation} I_n \triangleq \imath(X^n) = \sum_{i=1}^{n}\imath(X_i), \;\;
	\bar{I}_n \triangleq \imath(\bar{X}^n) = \sum_{i=1}^{n}\imath(\bar{X_i}).
\end{equation} 
By Theorem~\ref{thm-rcu-bound}, 
there exists an $(n,M,\epsilon')$ code such that 
\begin{equation} \label{eq-rcu-asymp} 
\epsilon' \leq \mathbb{E}\left[\min \left\{1, \, \frac{1}{M}\mathbb{E}
\left[\exp\left(\bar{I}_n\right)1\left\{\bar{I}_n \leq I_n\right\}|X^n\right] \right\}\right],
\end{equation} 
where $P_{X^n\bar{X}^n}= P_{X}^nP_{X}^n$, 
and each of $I_n$ and $\bar{I}_n$ is a sum of i.i.d. random variables. 
Applying Lemma~\ref{lem-p} 
with $Z_{i} = -\imath(\bar{X}_i)$ and $A = -I_n$ 
gives 
\begin{IEEEeqnarray}{rCl} 
\mathbb{E}\left[\exp\left(\bar{I}_n\right)
	1\left\{\bar{I}_n \leq I_n\right\}|X^n\right]
&\leq& 
\frac{C}{\sqrt{n}}\exp\left(I_n\right). \label{eq-rcu-asymp-1}
\end{IEEEeqnarray} 
 Plugging \eqref{eq-rcu-asymp-1} in~\eqref{eq-rcu-asymp}, we find 
\begin{IEEEeqnarray}{rCl}
\epsilon' 
&\leq& \mathbb{E}\left[\min \left\{1,
 \, \frac{C}{M\sqrt{n}}\exp\left(I_n\right) \right\}\right] \label{eq-rcu-asymp-2}\\
&=& \mathbb{P}\left[I_n > \log \frac{M\sqrt{n}}{C}\right] \notag \\
&& + \frac{C}{M\sqrt{n}} \mathbb{E}\left[\exp\left(I_n\right)
	1\left\{I_n \leq \log \frac{M\sqrt{n}}{C} \right\}\right] 
	\label{eq-rcu-asymp-3}\\
&\leq& \mathbb{P}\left[I_n > \log M + \frac{1}{2}\log n - \log C \right] 
	+ \frac{C}{\sqrt{n}}, \IEEEeqnarraynumspace \label{eq-rcu-asymp-4}
\end{IEEEeqnarray} 
where~\eqref{eq-rcu-asymp-2} 
plugs~\eqref{eq-rcu-asymp-1} into~\eqref{eq-rcu-asymp},
\eqref{eq-rcu-asymp-3} separates the cases 
$I_n > \log \left(M\sqrt{n}/C\right)$ 
and $I_n \leq \log \left(M\sqrt{n}/C\right)$, 
and \eqref{eq-rcu-asymp-4} applies Lemma~\ref{lem-p} 
to the second term in \eqref{eq-rcu-asymp-3}. 

Denote for brevity
\begin{align}
\delta_n \triangleq \frac{B+C}{\sqrt{n}}.
\end{align}
	
We now choose 
\begin{IEEEeqnarray}{rCl} 
\log M 
&=& nH(X) + \sqrt{nV(X)}Q^{-1}\left(\epsilon - \delta_n \right) - \frac{1}{2}\log n
\nonumber \\
&& 
 + \log C \label{eq-rcu-asymp-5}
\end{IEEEeqnarray} 
and apply the Berry-Esseen inequality (Theorem~\ref{thm-berry-esseen}) 
to~\eqref{eq-rcu-asymp-4}, 
giving $\epsilon' \leq \epsilon$ and proving achievability bound 
\begin{IEEEeqnarray}{rCl} R^{*}(n,\epsilon) \leq \frac{\log M}{n}.
\end{IEEEeqnarray}
To obtain \eqref{eq-thm-rcu-asymp} from \eqref{eq-rcu-asymp-5}, 
we note that as long as $\delta_n < \epsilon$,
\begin{IEEEeqnarray}{rCl}
Q^{-1}\left(\epsilon - \delta_n\right) 
&=& 
\Phi^{-1}\left(\Phi(Q^{-1}(\epsilon)) + \delta_n\right) \label{eq-rcu-asymp-11}\\
&=& Q^{-1}(\epsilon) + \delta_n (\Phi^{-1})'(\xi_n) 
	\label{eq-rcu-asymp-7}\\
&=& Q^{-1}(\epsilon) +  \frac{\delta_n}{\phi(\Phi^{-1}(\xi_n))}, 
\label{eq-rcu-asymp-8}
\end{IEEEeqnarray}
where~\eqref{eq-rcu-asymp-11} 
applies the definition of the Gaussian 
cumulative distribution function $\Phi(\cdot)$ 
and its complement $Q(\cdot)$ 
from~\eqref{def-gaussian-cdf} and~\eqref{def-qfunc}, 
\eqref{eq-rcu-asymp-7} holds by a first-order Taylor bound 
for some $\xi_n \in \left[\Phi(Q^{-1}(\epsilon)), 
\Phi(Q^{-1}(\epsilon))+\delta_n \right]$, 
and \eqref{eq-rcu-asymp-8} holds by the inverse function theorem.

	1) For $\epsilon \leq \frac{1}{2}$ and $\delta_n < \epsilon$, 
	$\xi_n \geq \frac{1}{2}$ and 
	$\phi(\Phi^{-1}(\xi_n))$ is decreasing in $\xi_n$. 
	We can further bound the right-hand side 
	of~\eqref{eq-rcu-asymp-8} and conclude that 
	\begin{IEEEeqnarray}{rCl}
	Q^{-1}\left(\epsilon - \delta_n \right) 
	&\leq&
	Q^{-1}(\epsilon) 
	+ \frac{\delta_n}{\phi\left(\Phi^{-1}\left(\Phi(Q^{-1}(\epsilon))
	+ \delta_n \right)\right)}. \IEEEeqnarraynumspace 
	\label{eq-rcu-asymp-9}
	\end{IEEEeqnarray}

	2) For $\epsilon > \frac{1}{2}$ and 
	$\delta_n < \epsilon-\frac{1}{2}$, 
	we have $\xi_n \leq \frac{1}{2}$ and 
	$\phi(\Phi^{-1}(\xi_n))$ is increasing in $\xi_n$. 
	We conclude that
	\begin{IEEEeqnarray}{rCl}
	Q^{-1}\left(\epsilon - \delta_n \right) \leq
	Q^{-1}(\epsilon) + \frac{\delta_n}{\phi(Q^{-1}(\epsilon))}. 
	\IEEEeqnarraynumspace \label{eq-rcu-asymp-10}
	\end{IEEEeqnarray}
	Plugging~\eqref{eq-rcu-asymp-9} and~\eqref{eq-rcu-asymp-10} 
	into~\eqref{eq-rcu-asymp-5} gives~\eqref{eq-thm-rcu-asymp-1} 
	and~\eqref{eq-thm-rcu-asymp-2}.
\end{proof}

\section{Composite Hypothesis Testing}
\label{sec-cht}

The meta-converse for channel coding~\cite[Th.~26]{pol-poo-ver}\footnote{The 
quantum information theory literature contains 
an earlier approach to channel coding converses 
using binary hypothesis testing~\cite{nagaoka2005strong},~\cite[Ch. 4.6]{hayashi2006quantum}.}
and its generalizations to lossy source coding~\cite{kostina-verdu} 
and joint source-channel coding~\cite{campo, kostina-jscc} 
apply binary hypothesis testing to derive converses 
in point-to-point communication problems.
To extend this approach to multi-terminal coding 
(see, e.g., Section~\ref{sec-SW} Theorem~\ref{thm-sw-cht-conv}, below),
we develop a corresponding method using composite hypothesis testing.
We first develop non-asymptotic tools 
and then analyze the asymptotics.  

A composite hypothesis test $P_{Z|X}: \mathcal{X} \rightarrow \{0, 1\}$ tests a simple hypothesis against a composite hypothesis:   
\begin{equation}
\begin{IEEEeqnarraybox}{l.l.l}
H_0: & X \sim P, & \text{ selected if } Z = 1 \\
H_1: & X \sim Q_j \text{ for some } j \in [k], & \text{ selected if } Z = 0, \nonumber
\end{IEEEeqnarraybox}
\end{equation}
where $X$ is the observation, 
$P$ is the distribution under the simple hypothesis, 
and $\{Q_j\}_{j=1}^{k}$ is the collection of possible distributions 
under the composite hypothesis. 
The following definition 
generalizes the optimal $\beta$-function 
from binary to composite hypothesis testing. 
(See, for example,~\cite[Def.~1]{huang-moulin}.)

\begin{defn} \label{def-beta-alpha}
The set of achievable false-positive errors 
for power-$\alpha$ tests 
between distribution $P$ 
and collection of distributions $\{Q_j\}_{j=1}^{k}$ 
is the subset of $[0,1]^k$ defined as
\begin{IEEEeqnarray}{rCl}
\lefteqn{\beta_\alpha \left(P, \{Q_j\}_{j=1}^{k}\right) 
\triangleq \left\{\boldsymbol{\beta} = (\beta_1,\ldots,\beta_k): 
\vphantom{\mathbb{Q}_j}\right.}  \nonumber \\
&&\;\; \left. \exists\, \text{test s.t. } 
\mathbb{P}\left[Z=1\right] \geq \alpha,
 \, \mathbb{Q}_j\left[Z=1\right] \leq \beta_j, \, \forall \, j \in [k] \right\}, 
 \IEEEeqnarraynumspace
\end{IEEEeqnarray} 
where $\mathbb{P}\left[\cdot\right]$ denotes a probability with respect to $P$, 
and for each $j \in [k]$, $\mathbb{Q}_j\left[\cdot\right]$ 
denotes a probability with respect to $Q_j$.
\end{defn}
Like binary hypothesis tests
(see \cite[Remark~5]{kostina-verdu}), 
composite hypothesis tests 
can be generalized to allow 
$P$ and $\{Q_j\}_{j=1}^{k}$ 
to be $\sigma$-finite measures; 
in such cases, $\beta_\alpha (P, \{Q_j\}_{j=1}^{k})$ 
may not be a subset of $[0,1]^k$. 
We apply this generalization in Section~\ref{sec-sw-ht-conv}
to derive our new MASC converse.

In \cite{huang-moulin}, 
Huang and Moulin study the asymptotics 
of the set $\beta_\alpha (P, \{Q_j\}_{j=1}^{k})$, 
giving a third-order-optimal characterization~\cite[Th.~1]{huang-moulin}. 
As noted in~\cite[Appendix~D]{recep-arxiv}, 
there is a gap in their converse proof 
(see also Remark~\ref{remark-hm}, below). 
We here present a comprehensive analysis 
of composite hypothesis testing, 
starting with non-asymptotic characterizations 
and then particularizing them to give a new proof of~\cite[Th.~1]{huang-moulin}.

\subsection{Non-Asymptotic Bounds}
The analysis of~$\beta_\alpha \left(P, \{Q_j\}_{j=1}^{k}\right)$  
in~\cite{huang-moulin} 
uses the test that achieves the minimal (boundary) points of that set. 
For each minimal point $\boldsymbol{\beta}$, 
there exists a vector 
$\mathbf{a} = (a_1,\ldots,a_k) \geq \mathbf{0}$, 
$\mathbf{a} \neq \mathbf{0}$, 
such that the generalized Neyman-Pearson test 
\begin{equation}
P_{Z|X}(1|x) =
\begin{cases}
1, &\text{ for } x \text{ s.t. } P(x) > \sum\limits_{j=1}^{k} a_jQ_j(x) \\
0, &\text{ for } x \text{ s.t. } P(x) < \sum\limits_{j=1}^{k} a_jQ_j(x) \\
\lambda, &\text{ for } x \text{ s.t. } P(x) = \sum\limits_{j=1}^{k} a_jQ_j(x),
\end{cases} 
\end{equation} 
achieves $\boldsymbol{\beta}$; 
here $\lambda \in [0,1]$ is chosen so that 
$\mathbb{P}\left[Z=1\right]=~\alpha$. 
While the above test is optimal, 
the achievability and converse bounds that follow 
simplify the asymptotic analysis.

\begin{lem}[Achievability] \label{lem-cht-achiev}
For any $\gamma_j \geq 0$, $j \in [k]$, 
there exists a composite hypothesis test $P_{Z|X}$ for which
\begin{IEEEeqnarray}{rCl}
\mathbb{P}\left[Z = 1 \right] 
&=& \mathbb{P}\left[\bigcap\limits_{j \in [k]} \left\{\frac{P(X)}{Q_j(X)} 
\geq \gamma_j \right\} \right] \label{eq-cht-achiv-1} \\
\mathbb{Q}_j\left[Z = 1 \right] 
&\leq& \mathbb{E}_{P}\left[\frac{Q_j(X)}{P(X)} 
1\left\{ \frac{P(X)}{Q_j(X)} \geq \gamma_j \right\} \right], \, \forall \, j \in [k]. 
\label{eq-cht-achiv-2} \nonumber \\*
\end{IEEEeqnarray}
\end{lem}
\begin{proof}
Fix any $\gamma_j \geq 0$, $j \in [k]$. 
Consider the (sub-optimal) likelihood-ratio 
threshold test\footnote{In \cite{huang-moulin}, 
Huang and Moulin also use this sub-optimal likelihood-ratio threshold 
in their asymptotic achievability analysis.}:
\begin{equation}
P_{Z|X}(1|x) =
\begin{cases}
1, &\text{ if } \frac{P(X)}{Q_j(X)} \geq \gamma_j, \, \forall \, j \in [k] \\
0, &\text{ otherwise.}
\end{cases}
\end{equation}
Under this test,~\eqref{eq-cht-achiv-1} 
follows immediately, and \eqref{eq-cht-achiv-2} holds by
\begin{IEEEeqnarray}{rCl}
\mathbb{Q}_j\left[Z = 1 \right] 
&=& \mathbb{Q}_j\left[\bigcap\limits_{j \in [k]} 
	\left\{\frac{P(X)}{Q_j(X)} \geq \gamma_j \right\} \right] \\
&\leq& \mathbb{Q}_j\left[\frac{P(X)}{Q_j(X)} \geq \gamma_j\right] \\
&=& \sum\limits_{x \in \mathcal{X}} P(x) \cdot 
	\frac{Q_j(x)}{P(x)} 1\left\{\frac{P(x)}{Q_j(x)} \geq \gamma_j \right\} 
	\IEEEeqnarraynumspace \\
&=& \mathbb{E}_{P}\left[\frac{Q_j(X)}{P(X)} 1\left\{ \frac{P(X)}{Q_j(X)} 
	\geq \gamma_j \right\} \right]. 
\end{IEEEeqnarray}
\end{proof}

The following converse bound extends~\cite[Eq.~(102)]{pol-poo-ver} 
from binary hypothesis testing to composite hypothesis testing. 
\begin{lem}[Converse] \label{lem-cht-conv}
	For any $\alpha$, 
	if $\boldsymbol{\beta} = (\beta_1,\ldots,\beta_k) 
	\in \beta_\alpha \left(P, \{Q_j\}_{j=1}^{k}\right)$, 
	then
	\begin{equation}
	\alpha - \sum\limits_{j=1}^{k}
	\gamma_j \beta_j \leq \mathbb{P}\left[\bigcap\limits_{j \in [k]} 
	\left\{\frac{P(X)}{Q_j(X)} \geq \gamma_j\right\}\right], \label{eq-sw-cht-4}
	\end{equation} where $\gamma_j \geq 0$, $j \in [k]$ 
	are arbitrary constants. 
\end{lem}
\begin{proof}
Appendix~\ref{append-lem-cht-conv}.
\end{proof}

Lemma~\ref{lem-cht-var} extends the argument 
of~\cite[Lemma~1]{elkayam-feder} 
from binary to composite hypothesis testing.

\begin{lem}[Variational lemma] \label{lem-cht-var}
	For any $\alpha$, 
	if $\boldsymbol{\beta} = (\beta_1,\ldots,\beta_k) 
	\in \beta_\alpha \left(P, \{Q_j\}_{j=1}^{k}\right)$, then 
	\begin{equation}
	\alpha - \sum\limits_{j=1}^{k}\gamma_j \beta_j 
	\leq 1 - \sum\limits_{x \in \mathcal{X}} 
	\min \left\{P(x),\, \sum\limits_{j=1}^{k}\gamma_j Q_j(x) \right\}, 
	\label{eq-sw-cht-7}
	\end{equation} 
	where $\gamma_j \geq 0$, $j \in [k]$, are arbitrary constants 
	and equality is achieved by a generalized Neyman-Pearson test.
\end{lem}
\begin{proof}
Appendix \ref{append-lem-cht-var}.
\end{proof}

Given any $\boldsymbol{\beta} = (\beta_1, \ldots, \beta_k)$, 
define
\begin{IEEEeqnarray}{rCl}
&&\epsilon^*(\boldsymbol{\beta}) \triangleq \inf 
	\left\{\epsilon \in [0, 1]: \exists \, \text{test s.t. } 
	\vphantom{\mathbb{Q}_j} \right.\nonumber \\
&&\qquad \left. \mathbb{P}\left[Z=1\right] \geq 1 - \epsilon, \, 
	\mathbb{Q}_j\left[Z=1\right] \leq \beta_j, \, \forall \, j \in [k] \right\}. 
	\IEEEeqnarraynumspace
\end{IEEEeqnarray}
Then Lemma~\ref{lem-cht-var} gives
\begin{IEEEeqnarray}{rCl}
\lefteqn{
\epsilon^*(\boldsymbol{\beta}) =}   \label{eq-sw-cht-8} \\ 
&& \sup\limits_{\gamma_1, \ldots, \gamma_k \geq 0} 
\left\{\sum\limits_{x \in \mathcal{X}} \min \left\{P(x), \, 
\sum\limits_{j=1}^{k}\gamma_j Q_j(x) \right\} - 
\sum\limits_{j=1}^{k}\gamma_j \beta_j \right\}. \nonumber 
\end{IEEEeqnarray}

\begin{remark}
We can derive Lemma~\ref{lem-cht-conv} 
from the variational characterization in Lemma~\ref{lem-cht-var} by noting that
\begin{IEEEeqnarray}{rCl}
\lefteqn{1 - \sum\limits_{x \in \mathcal{X}} 
	\min \left\{P(x), \, \sum\limits_{j=1}^{k}\gamma_j Q_j(x) \right\}} \nonumber \\
&\leq& 1 - \sum\limits_{x \in \mathcal{X}} P(x) 
	1\left\{P(x) < \sum\limits_{j=1}^{k}\gamma_j Q_j(x) \right\}
	 \IEEEeqnarraynumspace \\
&=& \mathbb{P}\left[P(X) \geq \sum\limits_{j=1}^{k}\gamma_j Q_j(X)\right] \\
&\leq& \mathbb{P}\left[\bigcap\limits_{j \in [k]} 
	\left\{\frac{P(X)}{Q_j(X)} \geq \gamma_j\right\}\right].
\end{IEEEeqnarray}
\end{remark}

Lemmas~\ref{lem-cht-conv} and~\ref{lem-cht-var} are useful 
beyond the asymptotic analysis of composite hypothesis testing. 
They also make it possible 
to recover previous converse bounds 
from our new MASC meta-converse, 
as presented in Section~\ref{sec-sw-ht-conv} below.

\subsection{Asymptotics for I.I.D. Distributions}
We here characterize the asymptotics 
of $\beta_\alpha \left(P, \{Q_j\}_{j=1}^{k}\right)$ 
when each of $P$ and $\{Q_j\}_{j=1}^{k}$ 
is a product of $n$ identical single-shot distributions, 
i.e., $P(X^n) = \prod_{i=1}^{n}P(X_i)$ 
and $Q_j(X^n) = \prod_{i=1}^{n}Q_j(X_i)$, $j \in [k]$.

We begin with notation.
For each $j \in [k]$, define 
\begin{IEEEeqnarray}{rCl}
D_j &\triangleq& \mathbb{E}_P \left[\log \frac{P(X)}{Q_j(X)} \right] \\
V_j &\triangleq& \text{Var}_P\left[\log \frac{P(X)}{Q_j(X)}\right] \\
T_j &\triangleq& \mathbb{E}_P\left[\left|\log \frac{P(X)}{Q_j(X)} - D_j\right|^3\right].
\end{IEEEeqnarray} 
Define vector $\mathbf{D}$ and matrix $\mathsf{V}$ as
\begin{IEEEeqnarray}{rCl}
\mathbf{D} & \triangleq & \left(D_j, \, j \in [k] \right) \label{eq-def-cht-D} \\
\mathsf{V} & \triangleq & 
\text{Cov}_P \left[\left(\log \frac{P(X)}{Q_j(X)}, \, j \in [k] \right) \right]. 
\label{eq-def-cht-v}
\end{IEEEeqnarray}

Let $\mathbf{Z}\in\mathbb{R}^d$ 
be a Gaussian random vector with mean zero 
and covariance matrix $\mathsf{V}$. 
Define the multidimensional counterpart 
of the function $Q^{-1}(\cdot)$ as
\begin{equation}
\mathscr{Q}_{\rm inv}(\mathsf{V},\epsilon) 
\triangleq \left\{\mathbf{z} \in 
\mathbb{R}^d: \mathbb{P}[\mathbf{Z} \leq \mathbf{z}] \geq 1-\epsilon \right\}. \label{eq-def-sve}
\end{equation} 
The set $\mathscr{Q}_{\rm inv}(\mathsf{V},\epsilon)$ 
appears in characterizations such as~\cite{tan-kosut, huang-moulin}. 
When $\mathsf{V}$ is non-singular, 
the boundary of $\mathscr{Q}_{\rm inv}(\mathsf{V},\epsilon)$ 
approaches $z_i = \sqrt{[\mathsf{V}]_{i,i}}Q^{-1}(\epsilon)$ 
in each dimension $i \in [d]$, 
as illustrated in Figure \ref{fig-sve}\protect\subref{fig-sve-1}. 
For $\epsilon \leq 1/2$, 
$\mathscr{Q}_{\rm inv}(\mathsf{V},\epsilon)$ 
lies in the positive orthant of $\mathbb{R}^d$; 
for $\epsilon > 1/2$, 
$\mathscr{Q}_{\rm inv}(\mathsf{V},\epsilon)$ extends 
outside of the positive orthant.
If $\epsilon' < \epsilon$, 
then $\mathscr{Q}_{\rm inv}(\mathsf{V},\epsilon')
\subset \mathscr{Q}_{\rm inv}(\mathsf{V},\epsilon)$.
See Figure \ref{fig-sve}\protect\subref{fig-sve-2} 
for plots of the boundaries of $\mathscr{Q}_{\rm inv}(\mathsf{V},\epsilon)$ 
in $\mathbb{R}^2$. 
If $\mathsf{V}$ is singular with rank $r < d$, 
then $\mathscr{Q}_{\rm inv}(\mathsf{V},\epsilon)$ 
lies in an $r$-dimensional subspace of $\mathbb{R}^d$.
\begin{figure}[!t]
	\centering
	\subfloat[]{\includegraphics[width=0.33\textwidth]{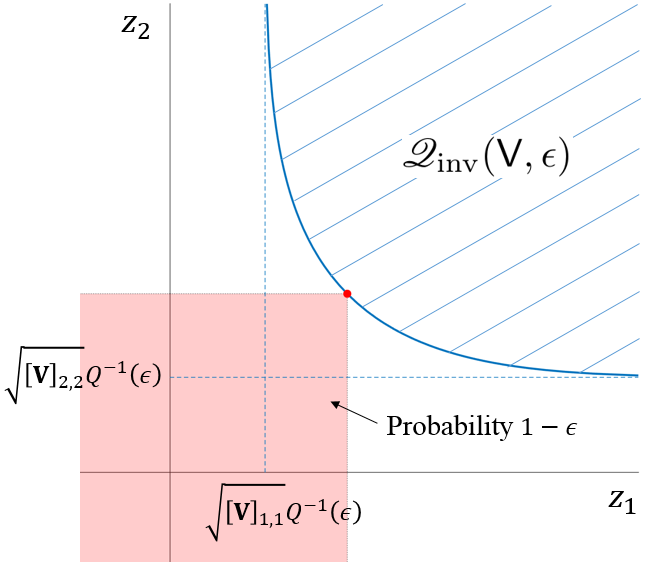}
		\label{fig-sve-1}}
	\hfil
	\subfloat[]{\includegraphics[width=0.3\textwidth]{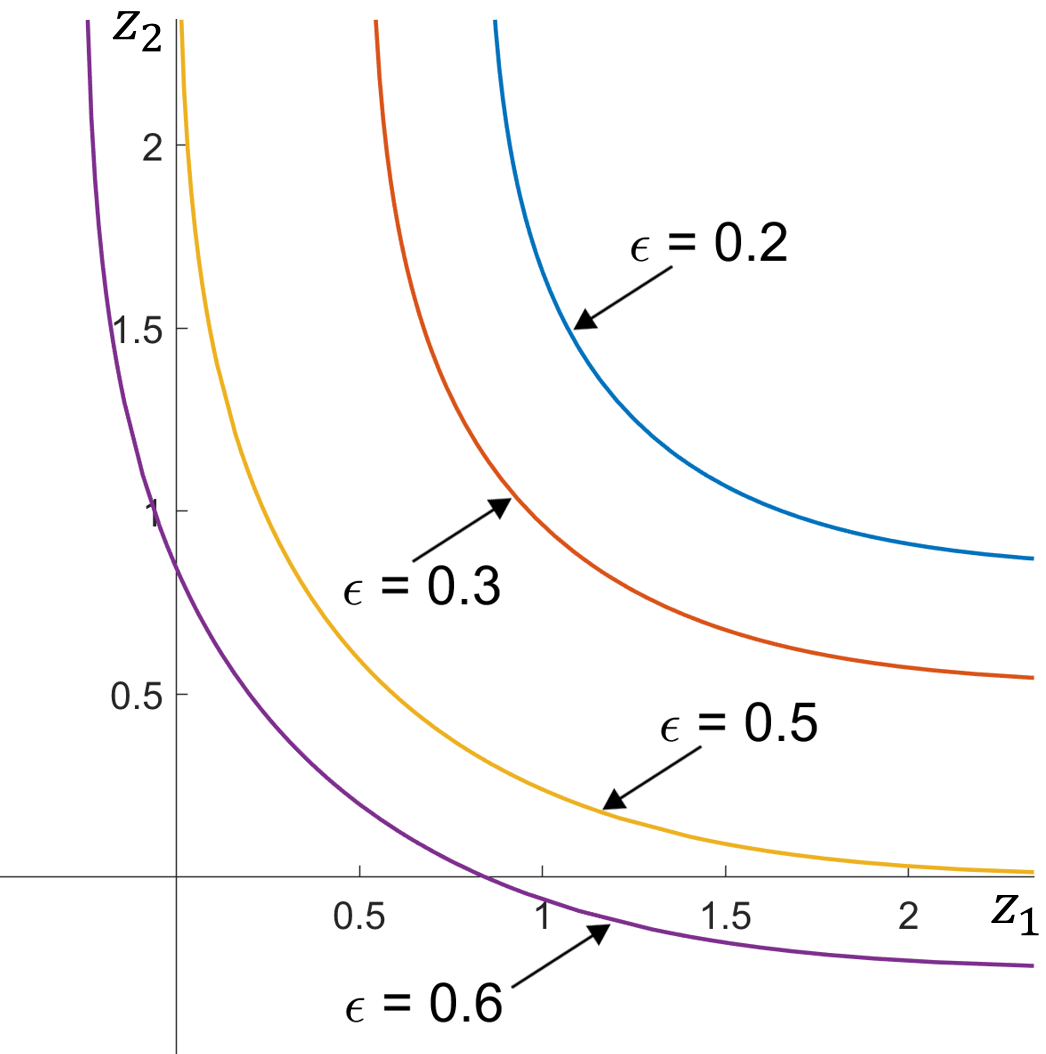}
		\label{fig-sve-2}}
	\caption{Illustrations of 
	$\mathscr{Q}_{\rm inv}(\mathsf{V},\epsilon)\subset\mathbb{R}^2$. 
	\protect\subref{fig-sve-1} A schematic drawing of 
	$\mathscr{Q}_{\rm inv}(\mathsf{V},\epsilon)$. 
	\protect\subref{fig-sve-2} A graph plotting the boundaries 
	of $\mathscr{Q}_{\rm inv}(\mathsf{V},\epsilon)\subset\mathbb{R}^2$ 
	with various values of $\epsilon$ when $\mathsf{V}$ is the identity matrix.}
	\label{fig-sve}
\end{figure}

Theorem~\ref{thm-cht-third-order} derives 
a third-order-optimal characterization 
of $\beta_\alpha \left(P, \{Q_j\}_{j=1}^{k}\right)$ 
under assumptions 
\begin{IEEEeqnarray}{rCl}
V_j &>& 0, \, \forall \, j \in [k] \label{assum-ht-1}\\
T_j &<& \infty, \, \forall \, j \in [k]. \label{assum-ht-2}
\end{IEEEeqnarray}
Define the inner and outer bounding sets 
\begin{IEEEeqnarray}{rCl}
\mathscr{B}^*_{\rm in}(n,\epsilon) 
&\triangleq& \exp \left\{ -n \mathbf{D} 
	+ \sqrt{n}\mathscr{Q}_{\rm inv}(\mathsf{V},\epsilon) 
	- \frac{\log n}{2} \mathbf{1} + O(1)\mathbf{1} \right\}  \nonumber \\
\mathscr{B}^*_{\rm out}(n,\epsilon)
&\triangleq&\exp \left\{ -n\mathbf{D} 
	+ \sqrt{n}\mathscr{Q}_{\rm inv}(\mathsf{V},\epsilon) 
	- \frac{\log n}{2} \mathbf{1} - O(1)\mathbf{1} \right\}, 
	\nonumber
\end{IEEEeqnarray} where vector $\mathbf{D}$ and matrix $\mathsf{V}$ are defined in \eqref{eq-def-cht-D} and \eqref{eq-def-cht-v}.

\begin{thm}[Third-order-optimal asymptotics] \label{thm-cht-third-order}
Assume that $P$ and $\left\{Q_j\right\}_{j=1}^{k}$ 
are product distributions composed of $n$ identical single-shot distributions 
that satisfy \eqref{assum-ht-1} and \eqref{assum-ht-2}. 
For any $\alpha\in(0,1)$, 
the set $\beta_\alpha \left(P, \{Q_j\}_{j=1}^{k}\right)$ satisfies
\begin{equation}
\mathscr{B}^*_{\rm in}(n,\epsilon) 
\subseteq \beta_\alpha \left(P, \{Q_j\}_{j=1}^{k}\right) 
\subseteq \mathscr{B}^*_{\rm out}(n,\epsilon),
\end{equation} where $\epsilon = 1 - \alpha$.
\end{thm}

\begin{remark} \label{remark-hm}
In \cite[Th.~1]{huang-moulin}, 
Huang and Moulin claim the third-order-optimal result 
in Theorem~\ref{thm-cht-third-order} when $\mathsf{V}$ is non-singular. 
Unfortunately, there is a gap in their converse proof. 
Applying~\cite[Lemma~2]{huang-moulin} to get~\cite[Eq.~(13)]{huang-moulin} 
requires that vector $\mathbf{b}$ is independent of $n$. 
However, they consider 
any $\mathbf{b} \in \mathscr{Q}_{\rm inv}\left(\mathsf{V}, \epsilon\right)$, 
which may grow with $n$ 
because set $\mathscr{Q}_{\rm inv}\left(\mathsf{V}, \epsilon\right)$ 
is unbounded. 
Thus, \cite[Eq.~(13)]{huang-moulin} does not always hold.
\end{remark}

We resolve this issue with a new proof of Theorem~\ref{thm-cht-third-order}
that leverages Lemmas~\ref{lem-cht-achiev} and \ref{lem-cht-conv}. 
We first show two auxiliary results. 

The multidimensional Berry-Esseen theorem 
bounds the probability of a sum of i.i.d. random vectors. 
Bentkus' theorem~\cite[Th.~1.1]{bentkus} 
for the case with mean zero and identity covariance 
achieves the best known dependence on dimension. 
Tan and Kosut extend~\cite[Th.~1.1]{bentkus} 
to non-singular covariance matrices~\cite[Cor.~8]{tan-kosut}.
We here extend~\cite[Cor.~8]{tan-kosut} 
to covariance matrices with non-zero rank. 

\begin{lem} \label{lem-b-e}
	Let $\mathbf{U}_1,\ldots,\mathbf{U}_n$ be i.i.d. random vectors 
	in $\mathbb{R}^d$ with mean zero and covariance matrix $\mathsf{V}$. 
	Let $\mathbf{Z} \sim \mathcal{N}(\mathbf{0},\mathsf{V})$ 
	be a Gaussian vector in $\mathbb{R}^d$. 
	Define $r \triangleq \emph{rank}(\mathsf{V})$. 
	Let $\mathsf{T}$ be a $d\times r$ matrix 
	whose columns are the $r$ normalized eigenvectors of $\mathsf{V}$ 
	with non-zero eigenvalues. 
	Define i.i.d. random vectors 
	$\mathbf{W}_1,\ldots,\mathbf{W}_n \in \mathbb{R}^r$ 
	such that $\mathbf{U}_i = \mathsf{T}\mathbf{W}_i$ for $i\in[n]$. 
	Let $\mathsf{V}_r \triangleq \emph{Cov}[\mathbf{W}_1]$ 
	and $\beta_r \triangleq \mathbb{E}[\|\mathbf{W}_1\|_2^3]$. 
	If $r \geq 1$, then for all $n$,
	\begin{equation}
	\underset{\mathbf{z} \in \mathbb{R}^d}{\sup} 
	\left|\mathbb{P}\left[\frac{1}{\sqrt{n}}
	\sum\limits_{i=1}^{n}\mathbf{U}_i \leq \mathbf{z}\right] - 
	\mathbb{P}[\mathbf{Z} \leq \mathbf{z}]\right| 
	\leq \frac{400 d^{1/4}\beta_r}{\lambda_{\min}(\mathsf{V}_r)^{3/2}\sqrt{n}}, 
	\label{eq-lem-b-e}
	\end{equation} 
	where $\lambda_{\min}(\mathsf{V}_r) > 0$ 
	is the smallest eigenvalue of matrix $\mathsf{V}_r$.
\end{lem}
\begin{proof}
	Appendix~\ref{append-lem-b-e}.
\end{proof}
If $r = d$, then $\mathsf{V}_r = \mathsf{V}$ 
and Lemma~\ref{lem-b-e} recovers \cite[Cor.~8]{tan-kosut}.

The following lemma is useful for our asymptotic analysis.
\begin{lem} \label{lem-sve} 
Fix an arbitrary $d \times d$ positive-semidefinite matrix~$\mathsf{V}$ 
and $0 < \epsilon < 1$. Then, the following results hold.  
\begin{enumerate}[leftmargin=1.3\parindent]
\item There exist constants $D_1$ and $\delta_1 > 0$ 
such that for all $0 \leq \delta < \delta_1$,
\begin{equation}
\mathscr{Q}_{\rm inv}(\mathsf{V},\epsilon) \subseteq 
\mathscr{Q}_{\rm inv}\left(\mathsf{V},\epsilon - \delta \right)
 - D_1\delta\mathbf{1}. \label{eq-lem-sve-1}
\end{equation}
\label{part1}
\item There exist constants $D_2$ and $\delta_2 > 0$ 
such that for all $0 \leq \delta < \delta_2$,
\begin{equation}
\mathscr{Q}_{\rm inv}(\mathsf{V},\epsilon) 
\supseteq \mathscr{Q}_{\rm inv}\left(\mathsf{V},\epsilon + \delta \right) 
+ D_2\delta\mathbf{1}. \label{eq-lem-sve-2}
\end{equation}
\label{part2}
\end{enumerate}
\end{lem}
\begin{proof}
	Appendix \ref{append-sve}.
\end{proof}

\begin{proof}[Proof of Theorem~\ref{thm-cht-third-order}]
Define random variables
\begin{equation}
I_{j} \triangleq 
\sum\limits_{i=1}^{n} \log \frac{P(X_i)}{Q_j(X_i)}, \, j \in [k]
\end{equation} and random vector
\begin{equation}
\mathbf{I} \triangleq \left(I_{j}, \, j \in [k] \right).
\end{equation} For brevity, denote 
\begin{equation}
\boldsymbol{\gamma} \triangleq \left(\gamma_j ,\, j \in [k]\right).
\end{equation}

To prove the achievability part of Theorem~\ref{thm-cht-third-order}, 
we particularize Lemma~\ref{lem-cht-achiev} 
to product distributions $P^{\otimes n}$ and $\left\{Q_j^{\otimes n}\right\}_{j=1}^k$  
to obtain that for any $\boldsymbol{\gamma} \geq \mathbf{0}$, 
there exists a test $P_{Z|X^n}$ for which 
\begin{IEEEeqnarray}{rCl}
\!\!\!\!\!\!\!\!\!\! \mathbb{P} \left[Z=1 \right] 
&=& \mathbb{P} \left[\mathbf{I} \geq \log \boldsymbol{\gamma} \right] 
\label{eq-cht-asymp-1} \\
\!\!\!\!\!\!\!\!\!\! \mathbb{Q}_j \left[Z=1 \right] 
&\leq& \mathbb{E}_P \left[\exp\left(-I_{j}\right)
1\left\{I_{j} \geq \log \gamma_j \right\} \right], \, \forall \, j \in [k]. 
\label{eq-cht-asymp-2} 
\end{IEEEeqnarray} 
Take any $\boldsymbol{\gamma}$ such that 
\begin{equation}
\log \boldsymbol{\gamma} \in n\mathbf{D} 
- \sqrt{n}\mathscr{Q}_{\rm inv}\left(\mathsf{V}, \epsilon - \frac{B}{\sqrt{n}}\right), \label{eq:gammachoice}
\end{equation} 
where $B$ is the constant 
on the right side of \eqref{eq-lem-b-e} for $\mathbf{I}_{n}$, 
which is finite under assumptions~\eqref{assum-ht-1} and~\eqref{assum-ht-2}. Applying Lemma~\ref{lem-b-e} to \eqref{eq-cht-asymp-1} gives 
\begin{IEEEeqnarray}{rCl}
\mathbb{P} \left[Z=1 \right] 
&=& \mathbb{P} \left[\frac{1}{\sqrt{n}}\left(-\mathbf{I}
+ n\mathbf{D}\right) \leq \frac{1}{\sqrt{n}}\left(-\log \boldsymbol{\gamma} 
+ n\mathbf{D}\right) \right] \IEEEeqnarraynumspace \\
&\geq& \mathbb{P} \left[\mathbf{Z} 
\leq \frac{1}{\sqrt{n}}\left(-\log \boldsymbol{\gamma} 
+ n\mathbf{D}\right) \right] - \frac{B}{\sqrt{n}} \\
&\geq& 1 - \epsilon,
\label{eq-cht-asymp-3}
\end{IEEEeqnarray} 
where $\mathbf{Z} \sim \mathcal{N}(\mathbf{0}, \mathsf{V})$ and matrix $\mathsf{V}$ is defined in \eqref{eq-def-cht-v}. 
Applying Lemma~\ref{lem-p} to \eqref{eq-cht-asymp-2} gives 
\begin{IEEEeqnarray}{rCl}
\mathbb{Q}_j \left[Z=1 \right] &\leq& \frac{K_j}{\sqrt{n}}\exp\left(
-\log \gamma_j \right), 
\label{eq-cht-asymp-5} \IEEEeqnarraynumspace
\end{IEEEeqnarray} 
where 
\begin{equation}
K_j 
\triangleq 2\left(\frac{\log 2}{\sqrt{2\pi V_j}} + 2C_0\frac{T_j}{V_j^{3/2}} \right)
\end{equation} 
is a finite positive constant 
by the assumptions in \eqref{assum-ht-1} and \eqref{assum-ht-2}. 
Plugging~\eqref{eq:gammachoice} into~\eqref{eq-cht-asymp-5} 
and noting~\eqref{eq-cht-asymp-3} gives 
\begin{IEEEeqnarray}{rCl}
\lefteqn{\beta_\alpha \left(P, \{Q_j\}_{j=1}^{k}\right)}  \\
&\supseteq& \exp \left\{ -n \mathbf{D} 
+ \sqrt{n}\mathscr{Q}_{\rm inv}\left(\mathsf{V},\epsilon - \frac{C}{\sqrt{n}}\right) 
- \frac{\log n}{2} \mathbf{1} + O(1)\mathbf{1} \right\} \nonumber \\
&\supseteq& \exp \left\{ -n \mathbf{D} 
+ \sqrt{n}\mathscr{Q}_{\rm inv}\left(\mathsf{V},\epsilon\right) 
- \frac{\log n}{2} \mathbf{1} + O(1)\mathbf{1} \right\}, \label{eq-cht-asymp-6}
\end{IEEEeqnarray} 
where \eqref{eq-cht-asymp-6} follows from Lemma~\ref{lem-sve}-\ref{part1}.

For the converse, 
recall from Lemma~\ref{lem-cht-conv} that if $\epsilon = 1 - \alpha$, 
then any $\boldsymbol{\beta}\in \beta_\alpha \left(P, \{Q_j\}_{j=1}^{k}\right)$ 
must satisfy 
\begin{equation}
\epsilon \geq 1 - \mathbb{P} \left[\bigcap\limits_{j=1}^k \left\{-I_{j} \leq \log \frac{1}{\gamma_j}\right\} \right] 
- \sum\limits_{j=1}^k \gamma_j \beta_j \label{eq-cht-asymp-7}
\end{equation} 
for all $\gamma_j\geq 0$,  $j \in [k]$.
Take 
\begin{equation}
\gamma_j = \frac{1}{\beta_j \sqrt{n}}, \, j \in [k].
\end{equation} Then, \eqref{eq-cht-asymp-7} becomes
\begin{IEEEeqnarray}{rCl}
\!\!\!\!\!\!\!\! \epsilon &\geq& 1 - \mathbb{P} \left[-\boldsymbol{I} \leq \log \boldsymbol{\beta} + \frac{\log n}{2} \mathbf{1} \right] - \frac{k}{\sqrt{n}} \\
\!\!\!\!\!\!\!\! &\geq& 1 - \mathbb{P} \left[\mathbf{Z} \leq \frac{1}{\sqrt{n}} \left(\log \boldsymbol{\beta} + n \mathbf{D} + \frac{\log n}{2} \mathbf{1} \right)\right] - \frac{B + k}{\sqrt{n}}, \label{eq-cht-asymp-8} 
\end{IEEEeqnarray} 
where~\eqref{eq-cht-asymp-8} applies Lemma~\ref{lem-b-e} and 
$B$ is the constant in the right side of \eqref{eq-lem-b-e}. 
By the definition of $\mathscr{Q}_{\rm inv}\left(\mathsf{V}, \epsilon\right)$ in \eqref{eq-def-sve}, \eqref{eq-cht-asymp-8} implies that
\begin{equation}
\! \boldsymbol{\beta} \in \exp \left\{ -n \mathbf{D} + \sqrt{n}\mathscr{Q}_{\rm inv}\left(\mathsf{V},\epsilon + \frac{B + k}{\sqrt{n}}\right) - \frac{\log n}{2} \mathbf{1} \right\}. \label{eq-cht-asymp-10}
\end{equation} Applying Lemma~\ref{lem-sve}-\ref{part2}, we conclude from \eqref{eq-cht-asymp-10} that 
\begin{IEEEeqnarray}{rCl}
\IEEEeqnarraymulticol{3}{l}{\beta_\alpha \left(P, \{Q_j\}_{j=1}^{k}\right)} \nonumber \\
&\subseteq& \exp \left\{ -n \mathbf{D} + \sqrt{n}\mathscr{Q}_{\rm inv}\left(\mathsf{V},\epsilon\right) - \frac{\log n}{2} \mathbf{1} - O(1)\mathbf{1} \right\}. \IEEEeqnarraynumspace
\end{IEEEeqnarray}
\end{proof}

\section{Multiple Access Source Coding}
\label{sec-SW}
To simplify notation, 
we focus on MASCs with two encoders. 
Our definitions and results generalize 
to more than two encoders, 
as briefly noted in Remark~\ref{rem-sw-general} below.

\subsection{Definitions} \label{sec-def-sw}
In a MASC~\cite{slepian-wolf}, 
also known as a Slepian-Wolf source code, 
independent encoders compress 
a pair of random variables $(X_{1}, X_{2})$ 
with discrete alphabets $\mathcal{X}_1$ and $\mathcal{X}_2$. 
Encoder $i$, $i\in[2]$, observes only $X_{i}$, 
which it maps to a codeword in $[M_i]$; 
a single decoder jointly decodes the pair of codewords 
to reconstruct $(X_{1}, X_{2})$. 
We first define codes for abstract random objects 
and then particularize to 
random objects that live in an alphabet 
endowed with a Cartesian product structure.

\begin{defn}[MASC] \label{def-sw} 
An $(M_1,M_2,\epsilon)$ MASC 
for random variables $(X_1,X_2)$ 
with discrete alphabets $\mathcal{X}_1$ and $\mathcal{X}_2$ 
comprises two encoding functions 
$\mathsf{f}_1\colon \mathcal{X}_1 \rightarrow [M_1]$ and 
$\mathsf{f}_2\colon \mathcal{X}_2 \rightarrow [M_2]$ 
and a decoding function, 
$\mathsf{g}\colon [M_1] \times [M_2] \rightarrow \mathcal{X}_1 \times \mathcal{X}_2$ 
with error probability 
$\mathbb{P}\left[\mathsf{g}(\mathsf{f}_1(X_1),\mathsf{f}_2(X_2)) \neq (X_1,X_2)\right] 
\leq \epsilon$.
\end{defn}

In block coding, 
encoders individually observe $X_{1}^n$ and $X_{2}^n$ 
drawn from distribution $P_{X_{1}^n X_{2}^n}$ 
on $\mathcal{X}_1^n \times \mathcal{X}_2^n$. 
Our block MASC definition
is similar to those in~\cite{tan-kosut} and~\cite{nomura-han}.

\begin{defn}[Block MASC] \label{def-sw-block} 
An $(n,M_1,M_2,\epsilon)$ MASC 
is an $(M_1,M_2,\epsilon)$ MASC 
for random vectors $(X_1^n, X_2^n)$ 
on $\mathcal{X}_1^n \times \mathcal{X}_2^n$. 
The code rate $\mathbf{R}=(R_1,R_2)$ is given by 
\begin{equation}
R_1 \triangleq \frac{1}{n}\log M_1, \, R_2 \triangleq \frac{1}{n}\log M_2.
\end{equation} 
\end{defn}

\begin{defn}[$(n,\epsilon)$-rate region] \label{def-sw-rate-region} 
Rate $\mathbf{R} = (R_1,R_2)$ is $(n,\epsilon)$-achievable 
if there exists an $(n,M_1,M_2,\epsilon)$ MASC 
with $R_1 \leq \frac{1}{n}\log M_1$ and $R_2 \leq \frac{1}{n}\log M_2$. 
The $(n,\epsilon)$-rate region $\mathscr{R}^{*}(n,\epsilon)$ 
is the closure of the set of $(n,\epsilon)$-achievable rate pairs. 
\end{defn}

While definitions~\ref{def-sw-block} and \ref{def-sw-rate-region} 
apply to arbitrary discrete random variables 
$(X_{1i},X_{2i})$, $i = 1, 2, \ldots$, 
with transition probability kernels 
$P_{(X_{1} X_{2})_{i}|(X_{1}X_{2})^{i-1}}$, 
our asymptotic analysis focuses on stationary, memoryless sources, 
where $P_{(X_{1} X_{2})_{i}|(X_{1}X_{2})^{i-1}} = P_{X_1 X_2}$ 
for all $i = 1, 2, \ldots$ 

For any rate $\mathbf{R} = (R_1, R_2)$ 
and distribution $P_{X_1X_2}$, define 
\begin{equation}
\overline{\mathbf{R}} \triangleq
\begin{bmatrix}
R_1 \\
R_2 \\
R_1 + R_2
\end{bmatrix}, \,
\overline{\mathbf{H}} \triangleq
\begin{bmatrix}
H(X_{1}|X_2) \\
H(X_{2}|X_1) \\
H(X_{1},X_2)
\end{bmatrix}. \label{def-R-H}
\end{equation}

\subsection{Background}
\label{sec-b-SW}
In \cite{slepian-wolf}, Slepian and Wolf prove that 
if $(X_{1}^{n},X_{2}^{n})$ are stationary and memoryless, 
then for every $\epsilon\in(0,1)$, 
\begin{IEEEeqnarray}{rCl}
\lim_{n \rightarrow \infty} \mathscr{R}^{*}(n,\epsilon) 
= \{(R_1,R_2): R_1 &\geq& H(X_1|X_2)\nonumber \\
R_2 &\geq& H(X_2|X_1) \nonumber \\
R_1 + R_2 &\geq& H(X_1,X_2)\}, \IEEEeqnarraynumspace \label{eq-asymp-sw}
\end{IEEEeqnarray} 
(i.e., the strong converse holds). 
We call this region the asymptotic MASC rate region. 

In~\cite{miyake-kanaya}, Miyake and Kanaya 
give achievability and converse bounds 
for finite-blocklength coding on finite-alphabet sources.  
In~\cite{han}, Han gives corresponding results 
for sources with countable alphabets. 
While these results are stated in~\cite{han} for general sources 
whose alphabets adopt $n$-fold Cartesian product structures, 
we here describe them in an abstract form.
 
\begin{thm}[Achievability, Han {\cite[Lemma~7.2.1]{han}}]
\label{lem-han-achiev}
Given discrete random variables $(X_1, X_2)$, 
there exists an $(M_1,M_2,\epsilon)$ MASC satisfying
\begin{IEEEeqnarray}{rCl}
\epsilon &\leq& \mathbb{P} \left[\left\{\imath(X_1|X_2) 
\geq \log M_1 - \gamma \right\} \right. \nonumber \\
&& \left. \cup  \left\{\imath(X_2|X_1) \geq \log M_2 - \gamma\right\} \right. 
\nonumber \\ 
&&\left. \cup \left\{\imath(X_1,X_2)\geq \log M_1M_2 - \gamma \right\} \right] 
+ 3\exp\left(-\gamma\right), \IEEEeqnarraynumspace
\end{IEEEeqnarray} where $\gamma >0$ is an arbitrary constant.
\end{thm} 

\begin{thm}[Converse, Han {\cite[Lemma~7.2.2]{han}}] \label{lem-han-conv}
Any $(M_1,M_2,\epsilon)$ MASC on discrete random variables $(X_1, X_2)$ satisfies 
\begin{IEEEeqnarray}{rCl}
\epsilon &\geq& \mathbb{P} \left[\left\{\imath(X_1|X_2) \geq \log M_1 
+ \gamma \right\} \right. \nonumber \\
&& \left. \cup  \left\{\imath(X_2|X_1) \geq \log M_2 
+ \gamma\right\} \right. \nonumber \\ 
&&\left. \cup \left\{\imath(X_1,X_2)\geq \log M_1M_2 
+\gamma \right\}\right] - 3\exp\left(-\gamma\right), \label{eq-lem-h} 
\IEEEeqnarraynumspace
\end{IEEEeqnarray} where $\gamma >0$ is an arbitrary constant.
\end{thm}

In \cite{jose-k}, Jose and Kulkarni 
derive a new linear programming (LP) finite-blocklength converse, 
tightening the bound in Theorem~\ref{lem-han-conv} 
with an extra non-negative term (see \cite[Cor.~13]{jose-k}).

\begin{thm}[LP-based converse, {\cite[Th.~12]{jose-k}}] \label{thm-lp-conv}
Any $(M_1,M_2,\epsilon)$ MASC 
on discrete random variables $(X_1, X_2)$ satisfies
\begin{IEEEeqnarray}{rCl}
\epsilon 
&\geq& \sup\limits_{\phi_1, \phi_2, \phi_3} 
\left\{\sum\limits_{\substack{{x_1 \in \mathcal{X}_1} \\ {x_2 \in \mathcal{X}_2}}} 
\min \left\{P_{X_1X_2}(x_1, x_2), \, \sum\limits_{j=1}^{3}\phi_j(x_1,x_2) 
\right\} \right. \nonumber \\
&& - M_1 \sum\limits_{x_2 \in \mathcal{X}_2} 
\max\limits_{\hat{x}_1 \in \mathcal{X}_1}\phi_1(\hat{x}_1,x_2) 
- M_2 \sum\limits_{x_1 \in \mathcal{X}_1} 
\max\limits_{\hat{x}_2 \in \mathcal{X}_2}\phi_2(x_1,\hat{x}_2) \nonumber \\
&& \left. \vphantom{\sum\limits_{x_1}} - M_1M_2 
\max\limits_{\hat{x}_1 \in \mathcal{X}_1, \, \hat{x}_2 \in \mathcal{X}_2} 
\phi_3(\hat{x}_1,\hat{x}_2) \vphantom{
\sum\limits_{\substack{{x_1 \in \mathcal{X}_1} \\ {x_2 \in \mathcal{X}_2}}}}\right\}, 
\label{eq-lem-jk} \IEEEeqnarraynumspace
\end{IEEEeqnarray} 
where the supremum is over 
$\phi_1, \phi_2, \phi_3: \mathcal{X}_1 \times \mathcal{X}_2 \rightarrow [0, 1]$ 
such that $0 \leq \phi_1(x_1,x_2), \, \phi_2(x_1,x_2), \, \phi_3(x_1,x_2) 
\leq P_{X_1X_2}(x_1,x_2)$ 
for all $(x_1,x_2) \in \mathcal{X}_1 \times \mathcal{X}_2$.
\end{thm}

The best prior asymptotic expansion of the MASC rate region 
is the second-order characterization 
developed independently in~\cite{tan-kosut, nomura-han}. 
In~\cite{tan-kosut}, Tan and Kosut introduce 
an \emph{entropy dispersion matrix}, 
which serves a role similar to the scalar dispersion 
in the point-to-point case~\cite{pol-poo-ver,kontoyiannis-verdu,kostina-verdu}. 
\begin{defn}[Tan and Kosut {\cite[Def.~7]{tan-kosut}}] \label{def-dis-matrix}
	The entropy dispersion matrix $\mathsf{V}$ 
	for random variables $(X_1,X_2)$ 
	is the covariance matrix 
	$\mathsf{V} \triangleq 
	\emph{Cov}\left[\overline{\boldsymbol{\imath}}(X_1,X_2)\right]$ 
	of random vector 
	\begin{equation}
	\overline{\boldsymbol{\imath}}(X_1,X_2) \triangleq
	\begin{bmatrix}
	\imath(X_{1}|X_2) \\
	\imath(X_{2}|X_1) \\
	\imath(X_{1},X_2)
	\end{bmatrix}.
	\end{equation}
\end{defn}
Note that $\mathsf{V}$ is a $3 \times 3$ positive-semidefinite matrix 
with $V(X_{1}|X_2)$, $V(X_{2}|X_1)$, and $V(X_{1},X_2)$ on the diagonal. 

Tan and Kosut \cite{tan-kosut} give a second-order characterization 
of the MASC rate region for finite-alphabet stationary, memoryless sources 
in terms of the asymptotic rate region 
and the entropy dispersion matrix. 
Their result, reproduced below, 
exhibits an $O\big(\frac{\log n}{n}\big)$ gap in the third-order term. 

Define 
\begin{IEEEeqnarray}{rCl}
\mathscr{R}_{\rm in}(n,\epsilon) &\triangleq& \bigg\{\mathbf{R} \in \mathbb{R}^2: 
\overline{\mathbf{R}} \in \overline{\mathbf{H}} 
+ \frac{\mathscr{Q}_{\rm inv}(\mathsf{V},\epsilon)}{\sqrt{n}} 
+  \frac{\nu \log n}{n}\mathbf{1} \bigg\} \notag \\
\label{eq-sw-tk-1}\\
\mathscr{R}_{\rm out}(n,\epsilon) 
&\triangleq& \bigg\{\mathbf{R} \in \mathbb{R}^2:
\overline{\mathbf{R}} \in \overline{\mathbf{H}} 
+ \frac{\mathscr{Q}_{\rm inv}(\mathsf{V},\epsilon)}{\sqrt{n}} 
- \frac{\log n}{n}\mathbf{1} \bigg\}, \notag\\ \label{eq-sw-tk-2}
\end{IEEEeqnarray} 
where $\overline{\mathbf{R}}$ and $\overline{\mathbf{H}}$ 
are defined in \eqref{def-R-H}, 
$\mathsf{V}$ is the entropy dispersion matrix 
for $(X_1,X_2)$ (Definition \ref{def-dis-matrix}), 
$\nu \triangleq |\mathcal{X}_1||\mathcal{X}_2| + \kappa + \frac{3}{2}$, 
and $\kappa$ is the absolute finite positive constant 
from~\cite[Def.~6]{tan-kosut}.

\begin{thm}[{Tan and Kosut \cite[Th.~1]{tan-kosut}}] \label{thm-sw-t-k} 
Consider finite-alphabet, stationary, memoryless sources $(X_1,X_2)$ with 
$P_{X_1 X_2}(x_1,x_2) > 0$ 
for every $(x_1,x_2) \in \mathcal{X}_1 \times \mathcal{X}_2$. 
For any $0 < \epsilon < 1$ and all $n$ sufficiently large, 
	\begin{equation}
	\mathscr{R}_{\rm in}(n,\epsilon)  
	\subseteq \mathscr{R}^{*}(n,\epsilon) 
	\subseteq \mathscr{R}_{\rm out}(n,\epsilon).
	\end{equation} 
\end{thm}


\begin{remark}
The inner boundary defined in \eqref{eq-sw-tk-1} 
is achievable by a universal coding scheme \cite[Sec.~VI]{tan-kosut}. 
The outer bounding region in \eqref{eq-sw-tk-2} 
is based on \cite[Lemma~7.2.2]{han}.
\end{remark}

In~\cite{nomura-han}, Nomura and Han 
use~\cite[Lemma~7.2.1]{han} and~\cite[Lemma~7.2.2]{han} 
to derive a second-order MASC coding theorem 
for stationary, memoryless, dependent sources. 
Their result is equivalent to Theorem~\ref{thm-sw-t-k}
up to the second-order term 
and applies also for countable alphabets. 
Neither~\cite{tan-kosut} nor~\cite{nomura-han} 
finds the precise third-order term.  
In Sections~\ref{sec-sw-nonasymp} and~\ref{sec-result-SW}, below, 
we give new non-asymptotic MASC bounds 
and then apply them to precisely characterize
the third-order asymptotics.  

\subsection{New Non-Asymptotic Bounds} 
\label{sec-sw-nonasymp}
\subsubsection{Achievability}
We present a MASC RCU bound, 
extending Theorem~\ref{thm-rcu-bound} to the multiple-encoder case. 
\begin{thm}[MASC RCU bound] \label{thm-sw-rcu}
	Given discrete random variables $(X_1,X_2)$, 
	there exists an $(M_1,M_2,\epsilon)$ MASC with 
	\begin{equation}
	\epsilon \leq \mathbb{E}\left[ \min \left\{1, \, A_1 + A_2 + A_{12} \right\} \right]
	\label{eq-sw-rcu}
	\end{equation} where
	\begin{IEEEeqnarray}{rCl}
	A_1 &\triangleq& \frac 1 {M_1}
	\mathbb{E}\left[ \exp\left(\imath(\bar{X}_{1}^\prime|X_2)\right) \right.
	\notag \\ 
	&&\left. \qquad 1\left\{\imath(\bar{X}_{1}^\prime |X_2) 
	\leq \imath(X_{1}|X_2)\right\} |X_1,X_2\right] \label{eq-sw-rcu-1} \\ 
	A_2 &\triangleq& \frac 1 {M_2}
	\mathbb{E}\left[ \exp\left(\imath(\bar{X}_{2}^\prime|X_1)\right)\right. \notag \\ 
	&&\left. \qquad 1\left\{\imath(\bar{X}_{2}^\prime|X_1) \leq \imath(X_{2}|X_1)\right\}
	|X_1,X_2 \right] \label{eq-sw-rcu-2}\\
	A_{12} &\triangleq& \frac{1}{M_1 M_2} 
	\mathbb{E}\left[ \exp\left(\imath(\bar{X}_{1},\bar{X}_2)\right) \right. \notag \\ 
	&&\left. \qquad 1\left\{\imath(\bar{X}_1,\bar{X}_2) 
	\leq \imath(X_1,X_2)\right\} |X_1,X_2 \right]  
	\IEEEeqnarraynumspace \label{eq-sw-rcu-3} \\
\lefteqn{
	P_{X_1X_{2}\bar{X}_1\bar{X}_2 \bar{X}_1^\prime \bar{X}_2^\prime}\left(a, b, \bar a, \bar b, \bar a^\prime, \bar b^\prime\right)} \label{eq-sw-rcu-4}\\
	& = &  P_{X_1X_2}(a,b)P_{X_1X_2}(\bar a, \bar b) P_{ X_1 | X_2 }(\bar a^\prime | b) P_{ X_2 | X_1 }(\bar b^\prime | a). \notag
	\end{IEEEeqnarray}
\end{thm}

\begin{proof}
	For every $x_i \in \mathcal{X}_i$, $i\in[2]$, 
	draw $\mathsf{F}_i(x_i)$ 
	i.i.d. uniformly at random from $[M_i]$. 
	The maximum likelihood decoder is 
	defined for each $(c_1,c_2) \in [M_1] \times [M_2]$ by
	\begin{IEEEeqnarray}{rcl}
	\mathsf{g}(c_1,c_2) 
	&=& \arg 
	\underset{\substack{{(x_1,x_2)\in \mathcal{X}_1\times \mathcal{X}_2:} \\
	\mathsf{F}_1(x_1)=c_1,\, \mathsf{F}_2(x_2)=c_2}}
	{\min} \imath(x_1,x_2),
	\end{IEEEeqnarray}
	where ties are broken equiprobably at random in the code design. 
	This decoder is optimal for the given encoder. 
	
	We bound the random code's expected error probability 
	by the probability of the union of events 
	\begin{IEEEeqnarray}{rCl}
	\mathcal{E}_1 
	&\triangleq& 
	\{\exists \, \bar{x}_1 \in \mathcal{X}_1\backslash\{X_1\}: \nonumber \\
	&& \imath(\bar{x}_1|X_2) \leq \imath(X_1|X_2), \, 
	\mathsf{F}_1(\bar{x}_1) = \mathsf{F}_1(X_1) \} \\
	\mathcal{E}_2 
	&\triangleq& 
	 	\{\exists \, \bar{x}_2 \in \mathcal{X}_2\backslash\{X_2\}: \nonumber \\
	&& \imath(\bar{x}_2|X_1) \leq \imath(X_2|X_1), \, 
	\mathsf{F}_2(\bar{x}_2) = \mathsf{F}_2(X_2) \} \\
	\mathcal{E}_{12} &\triangleq& \{\exists \, \bar{x}_1 
	\in \mathcal{X}_1\backslash\{X_1\}, \, \bar{x}_2 \in 
	\mathcal{X}_2\backslash\{X_2\}: \nonumber \\ 
	&& \imath(\bar{x}_1,\bar{x}_2) \leq \imath(X_1,X_2), \nonumber \\
	&&\mathsf{F}_1(\bar{x}_1) = \mathsf{F}_1(X_1), \, 
	\mathsf{F}_2(\bar{x}_2) = \mathsf{F}_2(X_2) \}.
	\IEEEeqnarraynumspace \IEEEeqnarraynumspace
	\end{IEEEeqnarray}
	By a derivation similar 
	to that in the proof of Theorem~\ref{thm-rcu-bound}, 
	\begin{IEEEeqnarray}{rCl}
	\lefteqn{\mathbb{E}[\mathbb{P}\left[\{g(F_1(X_1),F_2(X_2))
		\neq(X_1,X_2)\right]}
	\nonumber \\
	& \leq & \mathbb{P}\left[
	\mathcal{E}_1\cup\mathcal{E}_2\cup\mathcal{E}_{12}\right] \label{eq-sw-rcu-11} \\
	&\leq& \mathbb{E} \Bigg[\min \Bigg\{1, \nonumber \\ 
	&& \hphantom{+}\underset{\qquad \mathclap{
	\bar{x}_1\in\mathcal{X}_1\backslash\{X_1\}}}{\sum \,} 
	\mathbb{P}\left[\imath(\bar{x}_1 |X_2) \leq \imath(X_1 |X_2), \, 
	\mathsf{F}_1(\bar{x}_1) = \mathsf{F}_1(X_1)|X_1,X_2\right] \nonumber \\ 
	&&+ \underset{\qquad \mathclap{
	\bar{x}_2\in\mathcal{X}_2\backslash\{X_2\}}}{\sum \,} 
	\mathbb{P}\left[\imath(\bar{x}_2 |X_1) \leq 
	\imath(X_2 |X_1), \, \mathsf{F}_2(\bar{x}_2) = 
	\mathsf{F}_2(X_2)|X_1,X_2\right] \nonumber \\
	&&+ \underset{\qquad \mathclap{\substack{
	{\bar{x}_1 \in \mathcal{X}_1\backslash\{X_1\}}\\
	{\bar{x}_2 \in \mathcal{X}_2\backslash\{X_2\}}} }}{\sum\;} 
	\mathbb{P}\left[\imath(\bar{x}_1,\bar{x}_2) \leq \imath(X_1,X_2),  
	\right. \nonumber \\
	&& \left. \mathsf{F}_1(\bar{x}_1) = \mathsf{F}_1(X_1), \, 
	\mathsf{F}_2(\bar{x}_2) = \mathsf{F}_2(X_2)|X_1,X_2\right] \Bigg\} \Bigg] 
	 \\ 
	&\leq& \mathbb{E} \Bigg[\min \Bigg\{1, 
	\frac{1}{M_1} \sum\limits_{\bar{x}_1\in\mathcal{X}_1}  
	1\{\imath(\bar{x}_1 |X_2) \leq \imath(X_1 |X_2)\} \nonumber \\
	&& \hphantom{\mathbb{E}\min\Bigg\{\}}+ \frac{1}{M_2}
	\sum\limits_{\bar{x}_2\in\mathcal{X}_2} 
	1\{\imath(\bar{x}_2 |X_1) \leq \imath(X_2 |X_1)\} \nonumber \\
	&& + \frac{1}{M_1 M_2} 
	\sum\limits_{ \bar{x}_1\in\mathcal{X}_1, \, \bar{x}_2\in\mathcal{X}_2} 
	1\{\imath(\bar{x}_1,\bar{x}_2) \leq \imath(X_1,X_2)\} \Bigg\}\Bigg], 
	  \label{eq-sw-rcu-7} 
\end{IEEEeqnarray}
	and \eqref{eq-sw-rcu-7} is equal to the right side of \eqref{eq-sw-rcu} as desired. 
\end{proof}

Figure~\ref{fig-bounds} in Section~\ref{sec-sw-comp} 
plots the point-to-point (Theorem~\ref{thm-rcu-bound}) 
and MASC (Theorem~\ref{thm-sw-rcu}) RCU bounds.  

\subsubsection{Converse} \label{sec-sw-ht-conv}
The MASC composite hypothesis testing converse 
employs the set $\beta_\alpha \left(P, \{Q_j\}_{j=1}^{k}\right)$ 
(see Definition~\ref{def-beta-alpha}) 
and its generalization to $\sigma$-finite measures.

\begin{thm}[Hypothesis testing (HT) converse]\label{thm-sw-cht-conv} 
Let $P_{X_1X_2}$ be the source distribution 
defined on $\mathcal{X}_1 \times \mathcal{X}_2$. 
Let $Q^{(1)}_{X_1X_2}$, $Q^{(2)}_{X_1X_2}$, and $Q^{(3)}_{X_1X_2}$ 
be any $\sigma$-finite measures 
defined on $\mathcal{X}_1 \times \mathcal{X}_2$. 
Any $(M_1, M_2,\epsilon)$ MASC satisfies
\begin{IEEEeqnarray}{rCl}
(\beta_1^*, \beta_2^*, \beta_{3}^*) \in 
\beta_{1-\epsilon}\left(P_{X_1X_2}, 
\left\{Q^{(1)}_{X_1X_2},  Q^{(2)}_{X_1X_2}, Q^{(3)}_{X_1X_2} \right\} \right), 
\label{eq-sw-cht} \nonumber \\*
\end{IEEEeqnarray} where
\begin{IEEEeqnarray}{rCl}
\beta_1^* &\triangleq& M_1 \sum\limits_{x_2 \in \mathcal{X}_2} 
\max\limits_{\hat{x}_1 \in \mathcal{X}_1} Q_{X_1X_2}^{(1)}(\hat{x}_1,x_2), \label{eq:beta1star} \\
\beta_2^* &\triangleq& M_2 \sum\limits_{x_1 \in \mathcal{X}_1} 
\max\limits_{\hat{x}_2 \in \mathcal{X}_2} Q_{X_1X_2}^{(2)}(x_1,\hat{x}_2), \label{eq:beta2star} \\
\beta_{3}^* &\triangleq& M_1M_2
\max\limits_{\hat{x}_1 \in \mathcal{X}_1, \, \hat{x}_2 \in \mathcal{X}_2} 
Q_{X_1X_2}^{(3)}(\hat{x}_1,\hat{x}_2).\label{eq:beta3star}
\end{IEEEeqnarray}
\end{thm}

\begin{proof}
Consider an $(M_1,M_2,\epsilon)$ MASC 
with stochastic encoders $P_{F_1|X_1}$ and $P_{F_2|X_2}$ 
and stochastic decoder $P_{\hat{X}_1\hat{X}_2|F_1F_2}$, 
where $F_1$ and $F_2$ are the encoder outputs,  
and $(\hat{X}_1,\hat{X}_2)$ is the decoder output. 
Fix distributions $\{Q^{(j)}_{X_1X_2}\}_{j=1}^3$ 
on $\mathcal{X}_1 \times \mathcal{X}_2$. 
Then $Z = 1\big\{(\hat{X}_1,\hat{X}_2) = ({X}_1,{X}_2) \big\}$ 
defines a (sub-optimal) composite HT 
for testing $P_{X_1X_2}$ against $\{Q^{(j)}_{X_1X_2}\}_{j = 1}^{3}$, 
for which $\mathbb{P} \left[Z = 1 \right] \geq 1 - \epsilon$ 
and 
\begin{IEEEeqnarray}{rCl}
\lefteqn{\mathbb{Q}^{(1)} \left[Z = 1\right]} \nonumber \\
&=& \sum\limits_{x_1 \in \mathcal{X}_1} 
\sum\limits_{x_2 \in \mathcal{X}_2} Q^{(1)}_{X_1X_2}(x_1,x_2) 
\cdot \sum\limits_{m_1=1}^{M_1}\sum\limits_{m_2=1}^{M_2}
P_{F_1|X_1}(m_1|x_1) \nonumber \\
&&  \cdot P_{F_2|X_2}(m_2|x_2) 
\cdot P_{\hat{X}_1\hat{X}_2|F_1F_2}(x_1,x_2|m_1,m_2) \\
&\leq& \sum\limits_{x_2 \in \mathcal{X}_2} \max\limits_{\hat{x}_1 \in \mathcal{X}_1} Q^{(1)}_{X_1X_2}(\hat{x}_1,x_2)
 \sum\limits_{m_1=1}^{M_1} 
\sum\limits_{m_2=1}^{M_2}
 P_{F_2|X_2}(m_2|x_2) \nonumber \\
&&\cdot \sum\limits_{x_1 \in \mathcal{X}_1} 
P_{F_1|X_1}(m_1|x_1) P_{\hat{X}_1\hat{X}_2|F_1F_2}(x_1,x_2|m_1,m_2)  \label{eq-sw-cht-1}\\
&\leq& M_1 \sum\limits_{x_2 \in \mathcal{X}_2}
\max\limits_{\hat{x}_1 \in \mathcal{X}_1}Q^{(1)}_{X_1X_2}(\hat{x}_1,x_2), \label{eq-sw-cht-3} 
\end{IEEEeqnarray} 
where \eqref{eq-sw-cht-1} follows since 
$\max\limits_{\hat{x}_1 \in \mathcal{X}_1}Q^{(1)}_{X_1X_2}(\hat{x}_1,x_2)$ 
is independent of $x_1$, and  \eqref{eq-sw-cht-3} follows 
by bounding the probability in the sum over $x_1 \in \mathcal{X}_1$ by 1. Similarly, 
\begin{equation}
\mathbb{Q}^{(2)} \left[Z = 1\right] \leq M_2 \sum\limits_{x_1 \in \mathcal{X}_1} \max\limits_{\hat{x}_2 \in \mathcal{X}_2} Q_{X_1X_2}^{(2)}(x_1,\hat{x}_2),
\end{equation} 
\begin{IEEEeqnarray}{rCl}
\IEEEeqnarraymulticol{3}{l}{\mathbb{Q}^{(3)}\left[Z = 1\right]} \nonumber \\
&=& \sum\limits_{x_1 \in \mathcal{X}_1} \sum\limits_{x_2 \in \mathcal{X}_2} 
Q^{(3)}_{X_1X_2}(x_1,x_2) 
\sum\limits_{m_1=1}^{M_1} \sum\limits_{m_2=1}^{M_2}
P_{F_1|X_1}(m_1|x_1) \cdot \nonumber \\
&& P_{F_2|X_2}(m_2|x_2)P_{\hat{X}_1\hat{X}_2|F_1F_2}(x_1,x_2|m_1,m_2) \\
&\leq& \max\limits_{\hat{x}_1 \in \mathcal{X}_1, \, \hat{x}_2 \in \mathcal{X}_2} 
Q_{X_1X_2}^{(3)}(\hat{x}_1,\hat{x}_2) \cdot  \nonumber \\ 
&&\sum\limits_{m_1=1}^{M_1} \sum\limits_{m_2=1}^{M_2} 
\sum\limits_{x_1 \in \mathcal{X}_1} \sum\limits_{x_2 \in \mathcal{X}_2} 
P_{\hat{X}_1\hat{X}_2|F_1F_2}(x_1,x_2|m_1,m_2) \label{eq-sw-cht-6} \IEEEeqnarraynumspace \\
&=& M_1M_2 \max\limits_{\hat{x}_1\in\mathcal{X}_1, \, 
\hat{x}_2\in \mathcal{X}_2} Q_{X_1X_2}^{(3)}(\hat{x}_1,\hat{x}_2).
\end{IEEEeqnarray} 
Thus~\eqref{eq-sw-cht} holds 
by the definition of $\beta_{1-\epsilon} \left(P, \{Q^{(j)}\}_{j=1}^{k}\right)$.
\end{proof}

To recover Han's converse (Theorem~\ref{lem-han-conv}) 
from Theorem~\ref{thm-sw-cht-conv}, 
let $P_{X_1}$ and $P_{X_2}$ be the marginals of $P_{X_1X_2}$ 
and let $U_{X_1}$, $U_{X_2}$, and $U_{X_1X_2}$ 
be the counting measures over $\mathcal{X}_1$, $\mathcal{X}_2$, 
and $\mathcal{X}_1 \times \mathcal{X}_2$. 
By Theorem~\ref{thm-sw-cht-conv}, 
any $(M_1,M_2,\epsilon)$ MASC  satisfies 
\begin{IEEEeqnarray}{rCl}
\IEEEeqnarraymulticol{3}{l}{(M_1,M_2,M_1M_2)} \nonumber \\
&\in& \beta_{1-\epsilon}\left(P_{X_1X_2}, \left\{U_{X_1}P_{X_2}, P_{X_1}U_{X_2}, U_{X_1X_2} \right\} \right). \IEEEeqnarraynumspace \label{eq-sw-cht-5}
\end{IEEEeqnarray} 
Applying Lemma~\ref{lem-cht-conv} 
to~\eqref{eq-sw-cht-5} with $k = 3$ gives 
\begin{IEEEeqnarray}{rCl}
\epsilon &\geq& \mathbb{P}\left[ \left\{
\imath(X_1|X_2) \geq \log \frac{1}{\gamma_1}\right\} 
\cup \left\{\imath(X_2|X_1) \geq \log \frac{1}{\gamma_2}\right\} \right. 
\nonumber \\
&& \left. \cup \left\{\imath(X_1,X_2) \geq \log \frac{1}{\gamma_3}\right\}\right] 
- \gamma_1 M_1 - \gamma_2 M_2 - \gamma_3 M_1M_2. \nonumber
\end{IEEEeqnarray} 
Setting $\gamma_1 = \frac{\exp\left(-\gamma\right)}{M_1}$, $\gamma_2 = \frac{\exp\left(-\gamma\right)}{M_2}$, and $\gamma_3 = \frac{\exp\left(-\gamma\right)}{M_1M_2}$ for an arbitrary $\gamma > 0$ recovers Theorem~\ref{lem-han-conv}.

To show that Theorem~\ref{thm-sw-cht-conv} is equivalent 
to the LP-based converse (Theorem~\ref{thm-lp-conv}), 
we apply~\eqref{eq-sw-cht-8} to Theorem~\ref{thm-sw-cht-conv},  
showing that any $(M_1, M_2, \epsilon)$ MASC  satisfies
\begin{IEEEeqnarray}{rCl}
\epsilon &\geq& 
\sup\limits_{Q^{(1)}_{X_1X_2}, \, Q^{(2)}_{X_1X_2}, \, Q^{(3)}_{X_1X_2}} 
\epsilon^*(\beta_1^*, \beta_2^*, \beta_{3}^*) \\
&=& \sup\limits_{Q^{(1)}_{X_1X_2}, \, Q^{(2)}_{X_1X_2}, \, Q^{(3)}_{X_1X_2}} 
\sup\limits_{\gamma_1, \gamma_2, \gamma_3 \geq 0} \left\{
\vphantom{\sum\limits_{x_1 \in \mathcal{X}_1,x_2 \in \mathcal{X}_2}} \right.
 \nonumber \\
&&\sum\limits_{\substack{{x_1 \in \mathcal{X}_1} \\ {x_2 \in \mathcal{X}_2}}} 
\min \left\{P_{X_1X_2}(x_1,x_2), \, \sum\limits_{j=1}^{3} \gamma_j 
Q^{(j)}_{X_1X_2}(x_1, x_2) \right\}  \nonumber \\ 
&&-\gamma_1 M_1 \sum\limits_{x_2 \in \mathcal{X}_2} 
\max\limits_{\hat{x}_1 \in \mathcal{X}_1} Q^{(1)}_{X_1X_2}(\hat{x}_1, x_2) 
\nonumber \\
&& - \gamma_2 M_2 \sum\limits_{x_1 \in \mathcal{X}_1} 
\max\limits_{\hat{x}_2 \in \mathcal{X}_2} Q^{(2)}_{X_1X_2}(x_1, \hat{x}_2) 
\nonumber \\
&& \left. - \gamma_3 M_1M_2 
\max\limits_{\hat{x}_1 \in \mathcal{X}_1, \, \hat{x}_2 \in \mathcal{X}_2} 
Q_{X_1X_2}^{(3)}(\hat{x}_1,\hat{x}_2) 
\vphantom{\sum\limits_{x_1 \in \mathcal{X}_1, x_2 \in \mathcal{X}_2}} \right\}.
 \label{eq-sw-cht-9}
\end{IEEEeqnarray} 
The outer supremum is over 
$\sigma$-finite measures $Q^{(1)}_{X_1X_2}$, $Q^{(2)}_{X_1X_2}$, 
and $Q^{(3)}_{X_1X_2}$.
In Appendix~\ref{append-equivalence}, 
we show that the bounds in~\eqref{eq-sw-cht-9} and~\eqref{eq-lem-jk} 
are equivalent, 
establishing the equivalence between 
the MASC HT (Theorem~\ref{thm-sw-cht-conv}) 
and LP (Theorem~\ref{thm-lp-conv}) converses. 

\begin{remark}
When one of the sources is deterministic, 
the MASC HT converse reduces to 
the point-to-point HT converse~\cite[Eq.~(64)]{kostina-verdu}. 
For example, if $X_2$ is deterministic, then \eqref{eq-sw-cht-5} reduces to 
\begin{IEEEeqnarray}{rCl}
(M_1, 1, M_1) \in \beta_{1-\epsilon}\left(P_{X_1X_2}, \left\{U_{X_1}P_{X_2}, P_{X_1}U_{X_2}, U_{X_1X_2} \right\} \right), \nonumber 
\end{IEEEeqnarray}
which further reduces to
\begin{equation}
M_1 \geq \beta_{1-\epsilon}\left(P_{X_1}, U_{X_1}\right),\nonumber
\end{equation} 
where $\beta_\alpha(P,Q)$ is the optimal $\beta$-function for binary hypothesis testing between distributions $P$ and $Q$.
\end{remark}

\subsection{Asymptotics: Third-Order MASC Rate Region}
\label{sec-result-SW}
The following third-order asymptotic characterization 
of the MASC rate region 
for stationary, memoryless sources 
closes the $O\big(\frac{\log n}{n}\big)$ gap 
between \eqref{eq-sw-tk-1} and \eqref{eq-sw-tk-2}.

Consider stationary, memoryless sources 
with single-letter joint distribution $P_{X_1 X_2}$ 
for which 
\begin{IEEEeqnarray}{rCl}
\!\!\!\!\!\!\!\!\!\!  &&V(X_1,X_2) \!>\! 0, \,\mathbb E \left[ V_c (X_1 | X_2) \right] \!>\! 0, \, \mathbb E \left[ V_c (X_2 | X_1) \right]  \!>\! 0, \label{assump-b1}\\
\!\!\!\!\!\!\!\!\!\!  && T(X_1,X_2)  < \infty,\, T(X_1|X_2) < \infty, \, T(X_2|X_1) < \infty, \label{assump-b2}  \\
\!\!\!\!\!\!\!\!\!\!  && \mathbb E \left[ T_c^2 (X_1 | X_2) \right] < \infty,\, \mathbb E \left[ T_c^2 (X_2 | X_1) \right] < \infty. \label{assump-b3}
 \IEEEeqnarraynumspace
\end{IEEEeqnarray}
When \eqref{assump-b1} holds\footnote{In fact, the weaker condition 
$V(X_1|X_2) > 0$, $V(X_2|X_1)>0$, $V(X_1,X_2) > 0$ suffices.}, 
$\text{rank}(\mathsf{V}) \geq 1$. 
Technical assumptions \eqref{assump-b1}, \eqref{assump-b2}, and \eqref{assump-b3} are required 
to ensure applicability of the multidimensional Berry-Esseen theorem 
and Lemma~\ref{lem-p} in our asymptotic analysis. 
Assumption \eqref{assump-b2} is satisfied automatically if the alphabets $\mathcal X_1$ and $\mathcal X_2$ are finite. 

Define the set 
\begin{equation}
 \overline{\mathscr{R}}^{*}(n,\epsilon) 
\triangleq \left\{\overline{\mathbf{R}} \in \mathbb R^3 \colon \overline{\mathbf{R}} =  \overline{\mathbf{H}}
+ \frac{\mathscr{Q}_{\rm inv}(\mathsf{V},\epsilon)}{\sqrt{n}} 
- \frac{\log n}{2n}\mathbf{1} \right\}, \label{eq-def-third-order}
\end{equation} 
where vector $\overline{\mathbf{H}}$ is defined in \eqref{def-R-H}, 
$\mathsf{V}$ is the entropy dispersion matrix for $(X_1,X_2)$, 
and $\mathscr{Q}_{\rm inv}(\mathsf{V},\epsilon)$ is defined in \eqref{eq-def-sve}. 
Note that $\mathscr{R}^*(n,\epsilon) \subset \mathbb{R}^2$ 
(see Definition \ref{def-sw-rate-region}) 
but  $\overline{\mathscr{R}}^{*}(n,\epsilon) \subset \mathbb{R}^3$. 
Define the inner and outer bounding sets
\begin{IEEEeqnarray}{rCl}
\mathscr{R}_{\rm in}^{*}(n,\epsilon) 
&\triangleq& \bigg\{\mathbf{R} \in \mathbb{R}^{2}: \overline{\mathbf{R}} \in 
\overline{\mathscr{R}}^{*}(n,\epsilon) + O\left(\frac{1}{n}\right)\mathbf{1} \bigg\} \\
\mathscr{R}_{\rm out}^{*}(n,\epsilon) 
&\triangleq& \bigg\{\mathbf{R} \in \mathbb{R}^{2}: \overline{\mathbf{R}} \in 
\overline{\mathscr{R}}^{*}(n,\epsilon) - O\left(\frac{1}{n}\right)\mathbf{1} \bigg\}. 
\label{eq-def-sw-2-out} \IEEEeqnarraynumspace
\end{IEEEeqnarray}

\begin{thm}[Third-order MASC rate region] \label{thm-sw}
Consider a pair of stationary, memoryless sources 
with single-letter joint distribution $P_{X_1 X_2}$ 
satisfying \eqref{assump-b1}--\eqref{assump-b3}. 
For any $0 < \epsilon < 1$, 
the $(n,\epsilon)$-rate region $\mathscr{R}^{*}(n,\epsilon)$ satisfies
	\begin{equation}
	\mathscr{R}_{\rm in}^{*}(n,\epsilon)
	\subseteq \mathscr{R}^{*}(n,\epsilon)
	\subseteq \mathscr{R}_{\rm out}^{*}(n,\epsilon).
	\end{equation}
\end{thm}

Since the upper and lower bounds in Theorem~\ref{thm-sw} 
agree up to their third-order terms, 
we call $\overline{\mathscr{R}}^{*}(n,\epsilon)$ 
the {\em third-order MASC rate region}. 
Figure~\ref{fig-sw-third-order} 
plots the boundaries of $\overline{\mathscr{R}}^{*}(n,\epsilon)$ 
at different values of $n$ 
for an example pair of sources. 

\begin{figure}[!t]
	\centering
	\includegraphics[width=0.42\textwidth]{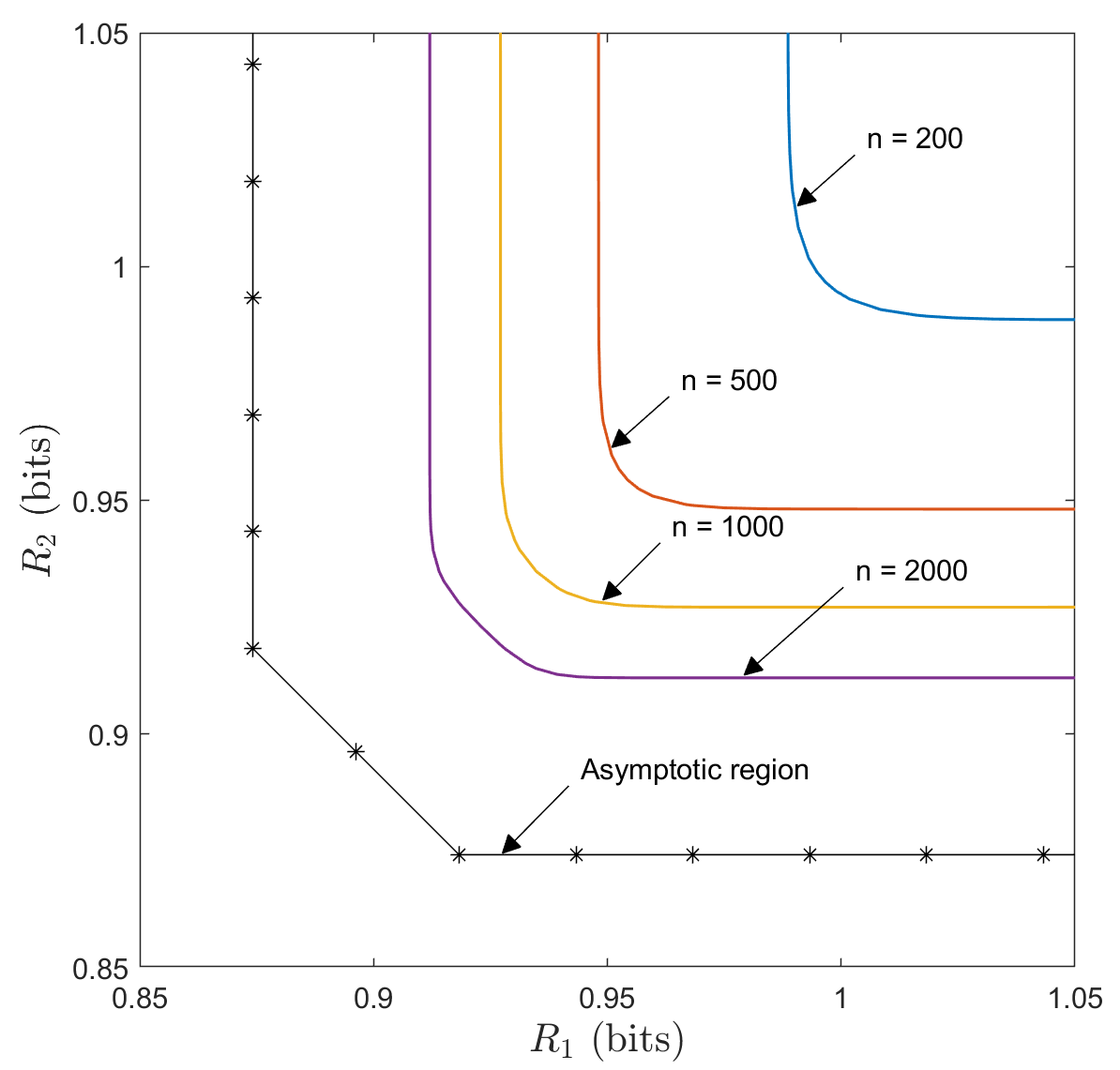}
	\caption{Third-order MASC rate regions $\overline{\mathscr{R}}^{*}(n,\epsilon)$ 
	at $\epsilon = 10^{-3}$ 
	for stationary, memoryless sources $(X_1,X_2)$ 
	with $p_{X_1,X_2}(0,0)=1/2$, 
	$p_{X_1,X_2}(0,1)=p_{X_1,X_2}(1,1)=p_{X_1,X_2}(1,1)=1/6$.}
	\label{fig-sw-third-order}
\end{figure} 

\begin{remark}
As noted in Remark~\ref{rem-zero-var},  
for point-to-point source coding, 
zero varentropy means that the source is uniform; 
the $-\frac{\log n}{2n}$ third-order term is absent in that case. 
While condition~\eqref{assump-b1} limits Theorem~\ref{thm-sw} 
to sources with positive varentropies, 
Appendix~\ref{append-zero-var} 
considers the case where one or more varentropies are zero.  
Roughly, each zero varentropy yields a zero dispersion, 
and the absence of a $-\frac{\log n}{2n}$ third-order term, 
similar to the point-to-point case. Furthermore, if $V(X_1 | X_2) > 0$ but $\mathbb E\left[ V_c(X_1 | X_2)\right] = 0$, the corresponding achievable third order term increases from $- \frac{\log n}{2n}$ to $0$.\footnote{This is seen by modifying the reasoning in \eqref{eq:S1}--\eqref{eq-sw-achiev-3} in the proof of Theorem~\ref{thm-sw} below.} This means that the optimal third order term lies in $[- \frac{\log n}{2n}, 0]$ in that case.
\end{remark}

\begin{proof}[Proof of Theorem \ref{thm-sw}: achievability]
We apply Theorem~\ref{thm-sw-rcu} 
to stationary, memoryless sources with 
$P_{X_{1}^{n}X_{2}^{n}} = P_{X_{1}X_{2}}^n$ and then apply Lemmas~\ref{lem-p} and~\ref{lem-b-e} to analyze the bound.
Let 
\begin{IEEEeqnarray}{rCl} 
I_{1} &\triangleq \imath(X_1^{n}|X_2^{n}) =& \sum_{i=1}^{n}\imath(X_{1i}|X_{2i}) 
\label{eq-def-I1n} \\
I_{2} &\triangleq \imath(X_2^{n}|X_1^{n}) =& \sum_{i=1}^{n}\imath(X_{2i}|X_{1i}) 
\label{eq-def-I2n} \\
I_{12} &\triangleq \imath(X_1^{n},X_2^{n}) =& \sum_{i=1}^{n}\imath(X_{1i},X_{2i}) 
\label{eq-def-In} \\
\bar{I}_{1} &\triangleq \imath(\bar{X}_1^{n\, \prime}|X_2^{n}) =& \sum_{i=1}^{n}
\imath(\bar{X}_{1i}^\prime |X_{2i}) \\
\bar{I}_{2} &\triangleq \imath(\bar{X}_2^{n\, \prime}|X_1^{n}) =& \sum_{i=1}^{n}
\imath(\bar{X}_{2i}^\prime|X_{1i})\\
\bar{I}_{12} &\triangleq \imath(\bar{X}_1^n,\bar{X}_2^n) =& \sum_{i=1}^{n}
\imath(\bar{X}_{1i},\bar{X}_{2i}) 
\end{IEEEeqnarray} 
where $(X_{1i}, X_{2i}, \bar{X}_{1i}, \bar{X}_{2i}, \bar{X}_{1i}^\prime, \bar{X}_{2i}^\prime )$, $i = 1, \ldots, n$, 
are drawn i.i.d. according to the joint distribution defined in \eqref{eq-sw-rcu-4}. With this notation, the random variables $A_1$, $A_2$, $A_{12}$ defined in \eqref{eq-sw-rcu-1}, \eqref{eq-sw-rcu-2}, \eqref{eq-sw-rcu-3} particularize as
\begin{align}
A_1 & = \frac 1 {M_1} \mathbb{E}\left[\exp \left(\bar{I}_{1}\right) 1\{\bar{I}_{1} \leq I_{1} \}|X_1^n,X_2^n \right] \\
A_2 &= \frac 1 {M_2} \mathbb{E}\left[\exp\left(\bar{I}_{2}\right) 1\{\bar{I}_{2} \leq I_{2} \}|X_1^n,X_2^n \right] \\
A_{12} &= \frac 1 {M_1 M_2} \mathbb{E}\left[\exp\left(\bar{I}_{12}\right) 1\{\bar{I}_{12} \leq I_{12} \}|X_1^n,X_2^n \right]
\end{align}
By Theorem~\ref{thm-sw-rcu}, there exists an $(n,M_1,M_2,\epsilon')$ 
MASC  such that  
\begin{align}
\epsilon' &\leq \mathbb{E}\left[ \min \left\{1, \, A_1 + A_2+ A_{12} \right\} \right] \label{eq-sw-achiev-1} \\
&= \mathbb E\left[ \left( A_1 + A_2 + A_{12} \right)  1 \left\{  A_1 +  A_2 +  A_{12}  \leq 1 \right\}  \right] \notag\\
&\phantom{=}+ \mathbb P \left[  A_1 +  A_2 +  A_{12}  > 1 \right] \label{eq-sw-achiev-2} \\
&\leq \mathbb E \left[   A_1 1 \left\{  A_1 \leq 1 \right\} \right] + \mathbb E \left[   A_2 1 \left\{  A_2 \leq 1 \right\} \right] \label{eq:second}\\
&\phantom{=}+ \mathbb E \left[   A_{12} 1 \left\{  A_{12} \leq 1 \right\} \right]  + \mathbb P \left[ 3  A_1 > 1 \cup 3  A_2 > 1 \cup 3  A_{12}  > 1 \right] \notag
\end{align}

To bound each of the terms in \eqref{eq:second}, we first bound the random variables $A_1$, $A_2$, and $A_{12}$ by random variables  $\bar A_1$, $\bar A_2$, and $\bar A_{12}$ that are easier to work with. 

Denote constants
\begin{align}
 K_1 &\triangleq \frac{2 \log 2}{\sqrt{2\pi V(X_1|X_2)}} 
+ \frac{2 C_0 T(X_1|X_2)  }{ V(X_1|X_2)^{3/2}} \label{eq-def-K1} \\
 K_2 &\triangleq \frac{2 \log 2}{\sqrt{2\pi V(X_2|X_1)}} 
+ \frac{2 C_0 T(X_2|X_1)  }{ V(X_2|X_1)^{3/2}} \label{eq-def-K2}\\
 K_{12} &\triangleq \frac{2 \log 2}{\sqrt{2\pi V(X_1,X_2)}} 
+ \frac{2 C_0  T(X_1,X_2)}{V(X_1,X_2)^{3/2}} \label{eq-def-K12} 
\end{align} 
that are finite by assumptions \eqref{assump-b1} and \eqref{assump-b2}. Define
\begin{align}
V_{1}(x_1^n) &\triangleq \frac 1 n \sum_{i = 1}^n \mathrm{Var} \left[ \imath(X_{2i} | X_{1i} = x_{1i})\right]\\
T_{1}(x_1^n) &\triangleq \frac 1 n \sum_{i = 1}^n \mathbb E \left[\! \left| \imath(X_{2i} | X_{1i} = x_{1i}) \!- \! \mathbb E \left[ \imath(X_{2i} | X_{1i} = x_{1i})\right] \right|^3 \!\right] 
\end{align}
for $x_1^n \in \mathcal X_1^n$.  Define $V_{2}(x_2^n)$ and $T_{2}(x_2^n)$ for $x_2^n \in \mathcal X_2^n$ analogously.

Applying Lemma~\ref{lem-p} to $A_{12}$ yields
\begin{align}
  A_{12} &\leq \bar A_{12} \triangleq \frac{ K_{12}\exp\left(I_{12}\right)}{M_{1}M_2\sqrt{n}}. \label{eq:A12bar} 
\end{align}

To bound $A_1$, we consider the cases $V_{2}(x_2^n) > 0$ and $V_{2}(x_2^n) = 0$ separately. If $V_{2}(x_2^n) > 0$, then 
\begin{align}
 K_1(x_2^n) &\triangleq \frac{2 \log 2}{\sqrt{2\pi V_2(x_2^n)}} 
+ \frac{2C_0 T_{2}(x_2^n) }{V_{2}(x_2^n)^{3/2}} \label{eq-def-K1x} 
\end{align}
is finite by assumption \eqref{assump-b2},  and Lemma~\ref{lem-p} yields 
\begin{align}
A_1 &\leq \frac{ K_1(X_2^n) \exp\left(I_{1}\right)}{M_1\sqrt{n}} , \quad \text{if } V_{2}(X_2^n) > 0 \label{eq:A1bar}. 
\end{align}
If $V_{2}(x_2^n) = 0$, then $I_1 = \bar I_1 = H(X_1^n | X_2^n = x_2^n)$ irrespective of the realization of $X_1^n$, and
\begin{align}
A_1 &= \frac{\exp\left(I_{1}\right)}{M_1}   , \quad \text{if } V_{2}(X_2^n) = 0 \label{eq:A1bar0}.
\end{align}
Putting \eqref{eq:A1bar} and \eqref{eq:A1bar0} together yields
\begin{align}
A_1 \leq \bar A_1 \triangleq 
\begin{cases}
 \frac{ K_1(X_2^n) \exp\left(I_{1}\right)}{M_1\sqrt{n}}, & V_{2}(X_2^n) > 0 \\
 \frac{\exp\left(I_{1}\right)}{M_1}, & V_{2}(X_2^n) = 0.
\end{cases}
\end{align}
Similarly, 
\begin{align}
A_2 \leq \bar A_2 \triangleq 
\begin{cases}
 \frac{ K_2(X_1^n) \exp\left(I_{2}\right)}{M_2\sqrt{n}}, & V_{1}(X_1^n) > 0 \\
 \frac{\exp\left(I_{2}\right)}{M_2}, & V_{1}(X_1^n) = 0,
\end{cases}
\end{align}
where $K_2(x_1^n)$ is defined analogously to \eqref{eq-def-K1x}.

Next, we apply Lemma~\ref{lem-p} again to further bound each of the first three terms in \eqref{eq:second}:
\begin{align}
\mathbb E \left[   A_1 1 \left\{  A_1 \leq 1 \right\} \right]  &\leq \mathbb E \left[   \bar A_1 1 \left\{  \bar A_1 \leq 1 \right\} \right] 
\leq \frac{K_1}{\sqrt n} \label{eq:sec1}\\
\mathbb E \left[   A_2 1 \left\{  A_2 \leq 1 \right\} \right]  &\leq \mathbb E \left[   \bar A_2 1 \left\{  \bar A_2 \leq 1 \right\} \right]  
\leq \frac{K_2}{\sqrt n} \label{eq:sec2}\\
\mathbb E \left[   A_{12} 1 \left\{  A_{12} \leq 1 \right\} \right]  &\leq \mathbb E \left[   \bar A_{12} 1 \left\{  \bar A_{12} \leq 1 \right\} \right]  
\leq \frac{K_{12}}{\sqrt n}. \label{eq:sec12}
\end{align}
 
We proceed to bound the last term in \eqref{eq:second}.
For fixed constants $s_{1} <  \mathbb E\left[V_c(X_2|X_1)\right]$ and $s_{2} < \mathbb E\left[V_c(X_1|X_2) \right]$, define the events $\mathcal  S_1$ and $\mathcal S_2$ that $X_1^n$ and $X_2^n$ are typical, respectively: 
\begin{align}
 \mathcal  S_1 \triangleq &~ \big\{ V_{1}(X_1^n) \geq \mathbb E[V_c(X_2 | X_1)] - s_1, \notag\\
  &~\phantom{\{}T_{1}(X_1^n) \leq \mathbb E\left[T_{c}(X_2 | X_1)\right] + s_{1} \big \}\label{eq:S1} \\
 \mathcal  S_2 \triangleq&~ \big\{ V_{2}(X_2^n) \geq \mathbb E[V_c(X_1 | X_2)] - s_2, \notag\\
   &~\phantom{\{} T_{2}(X_2^n) \leq \mathbb E[T_{c}(X_1 | X_2)] + s_2 \big\} \label{eq:S2}
\end{align}
Note that
\begin{align}
\bar A_1 &\leq  \bar{\bar A}_1 \triangleq \frac{ \bar K_1 \exp\left(I_{1}\right)}{M_1\sqrt{n}} , \quad \text{if } \mathcal S_2 \text{ occurs} \label{eq:A1baru}\\
\bar A_2 &\leq  \bar{\bar A}_2 \triangleq \frac{ \bar K_2 \exp\left(I_{2}\right)}{M_2\sqrt{n}} , \quad \text{if } \mathcal S_1 \text{ occurs,} \label{eq:A2baru}
\end{align}
where
\begin{align}
\bar K_1 &\triangleq \frac{2 \log 2}{\sqrt{2\pi (\mathbb E\left[V_c(X_1|X_2)\right] - s_2)}} 
+ \frac{2C_0 \left( \mathbb E\left[ T_c(X_1|X_2)\right]+s_2 \right)}{(\mathbb E \left[V_c(X_1|X_2) \right]- s_2)^{3/2}} \label{eq-def-K1bar} \\
\bar K_2 &\triangleq \frac{2 \log 2}{\sqrt{2\pi (\mathbb E\left[V_c(X_2|X_1) \right]- s_1)}} 
+ \frac{2 C_0 \left(\mathbb E\left[ T_c(X_2|X_1)\right] + s_1 \right)}{(\mathbb E\left[V_c(X_2|X_1)\right]-s_1)^{3/2}} \label{eq-def-K2bar}
\end{align} 
are both finite by the assumptions in \eqref{assump-b1} and \eqref{assump-b2}.

Applying the union bound to $\mathbb P \left[ \mathcal  S^c_k\right]$, $k \in \{1, 2\}$, and Chebyshev's inequality 
\begin{align}
\mathbb P \left[ |Z - \mathbb E\left[Z \right]| > \delta \right] \leq  \frac{ \mathrm {Var} \left[ Z  \right]  }{\delta^2} 
\end{align}
to both terms, we observe that for each $k \in \{1, 2\}$,  
\begin{align}
\mathbb P \left[ \mathcal  S_k^c\right] \leq \frac {S_k} {n}, \label{eq:cheb}
\end{align}
where 
\begin{align}
S_1 &\triangleq \frac 1 {s^2_1} \left( \mathbb E[V_c^2(X_2 | X_1)] + \mathbb E[T_{c}^2(X_2 | X_1)]  \right) \label{eq:T1}\\
S_2 &\triangleq \frac 1 {s^2_2}   \left( \mathbb E[V_c^2(X_1 | X_2)]  + \mathbb E[T_{c}^2(X_1 | X_2)] \right) \label{eq:T2}
\end{align}
are finite by assumption \eqref{assump-b3}.

 We are now prepared to apply Lemma~\ref{lem-b-e} to the last term in \eqref{eq:second}.
Pick any pair of rates $(R_1,R_2)$ satisfying 
\begin{align}
\overline {\mathbf R} \in  \overline{\mathscr{R}}^{*}\left(n,\epsilon - \frac{B}{\sqrt n} -  \frac{K_1 + K_2 + K_{12}}{\sqrt{n}} - \frac{S_1 + S_2}{n}\right) + \frac 1 n \mathbf C, \label{eq:MASCrates}
\end{align}
where the set $\overline{\mathscr{R}}^{*}(n,\epsilon)$ is defined in \eqref{eq-def-third-order}, $\mathbf C \triangleq  \left( \log (3 \bar K_1), \log (3 \bar K_2), \log (3 K_{12}) \right)^T$, and $B$ is the Bentkus constant in the right-side of \eqref{eq-lem-b-e} for zero-mean i.i.d. random vectors
\begin{align}
\mathbf{I}_i &\triangleq
\begin{bmatrix}
\imath(X_{1i}|X_{2i}) \\
\imath(X_{2i}|X_{1i}) \\
\imath(X_{1i},X_{2i})
\end{bmatrix} - \overline{\mathbf{H}}, \, \text{for } i = 1,\ldots,n. \label{vardef-ui}
\end{align} 
Note that $B < \infty$ by assumption~\eqref{assump-b2}. We have 
\begin{align}
&~\mathbb P \left[ 3  A_1 \leq 1 \cap 3  A_2 \leq 1 \cap 3  A_{12}  \leq 1 \right]  \notag\\
\geq&~ \mathbb P \left[ 3 \bar A_1 \leq 1 \cap 3 \bar A_2 \leq 1 \cap 3 \bar A_{12}  \leq 1 \cap \mathcal S_1 \cap \mathcal S_2 \right]   \\
\geq&~ \mathbb P \left[ 3 \bar {\bar A}_1 \leq 1 \cap 3 \bar {\bar A}_2 \leq 1 \cap 3 \bar A_{12}  \leq 1  \right]  - \mathbb P \left[\mathcal S_1^c \cup \mathcal S_2^c\right] \label{eq:intersect} \\
=&~ \mathbb{P} \Bigg[ \sum_{i = 1}^n \mathbf{I}_i \leq n \left(\overline{\mathbf{R}} 
- \overline{\mathbf{H}} + \frac{\log n}{2n}\mathbf{1} - \frac{1}{n} 
\mathbf{C}\right)  \Bigg]  - \mathbb P \left[\mathcal S_1^c \cup \mathcal S_2^c\right]  \label{eq:subR} \\
\geq&~ 1 - \epsilon + \frac{K_1 + K_2 + K_{12}}{\sqrt{n}}, \label{eq-sw-achiev-3} 
\end{align}
where \eqref{eq:intersect} applies \eqref{eq:A1baru} and $\mathbb P[\mathcal A \cap \mathcal B] \geq \mathbb P[\mathcal A] - \mathbb P[\mathcal B^c]$, and \eqref{eq-sw-achiev-3} applies \eqref{eq:MASCrates}, Lemma~\ref{lem-b-e} and \eqref{eq:cheb}.

Substituting \eqref{eq:sec1}, \eqref{eq:sec2}, \eqref{eq:sec12}, and \eqref{eq-sw-achiev-3} into \eqref{eq:second}
yields $\epsilon^\prime \leq \epsilon$, and the proof is complete since the set of $(R_1,R_2)$ satisfying  \eqref{eq:MASCrates} contains $\mathscr{R}_{\rm in}^{*}(n,\epsilon)$ by Lemma~\ref{lem-sve}-\ref{part1}.
\end{proof}
\begin{proof}[Proof of Theorem \ref{thm-sw}: converse]
We invoke Theorem~\ref{thm-sw-cht-conv} 
with $P_{X_1X_2} = P_{X_1X_2}^n$, 
$Q^{(1)}_{X_1X_2} = U_{X_1}^nP_{X_2}^n$, 
$Q^{(2)}_{X_1X_2} = P_{X_1}^nU_{X_2}^n$, and 
$Q^{(3)}_{X_1X_2} = U_{X_1X_2}^n$, 
where $P_{X_1}$ and $P_{X_2}$ are the marginals of $P_{X_1X_2}$, 
and $U_{X_1}$, $U_{X_2}$, and $U_{X_1X_2}$ are the counting measures 
over $\mathcal{X}_1$, $\mathcal{X}_2$, and $\mathcal{X}_1\times\mathcal{X}_2$.
Applying Theorem~\ref{thm-cht-third-order} 
to $\beta_{1-\epsilon}\left(P_{X_1X_2}, 
\left\{U_{X_1}P_{X_2}, P_{X_1}U_{X_2}, U_{X_1X_2} \right\} \right)$ 
under the assumptions in \eqref{assump-b1} and \eqref{assump-b2}, 
we conclude that in order to attain error probability $\epsilon$, $M_1$ and $M_2$ must satisfy
\begin{IEEEeqnarray}{rCl}
\lefteqn{(M_1,M_2,M_1M_2)} \notag\\
&&\in \exp\left\{n\overline{\mathbf{H}} + \sqrt{n}\mathscr{Q}_{\rm inv}(\mathsf{V},\epsilon) - \frac{\log n}{2}\mathbf{1} - O\left(1\right)\mathbf{1}\right\}, \IEEEeqnarraynumspace
\end{IEEEeqnarray} which is equivalent to $(R_1, R_2) \in \mathscr{R}_{\rm out}^{*}(n,\epsilon)$ \eqref{eq-def-sw-2-out}, as desired.
\end{proof}
\begin{remark}
The converse of Theorem~\ref{thm-sw} 
can also be proved using Han's converse (Theorem~\ref{lem-han-conv}) with $\gamma = \frac{\log n}{2}$ 
and Lemmas~\ref{lem-b-e} and~\ref{lem-sve} 
in a way similar to that in the achievability proof above, 
except that we would use Lemma~\ref{lem-sve} 
to bound $\mathscr{Q}_{\rm inv}\left(\mathsf{V},\epsilon 
+ O \left(\frac{1}{\sqrt{n}} \right)\big) \subseteq \mathscr{Q}_{\rm inv}(\mathsf{V},\epsilon) 
- O \left(\frac{1}{\sqrt{n}} \right) \right) \mathbf{1}$ 
instead of $\mathscr{Q}_{\rm inv}\left( \mathsf{V},\epsilon - O \left(\frac{1}{\sqrt{n}} \right) \right)$. 
Our HT converse (Theorem~\ref{thm-sw-cht-conv}) 
is stronger than Han's converse,  
but the gap is in the fourth- or higher-order terms, 
as illustrated through computation in Figure~\ref{fig-bounds}. Han's achievability bound (Theorem~\ref{lem-han-achiev}) with the third-order optimal choice of $\gamma = \frac{\log n}{2}$ leads to the third order term of $+ \frac {\log n}{2n}$ instead of $- \frac {\log n}{2n}$. Thus Han's achievability is weaker than our RCU bound (Theorem~\ref{thm-sw-rcu}) in the third order term. 
\end{remark} 

\begin{remark} \label{rem-t-k-conv}
Tan and Kosut's converse (Theorem~\ref{thm-sw-t-k}) 
is also based on Han's converse.  
Instead of deriving an outer bound on 
$\mathscr{Q}_{\rm inv}\left(\mathsf{V},\epsilon + O \left(\frac{1}{\sqrt{n}} \right) \right)$ 
as given in Lemma~\ref{lem-sve}, 
they apply the multivariate Taylor approximation 
to expand the probability, giving a bound that is loose in the third-order term.
\end{remark}

\begin{remark} \label{rem-sw-general}
Theorem~\ref{thm-sw} 
generalizes to any finite number of encoders. 
Let $\mathcal{T} \subset \mathbb{N}$ be a nonempty ordered set 
with a unique index for each encoder. 
For any vector $\mathbf{R}_\mathcal{T} \in \mathbb{R}^{|\mathcal{T}|}$, 
define the $\left(2^{|\mathcal{T}|}-1\right)$-dimensional vector 
of its partial sums as
\begin{equation}
\overline{\mathbf{R}}_{\mathcal{T}} 
\triangleq \left(
\sum\limits_{i \in \mathcal{A}}R_i, \, \hat{\mathcal{T}} 
\in \mathcal{P}(\mathcal{T})\right). \label{eq-def-RPT}
\end{equation} 
For any distribution $P_{\mathbf{X}_\mathcal{T}}$ 
defined on $\mathcal{X}_\mathcal{T}$ 
and any $\mathbf{x}_\mathcal{T} \in \mathcal{X}_\mathcal{T}$, 
define $\left(2^{|\mathcal{T}|}-1\right)$-dimensional vectors  
\begin{IEEEeqnarray}{rCl}
\overline{\boldsymbol{\imath}}_{\mathcal{T}}(\mathbf{x}_{\mathcal{T}}) 
&\triangleq& \left(\imath\big(
\mathbf{x}_{\hat{\mathcal{T}}}|\mathbf{x}_{\mathcal{T}\backslash\hat{\mathcal{T}}}
\big), \, \hat{\mathcal{T}} \in \mathcal{P}(\mathcal{T}) \right) \\
\overline{\mathbf{H}}_{\mathcal{T}} 
&\triangleq& \mathbb{E}\left[\overline{\boldsymbol{\imath}}_{\mathcal{T}}
(\mathbf{X}_{\mathcal{T}})\right], \label{eq-def-Hk}
\end{IEEEeqnarray} 
and $\left(2^{|\mathcal{T}|}-1\right) \times \left(2^{|\mathcal{T}|}-1\right)$ 
entropy dispersion matrix
\begin{equation}
\mathsf{V}_{\mathcal{T}} \triangleq 
{\rm Cov}\left[\overline{\boldsymbol{\imath}}_{\mathcal{T}}
(\mathbf{X}_{\mathcal{T}})\right] \label{eq-def-Vk}
\end{equation} 
for random vector $\mathbf{X}_\mathcal{T}$. 
Define set
\begin{equation}
\overline{\mathscr{R}}_{\mathcal{T}}^{*}(n,\epsilon) 
\triangleq \overline{\mathbf{H}}_{\mathcal{T}}
+ \frac{\mathscr{Q}_{\rm inv}(\mathsf{V}_{\mathcal{T}},\epsilon)}{\sqrt{n}} 
- \frac{\log n}{2n}\mathbf{1}. \label{eq-def-R3PT}
\end{equation} 
Thus, 
$\mathscr{R}_{\mathcal{T}}^{*}(n,\epsilon) \subset \mathbb{R}^{|\mathcal{T}|}$ 
while 
$\overline{\mathscr{R}}_{\mathcal{T}}^{*}(n,\epsilon) 
\subset \mathbb{R}^{2^{|\mathcal{T}|}-1}$.
Finally, 
\begin{IEEEeqnarray}{rCl}
\mathscr{R}_{{\rm in}, \mathcal{T}}^{*}(n,\epsilon)
&\triangleq& \bigg\{\mathbf{R}_\mathcal{T} \in \mathbb{R}^{|\mathcal{T}|}: 
\overline{\mathbf{R}}_{\mathcal{T}} 
\in \overline{\mathscr{R}}_{\mathcal{T}}^{*}(n,\epsilon) 
+ O\left(\frac{1}{n}\right)\mathbf{1} \bigg\} \nonumber \\ \label{eq-def-sw-in} \\ 
\mathscr{R}_{{\rm out}, \mathcal{T}}^{*}(n,\epsilon)
&\triangleq& \bigg\{\mathbf{R}_\mathcal{T} \in \mathbb{R}^{|\mathcal{T}|}: 
\overline{\mathbf{R}}_{\mathcal{T}} 
\in \overline{\mathscr{R}}_{\mathcal{T}}^{*}(n,\epsilon) 
- O\left(\frac{1}{n}\right)\mathbf{1} \bigg\}. \nonumber \\ \label{eq-def-sw-out} 
\end{IEEEeqnarray}
If every element of 
$\overline{\boldsymbol{\imath}}_{\mathcal{T}}(\mathbf{X}_{\mathcal{T}})$ 
has a positive variance and a finite third centered moment, 
then for any $0 < \epsilon < 1$, 
\begin{equation}
\mathscr{R}_{{\rm in}, \mathcal{T}}^{*}(n,\epsilon)
\subseteq \mathscr{R}_{\mathcal{T}}^{*}(n,\epsilon) \subseteq
\mathscr{R}_{{\rm out}, \mathcal{T}}^{*}(n,\epsilon).
\end{equation}
\end{remark}

\subsubsection{Comparison with Point-to-Point Source Coding}
\label{sec-sw-comp}
\begin{figure*}[!t]
	\centering
	\subfloat[]{\includegraphics[width=0.8\textwidth]{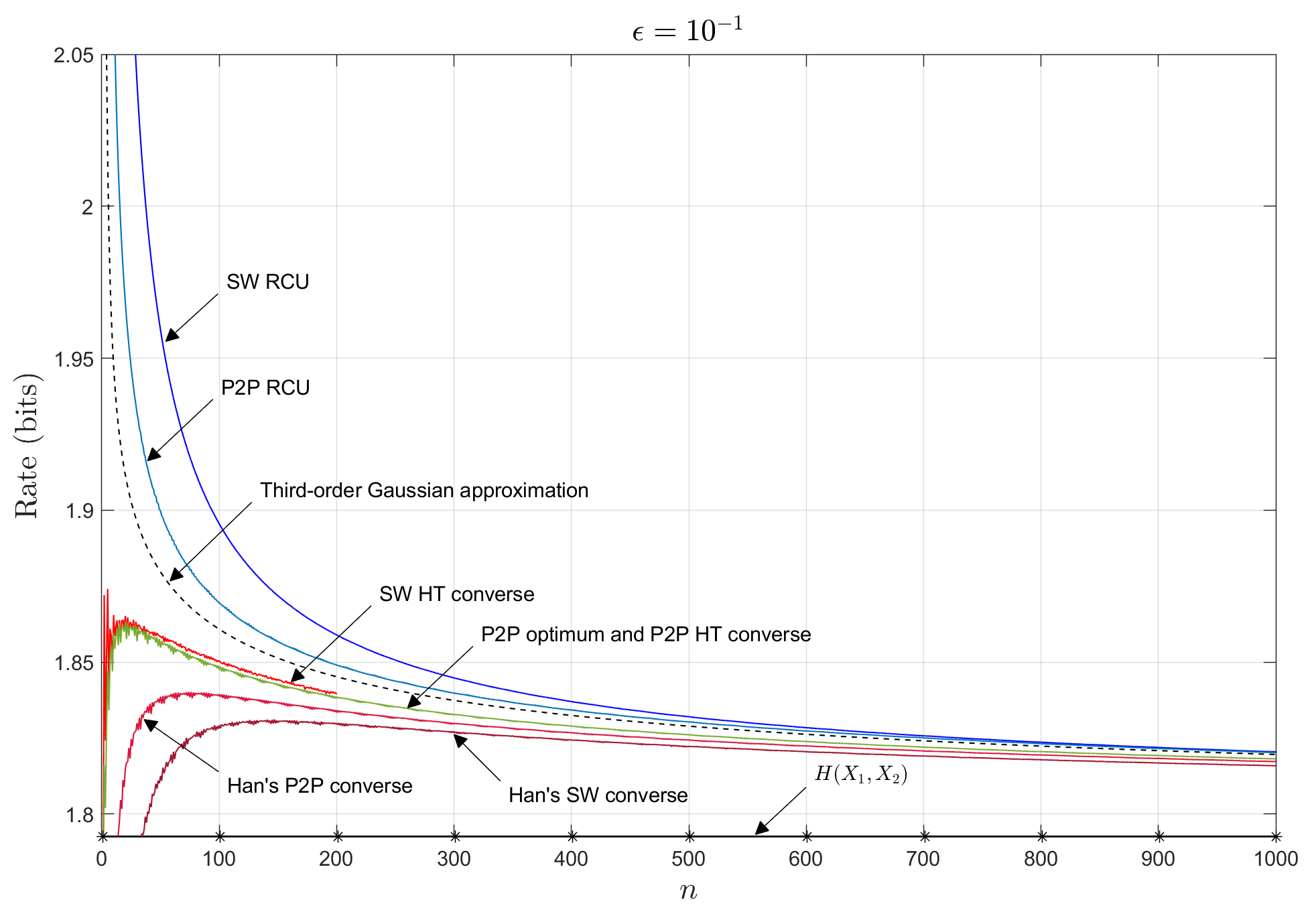}
		\label{fig-bounds-1}}
	\hfil
	\subfloat[]{\includegraphics[width=0.8\textwidth]{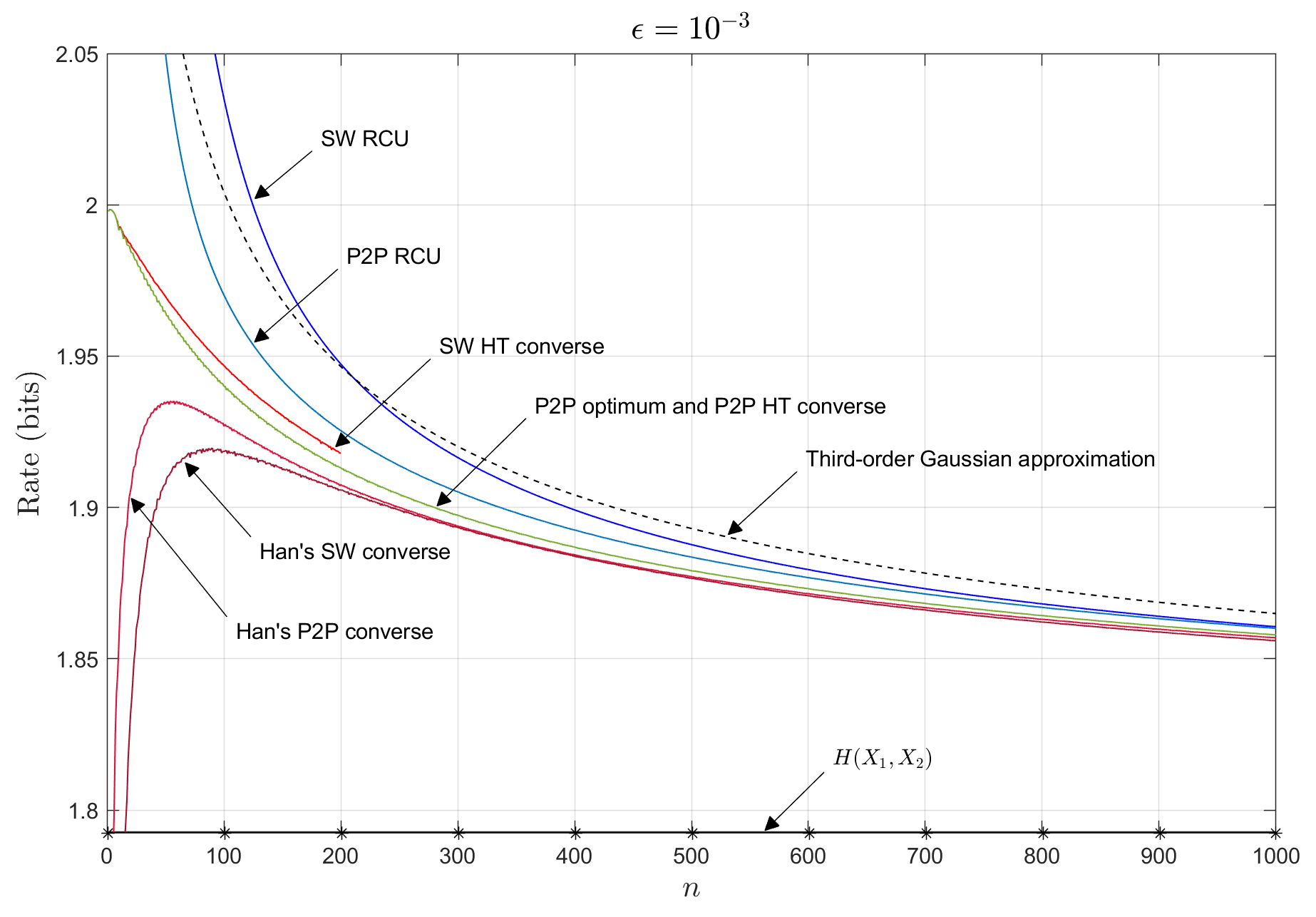}
		\label{fig-bounds-2}}
	\caption{Rate-blocklength trade-offs 
	at \protect\subref{fig-bounds-1} 	$\epsilon = 10^{-1}$ 
	and \protect\subref{fig-bounds-2} $\epsilon = 10^{-3}$ 
	for a pair of binary, stationary, memoryless sources 
	with joint distribution $p_{X_1X_2}(0,0)=1/2$, 
	$p_{X_1X_2}(0,1)=p_{X_1X_2}(1,0)=p_{X_1X_2}(1,1)=1/6$.  
	Due to computational limitations, 
	we only plot the MASC HT converse for small blocklengths ($n \leq 200$). 
	We evaluate the MASC HT converse with the sub-optimal choice of measures 
	in \eqref{eq-sw-cht-5}, 
	Han's point-to-point (P2P) converse is from \cite[Lemma~1.3.2]{han}, 
	Han's MASC converse is from Theorem~\ref{lem-han-conv} 
	(\cite[Lemma~7.2.2]{han}), and the P2P HT converse 
	is given in \cite[Appendix~A]{kostina-verdu}, 
	which coincides with the optimum $R^*(n, \epsilon)$.} 
	\label{fig-bounds}
\end{figure*}
Figure~\ref{fig-bounds} compares 
joint (point-to-point) compression of $(X_1^n,X_2^n)$ 
to the MASC sum-rate 
at the symmetrical rate point ($R_1 = R_2$). 
The gap between the MASC and point-to-point HT converses 
captures a penalty due to separate encoding. 
For small $n$, 
the third-order Gaussian approximation 
(without the $O\left(\frac{1}{n}\right)$ term) 
is more accurate at $\epsilon=10^{-1}$ than at $\epsilon=10^{-3}$ 
since the $O\left(\frac{1}{n}\right)$ term blows up as $\epsilon$ approaches $0$. 


It is well-known that optimal MASCs 
incur no first-order penalty in achievable sum rate 
when compared to joint coding~\cite{slepian-wolf, miyake-kanaya, han}. 
We next investigate the higher-order penalty 
of the MASC's independent encoders. 

Tan and Kosut introduce a quantity known as the 
\emph{local dispersion}~\cite[Def.~4]{tan-kosut} 
to characterize the second-order speed of convergence 
to any asymptotic MASC rate point from any direction. 
For any non-corner point on the diagonal boundary 
of the asymptotic MASC rate region,  
the sum rate's second-order coefficient is optimal 
when approached either vertically or horizontally.
Approaching corner points 
incurs a positive second-order penalty 
relative to point-to-point coding. 

Two corollaries of Theorem~\ref{thm-sw}, below, 
bound the MASC penalty 
by considering the achievable sum rate $R_1 + R_2$ 
for different choices of $R_1$ and $R_2$. 
We treat the cases where $X_1$ and $X_2$ are dependent 
and $X_1$ and $X_2$ are independent separately, 
assuming throughout 
that \eqref{assump-b1} and \eqref{assump-b2} hold.

When $X_1$ and $X_2$ are dependent, 
$H(X_1) + H(X_2) > H(X_1, X_2) > H(X_1|X_2) + H(X_2|X_1)$, 
and the asymptotic sum-rate boundary 
contains non-corner and corner points.
Corollary~\ref{cor-sw-dep}, below, shows that 
a MASC incurs no \mbox{first-,} \mbox{second-,} or third-order performance penalty 
relative to joint coding at non-corner points 
(i.e., when $R_1<H(X_1)$ and $R_2<H(X_2)$); 
in contrast, a MASC suffers a second-order performance penalty 
at corner points (i.e., when $R_1 = H(X_1)$ or $R_2 = H(X_2)$).  
See Figure~\ref{fig-sw-sumrate}\protect\subref{fig-sw-sumrate-1} for an illustration.

\begin{cor} \label{cor-sw-dep}
Suppose that $X_1$ and $X_2$ are dependent. 
\begin{enumerate}[leftmargin=1.3\parindent]
\item Fix constants $\delta_1,\delta_2,G>0$ and 
$\epsilon\in(0,1)$. 
Then there exists some constant $n(\delta_1,\delta_2,G)$ such that if 
\begin{IEEEeqnarray}{rCl}
R_1 &\leq& H(X_1) - \delta_1 \label{eq-sum-rate-1}\\
R_2 &\leq&  H(X_2) - \delta_2 \\
R_1 + R_2 &=& H(X_1,X_2) 
+ \sqrt{\frac{V(X_1,X_2)}{n}}Q^{-1}\left(\epsilon 
- \frac{G}{\sqrt{n}}\right) \nonumber \\
&&- \frac{\log n}{2n} \label{eq-sum-rate-2}
\end{IEEEeqnarray} 
then $\mathbf{R} = (R_1,R_2)\in\overline{\mathscr{R}}^{*}(n,\epsilon)$ 
for all $n > n(\delta_1,\delta_2,G)$.

\item Fix $\epsilon\in(0,1)$. If 
\begin{equation}
R_2 \geq H(X_2|X_1) + \frac{r^*}{\sqrt{n}} - \frac{\log n}{2n} 
+ \frac{G}{n} \label{eq:R2cor}
\end{equation} 
for some $G > 0$, then $\mathbf{R} = (H(X_1), R_2)\in\overline{\mathscr{R}}^{*}(n,\epsilon)$. 
Conversely, if $\mathbf{R} = (H(X_1), R_2)\in\overline{\mathscr{R}}^{*}(n,\epsilon)$, 
then 
\begin{equation}
R_2 \geq H(X_2|X_1) + \frac{r^*}{\sqrt{n}} - \frac{\log n}{2n}, \label{eq:R2corc}
\end{equation} 
where $r^*$ is the solution to equation
\begin{equation}
\Phi(\mathsf{V}_2; r, r) = 1 - \epsilon, \label{eq-def-r*}
\end{equation} 
and $\mathsf{V}_2$ is the covariance matrix 
for random vector $(\imath(X_2|X_1), \imath(X_1,X_2))$. 
\end{enumerate}
\end{cor}

\begin{proof}
Appendix \ref{append-cor-sw-dep}. 
\end{proof}


\begin{figure}[!t]
\centering
\subfloat[]{\includegraphics[width=0.48\textwidth]{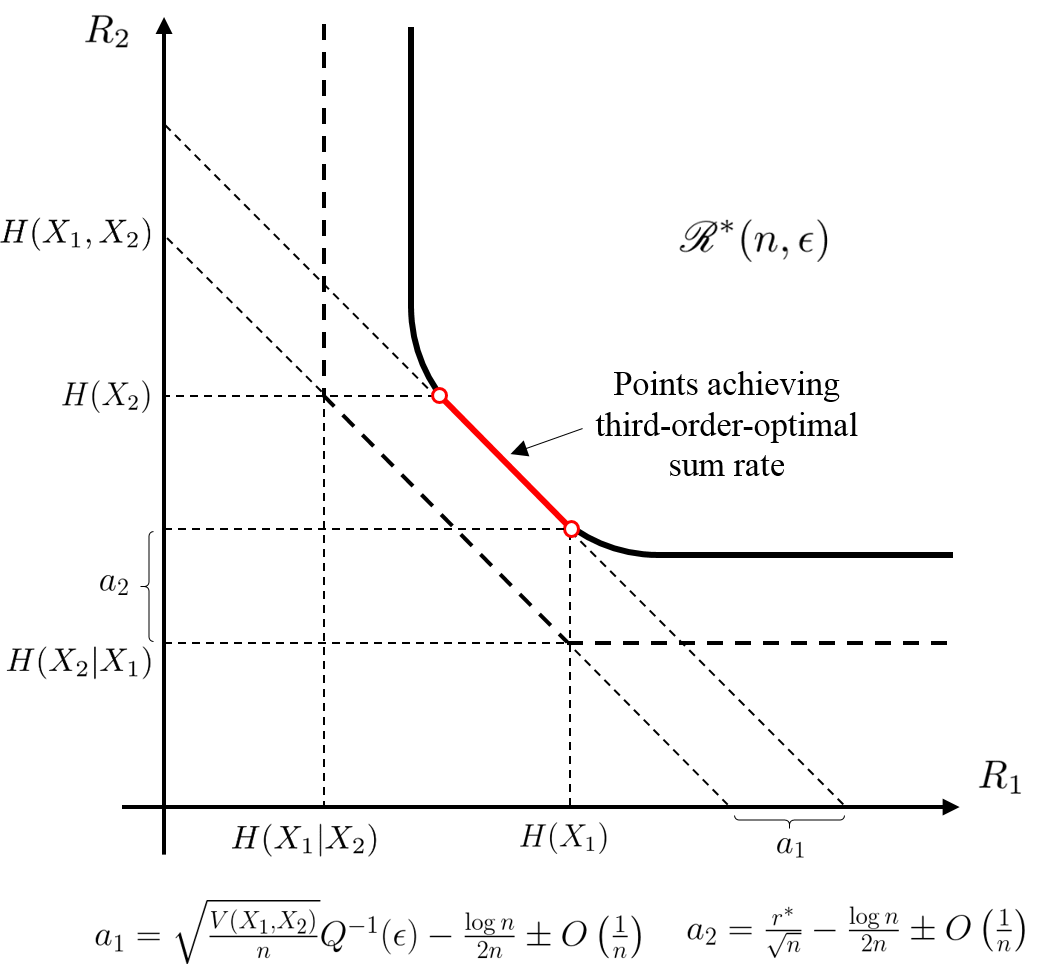}
\label{fig-sw-sumrate-1}}
\hfil
\subfloat[]{\includegraphics[width=0.46\textwidth]{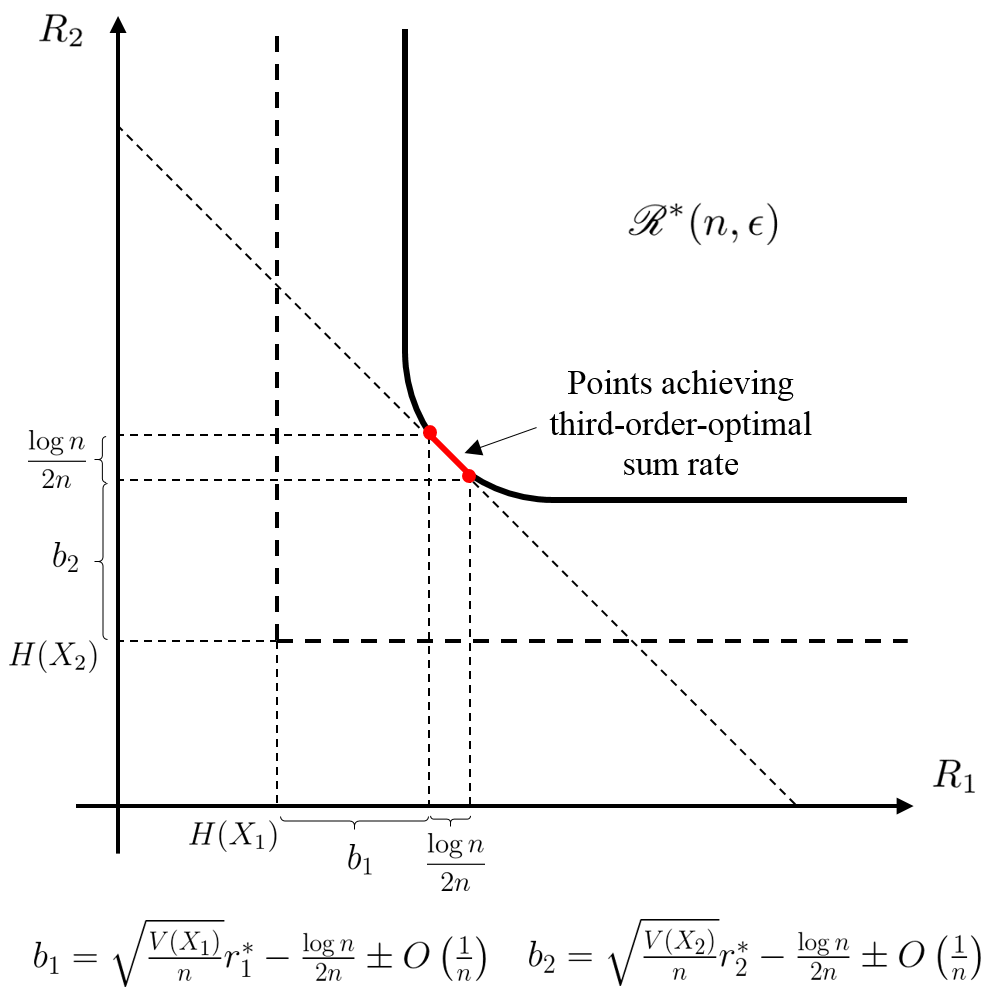}
\label{fig-sw-sumrate-2}}
\caption{Schematic plots of the $(n,\epsilon)$-rate region 
and the third-order-optimal sum rate when 
\protect\subref{fig-sw-sumrate-1} $X_1,X_2$ are dependent,  
\protect\subref{fig-sw-sumrate-2} $X_1,X_2$ are independent. 
In \protect\subref{fig-sw-sumrate-1}, the boundary of $\mathscr{R}^*(n,\epsilon)$ 
between $H(X_1)$ and $H(X_2)$ (excluding the end points) 
contains rate points that achieve the optimal point-to-point rate up to the third order, 
while the end points do not achieve that optimal rate. 
The value of $r^*$ in \protect\subref{fig-sw-sumrate-1} is defined in \eqref{eq-def-r*}; 
the values of $r_1^*, r_2^*$ in \protect\subref{fig-sw-sumrate-2} 
are defined in \eqref{eq-sw-sumrate-r1r2}.} 
\label{fig-sw-sumrate}
\end{figure}

For independent sources, 
the asymptotic sum-rate boundary contains only the single (corner) point 
$(R_1,R_2)=(H(X_1),H(X_2))$, 
and the entropy dispersion matrix
\begin{equation}
\begin{bmatrix}
V(X_1) &0 \\  0 &V(X_2)
\end{bmatrix}.\nonumber
\end{equation} 
is singular. 

The next result concerns the third-order-optimal sum rate 
\begin{align}
\overline{R}^*_{\rm sum}(n,\epsilon) \triangleq \min \Big\{R_1+R_2:
\ \exists \, &\mathbf{R} = (R_1, R_2) \notag\\
\text{ s.t. } &\overline{\mathbf{R}} \in \overline{\mathscr{R}}^{*}(n,\epsilon) \Big\}.
\end{align}
According to Theorem~\ref{thm-sw}, $\overline{R}^*_{\rm sum}(n,\epsilon)$ characterizes the best
achievable sum rate in SW source coding up to an $O \left( \frac 1 n \right)$
gap.

\begin{cor} \label{cor-sw-indep}
For $X_1,X_2$ independent and $\epsilon\in(0,1)$,  
\begin{IEEEeqnarray}{rCl}
\lefteqn{\overline{R}^*_{\rm sum}(n,\epsilon)}  \\ 
&=& H(X_1) + H(X_2) 
+ \frac{\sqrt{V(X_1)}r_1^* + \sqrt{V(X_2)}r_2^*}{\sqrt{n}} - \frac{\log n}{2n}, 
\nonumber 
\end{IEEEeqnarray} 
which is achieved by $\mathbf{R} = (R_1, R_2)$ with 
\begin{align}
R_1 &= H(X_1) + \sqrt{\frac{V(X_1)}{n}}r_1^* - \lambda\frac{\log n}{2n}  \\
R_2 &= H(X_2) + \sqrt{\frac{V(X_2)}{n}}r_2^* - (1-\lambda)\frac{\log n}{2n}, 
\end{align} 
for any $\lambda \in [0, 1]$ and 
\begin{IEEEeqnarray}{rCl}
(r_1^*, r_2^*) = \arg\min\limits_{\substack{{(r_1,r_2):} \\{\Phi(r_1)\Phi(r_2) \geq 1 - \epsilon} }} \left(\sqrt{V(X_1)}r_1 + \sqrt{V(X_2)}r_2\right). \label{eq-sw-sumrate-r1r2} \nonumber \\*
\end{IEEEeqnarray} 
\end{cor}
\begin{proof}
Appendix \ref{append-cor-sw-indep}.
\end{proof}

By Corollary~\ref{cor-sw-indep}, for independent sources 
a unique $(r_1^*,r_2^*)$ captures the best MASC second-order sum-rate;  
the third-order term 
is achieved at all points on a segment of the rate region boundary.
See Figure~\ref{fig-sw-sumrate}\protect\subref{fig-sw-sumrate-2}. 
Under assumption~\eqref{assump-b1}, 
\begin{IEEEeqnarray}{rCl}
\IEEEeqnarraymulticol{3}{l}{ \min\limits_{\substack{{(r_1,r_2):} \\
{\Phi(r_1)\Phi(r_2) \geq 1 - \epsilon} }} 
\left(\sqrt{V(X_1)}r_1 + \sqrt{V(X_2)}r_2\right)} \nonumber \\
&>& \sqrt{V(X_1) + V(X_2)}Q^{-1}(\epsilon), \label{eq-sw-indep-r1r2-1}
\end{IEEEeqnarray} 
where $V(X_1)+V(X_2)=V(X_1,X_2)$ for $(X_1,X_2)$ independent.
Here \eqref{eq-sw-indep-r1r2-1} 
follows since its left-hand side solves 
\begin{align}
\min\limits_{(a_1,a_2)} & \left(a_1 + a_2\right) \notag \\
\text{ s.t. }  \Phi\left(\frac{a_1}{\sqrt{V(X_1)}}\right)
&\Phi\left(\frac{a_2}{\sqrt{V(X_2)}}\right) \geq 1 - \epsilon, \label{eq-sw-indep-r1r2-2}
\end{align} 
and the constraint in \eqref{eq-sw-indep-r1r2-2} requires 
$a_1 > \sqrt{V(X_1)}Q^{-1}(\epsilon)$ and 
$a_2 > \sqrt{V(X_2)}Q^{-1}(\epsilon)$, 
which gives the bound since 
\begin{equation}
\sqrt{V(X_1)} + \sqrt{V(X_2)}
> \sqrt{V(X_1)+V(X_2)}.
\end{equation}
Therefore, when $X_1$ and $X_2$ are independent, 
a MASC incurs a positive second-order sum-rate penalty relative to joint coding.
Closed-form expressions for this penalty are available in special cases. 
When $V(X_1) = V(X_2)$, 
$r_1^* = r_2^* = Q^{-1}\left(1-\sqrt{1-\epsilon}\right)$,  
and the penalty is 
\begin{align}
2\sqrt{\frac{V(X_1)}{n}}Q^{-1}\left(1-\sqrt{1-\epsilon}\right) 
- \sqrt{\frac{2V(X_1)}{n}}Q^{-1}(\epsilon).
\end{align}
When $X_1$ and $X_2$ are i.i.d., 
the penalty equals the penalty for coding a vector $X^{2n}$ 
of $2n$ i.i.d. outputs from $P_X$ 
by applying an independent $(n,\epsilon)$ (point-to-point) code 
with error probability $1-\sqrt{1-\epsilon}$ 
to each of $(X_1, \ldots, X_n)$ and $(X_{n+1}, \ldots, X_{2n})$ 
instead of a single $(2n,\epsilon)$ code to vector $X^{2n}$.

\subsection{Limited Feedback and Cooperation}
\label{sec-sw-feedback}
The RASC proposed in Section~\ref{sec-rasc} employs limited feedback. 
We here analyze the impact of feedback on the underlying MASC. 
In our feedback model, 
the decoder broadcasts the same $\ell$ bits of feedback to both encoders.  
A bit sent at time $i$ must be a function 
of the encoder outputs received in time steps $1, \ldots, i - 1$. 
(See Figure~\ref{fig-fb-cf}\protect\subref{fig-fb}.)
We bound the impact of feedback by studying a MASC with 
a \emph{cooperation facilitator (CF).}\footnote{The CF 
is introduced for multiple access channel coding in \cite{noorzad-ho} 
and extended to source and network coding in~\cite{langberg}.}  
The CF broadcasts the same $\ell$-bit function of the sources 
to both encoders prior to their encoding operations. 
(See Figure~\ref{fig-fb-cf}\protect\subref{fig-cf}.)
Since the MASC network has no channel noise, 
feedback from the decoder cannot convey more information 
than feedback from the CF. 
As a result, we bound the impact of feedback 
by bounding the impact of cooperation, 
which is easier to work with in our analysis.  

We begin by defining the CF-MASC and its rate region. 

\begin{figure}[!t]
	\centering
	\subfloat[]{\includegraphics[width=0.45\textwidth]{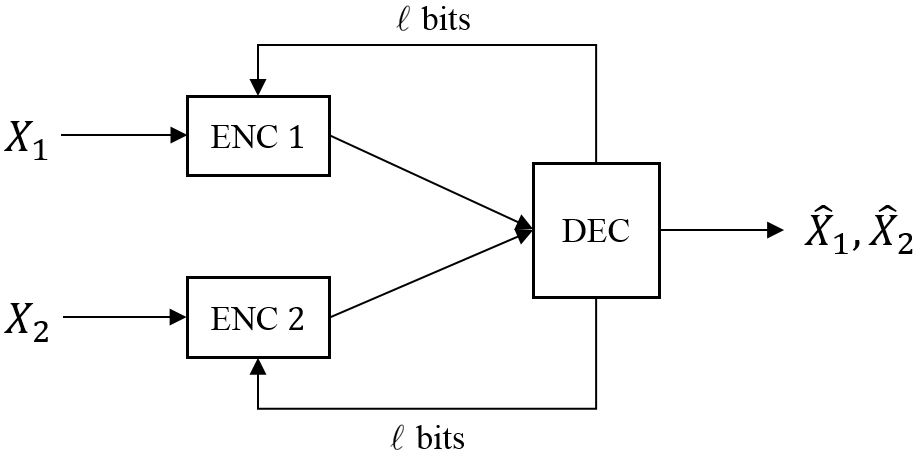}
		\label{fig-fb}}
	\hfil
	\subfloat[]{\includegraphics[width=0.45\textwidth]{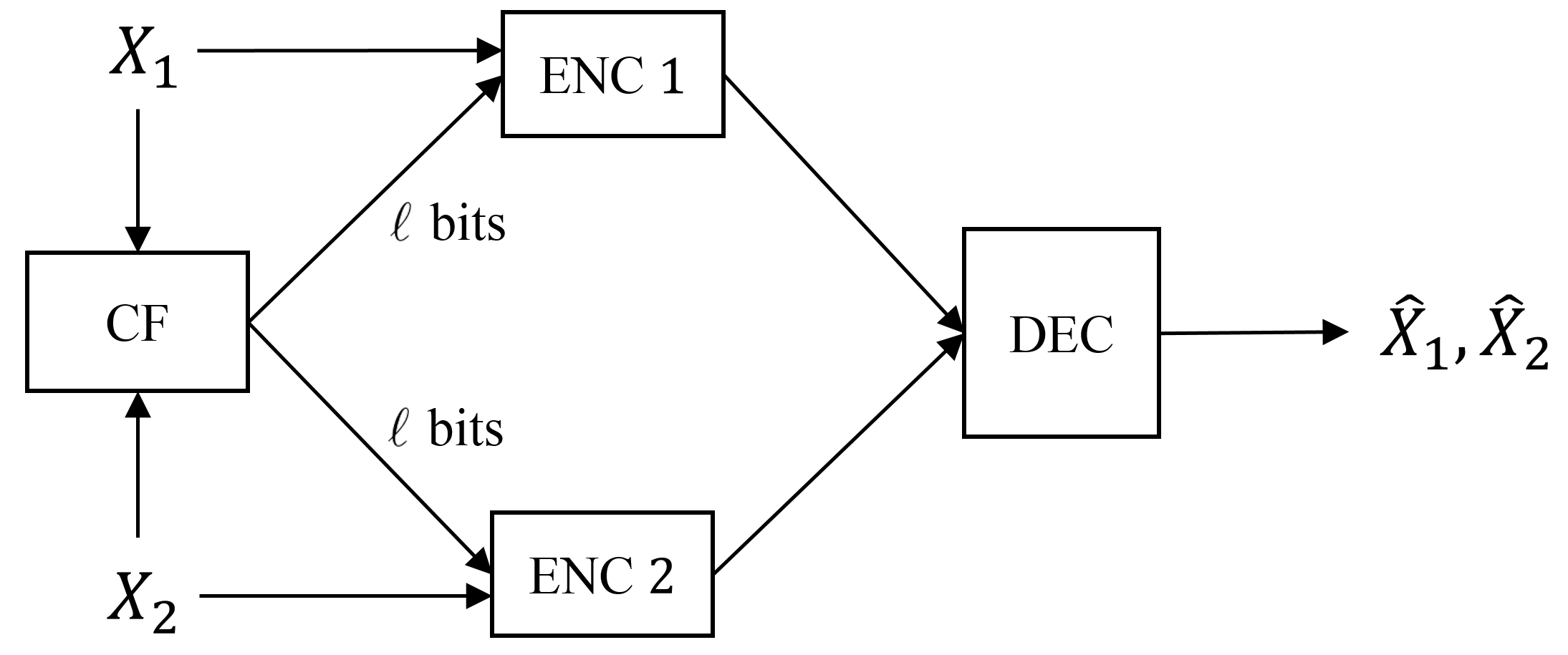}
		\label{fig-cf}}
	\caption{The \protect\subref{fig-fb} FB-MASC and \protect\subref{fig-cf} CF-MASC.}
	\label{fig-fb-cf}
\end{figure}

\begin{defn}[CF-MASC] \label{def-cs-sw} 
An $(L, M_1, M_2, \epsilon)$ CF-MASC 
for random variables $(X_1, X_2)$ on $\mathcal{X}_1 \times \mathcal{X}_2$ 
comprises a CF function $\mathsf{L}$, 
two encoding functions $\mathsf{f}_1$ and $\mathsf{f}_2$, 
and a decoding function $\mathsf{g}$ given by 
	\begin{IEEEeqnarray*}{rcl}
	\mathsf{L} & \ : \  &  \mathcal{X}_1 \times \mathcal{X}_2 \rightarrow [L] \\
	\mathsf{f}_1 & : & [L] \times \mathcal{X}_1 \rightarrow [M_1] \\
	\mathsf{f}_2 & : & [L] \times \mathcal{X}_2 \rightarrow [M_2] \\
	\mathsf{g} & : & [M_1]\times [M_2] \rightarrow \mathcal{X}_1 \times \mathcal{X}_2,
	\end{IEEEeqnarray*}
	with error probability 
	\begin{align*}
	\mathbb{P}\left[\mathsf{g} \left(\mathsf{f}_1\left(\mathsf{L}(X_1, X_2), X_1\right), \mathsf{f}_2\left(\mathsf{L}(X_1, X_2), X_2\right) \right) 
	\neq (X_1,X_2)\right] \leq \epsilon.
	\end{align*} 
\end{defn}

\begin{defn}[Block CF-MASC] \label{def-cf-sw-block}
An $(n, L, M_1, M_2, \epsilon)$ MASC 
is a CF-MASC for random variables $(X_1^n,X_2^n)$ 
on $\mathcal{X}_1^n \times \mathcal{X}_2^n$.
	
The code's finite blocklength rates are defined by 
	\begin{equation}
	R_1 = \frac{1}{n}\log M_1, \; R_2 = \frac{1}{n}\log M_2.
	\end{equation} 
\end{defn}

\begin{defn}[$(n,\ell,\epsilon)$-CF rate region] \label{def-cf-sw-rate-region} 
A rate pair $(R_1, R_2)$ is $(n, \ell,\epsilon)$-CF achievable 
if there exists an $(n,L,M_1,M_2,\epsilon)$ CF-MASC 
with $M_1 \leq \exp(nR_1)$, $M_2 \leq \exp(nR_2)$, and $L \leq \exp(\ell)$.
The $(n,\ell,\epsilon)$-CF rate region $\mathscr{R}_{\rm CF}^{*}(n,\ell,\epsilon)$ 
is defined as the closure of the set of all $(n,\ell,\epsilon)$-CF achievable rate pairs. 
\end{defn}

We use $\mathscr{R}_{\rm FB}^{*}(n,\ell,\epsilon)$ 
to denote the feedback-MASC (FB-MASC) rate region, 
which is defined as the closure of the set of all $(n,\epsilon)$-achievable rate pairs 
when the same $\ell$ bits of feedback from the decoder are available to both encoders.

Since the CF sees the source vectors 
while the decoder sees a coded description of those vectors 
(using a deterministic code),  
an $\ell$-bit CF can implement any function 
used to determine the decoder's $\ell$-bit feedback. 
As a result, any rate point that is achievable by an $\ell$-bit FB-MASC 
is also achievable by an $\ell$-bit CF-MASC. 
Therefore, for any $0 < \epsilon < 1$ and $\ell < \infty$,
\begin{equation}
\mathscr{R}_{\rm FB}^{*}(n,\ell,\epsilon) \subseteq \mathscr{R}_{\rm CF}^{*}(n,\ell,\epsilon). \label{eq-fb-cf-conv}
\end{equation}

Theorem~\ref{thm-cf-rasc-conv} bounds CF-MASC (and FB-MASC) performance, 
showing that for any $\ell < \infty$, 
the third-order rate region for $\ell$-bit CF-MASCs 
cannot exceed the corresponding MASC rate region. 
Hence finite feedback does not enlarge 
the third-order $(n,\epsilon)$ MASC rate region. 
This result generalizes to scenarios with more than two encoders.

\begin{thm}[CF-MASC Converse] \label{thm-cf-rasc-conv}
	Consider stationary, memoryless sources 
	with single-letter distribution $P_{X_1X_2}$ 
	satisfying \eqref{assump-b1} and \eqref{assump-b2}. 
	For any $0 < \epsilon < 1$ and $\ell < \infty$, 
	\begin{equation}
	\mathscr{R}_{\rm CF}^{*}(n,\ell,\epsilon)
	\subseteq \mathscr{R}_{\rm out}^{*}(n,\epsilon).  \label{eq-cf-rasc-conv}
	\end{equation} 
	Thus $\mathscr{R}_{\rm CF}^{*}(n,\ell,\epsilon)$ 
	and $\overline{\mathscr{R}}^{*}(n,\epsilon)$ 
	share the same outer bound.
\end{thm}

\begin{proof}
Appendix \ref{append-cf-rasc-conv}.
\end{proof}

\begin{remark}
The same proof can be used to 
show that allowing $\ell$ to grow as $o\left(\log \log n\right)$ 
does not change the first three terms 
in the optimal characterization of the $(n,\epsilon)$-MASC.
\end{remark}

\begin{remark} \label{rem-feedback-dep}
For dependent sources, 
the optimal third-order MASC sum rate equals 
the optimal third-order sum rate with full cooperation. 
(See the discussion in Section \ref{sec-sw-comp}, above.). 
Since even an infinite amount of decoder feedback 
is weaker than full cooperation, 
an infinite amount of feedback does not improve 
the third-order sum rate in this case.   
\end{remark}

\section{Random Access Source Code (RASC)}
\label{sec-rasc}

An RASC is a generalization of an MASC 
for networks where the set of participating encoders 
is unknown to both the encoders and the decoder~\emph{a priori}. 
We begin by defining the problem 
and describing our proposed communication strategy.

\subsection{Definitions and Coding Strategy}
\label{sec-def-ra}
Let $K < \infty$ be the maximal number of active encoders.
We associate each encoder with a unique source 
from the set of sources indexed by $[K]$. 
Each encoder chooses whether to be \emph{active} or \emph{silent}.  
Only sources associated with active encoders 
are compressed and reconstructed.  
By assumption, the decision to remain silent 
is independent of the observed source instance.
Given the joint distribution $P_{\mathbf{X}_{[K]}}$ 
on countable alphabet $\mathcal{X}_{[K]}$, 
when ordered set $\mathcal{T}\in \mathcal{P}([K])$ 
of $[K]$ is active, the marginal on the transmitted sources is 
	\begin{equation}
	P_{\mathbf{X}_\mathcal{T}}(\mathbf{x}_\mathcal{T}) = \sum\limits_{\mathbf{x}_{[K]\backslash \mathcal{T}} \in  \mathcal{X}_{[K]\backslash \mathcal{T}}}  P_{\mathbf{X}_{[K]}}(\mathbf{x}_{[K]}), \; \forall \, \mathbf{x}_\mathcal{T} \in \mathcal{X}_\mathcal{T}. \label{eq-rasc-red}
	\end{equation}
Thus, 
each encoder's state 
has no effect on the statistical relationship among 
sources observed by other encoders. 

As in the random access {\em channel code} from~\cite{recep}, 
our proposed RASC organizes communication into epochs. 
At the beginning of each epoch, 
each encoder independently decides its activity state; 
that activity state remains unchanged until the end of the epoch. 
Thus, the active encoder set $\mathcal{T}$ 
is fixed in each epoch. 
Each active encoder $i \in \mathcal{T}$ 
observes source output $X_i\in\mathcal{X}_i$ 
and independently maps it to a codeword 
comprised of a sequence of \emph{code symbols} from alphabet $[Q_i]$. 
The $|\mathcal{T}|$ codewords are sent simultaneously to the decoder.
Since set $\mathcal{T}$ is unknown {\em a priori}, 
the encoder behavior cannot vary with $\mathcal{T}$. 
The decoder sees $\mathcal{T}$ 
and decides a time $m_\mathcal{T}$, called the \emph{decoding blocklength}, 
at which to jointly decode all received partial codewords.
The set of potential decoding blocklengths 
$\mathcal{M} \triangleq \left\{m_{\mathcal{T}}: \mathcal{T} \in \mathcal{P}([K]) \right\}$ 
is part of the code design; it is known to all encoders and to the decoder.

\begin{figure}[!t]
	\centering
	\includegraphics[width=0.46\textwidth]{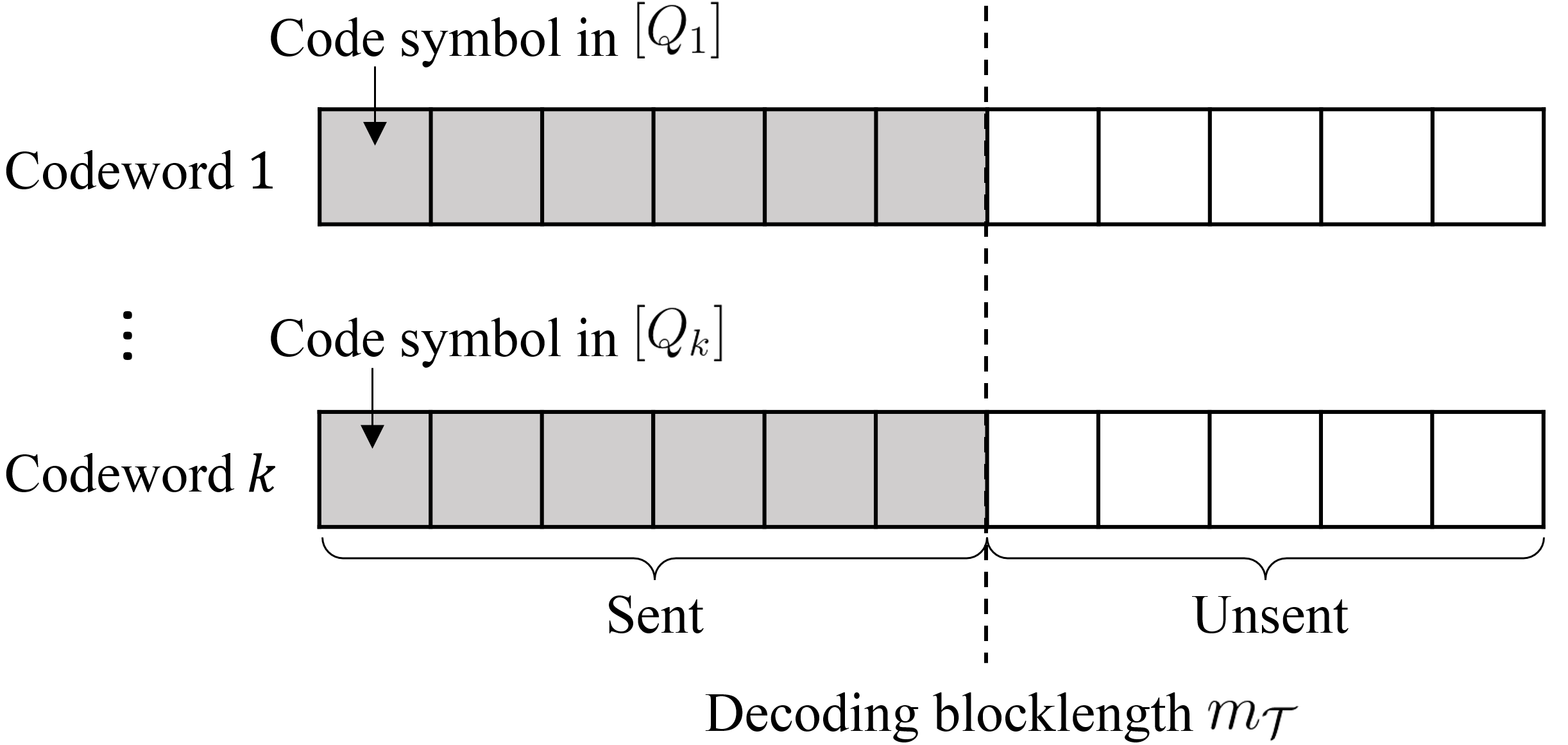}
	\caption{Coding scheme in one epoch with $\mathcal{T} = [k]$.}
	\label{fig-rasc-code}
\end{figure} 

Figure~\ref{fig-rasc-code} illustrates our coding scheme in one epoch 
when $\mathcal{T}=[k]$. 
Each encoder $i\in\mathcal{T}$ sends a single code symbol per time step.  
At each time $m\in\left\{m' \in \mathcal{M}: m' <m_{\mathcal{T}} \right\}$, 
the decoder sends a ``0'' to indicate that it is not yet ready to decode; 
at time $m=m_\mathcal{T}$, the decoder sends a ``1,'' 
ending one epoch and starting the next.
The decoder then reconstructs source vector $\mathbf{X}_{\mathcal{T}}$ 
using the first $m_{\mathcal{T}}$ code symbols from each active encoder.
To avoid wasting time in an epoch with no active encoders, 
we include decoding time $m_{\emptyset}=1$ in set $\mathcal{M}$.
The decoder sends at most $2^K$ bits of feedback, 
and encoders need only listen for decoder feedback 
at the times in set $\mathcal{M}$. 

To formalize the above strategy, fix $K\geq 1$.  Define vectors 
\begin{IEEEeqnarray}{rCl}
\overline{\boldsymbol{\epsilon}}_{K} 
&\triangleq& \left( \epsilon_{\mathcal{T}}, \,\mathcal{T} \in \mathcal{P}([K]) \right) \\
\overline{\mathbf{m}}_{K} 
&\triangleq& \left(m_{\mathcal{T}}, \,\mathcal{T} \in \mathcal{P}([K])\cup\{\emptyset\} \right)
\end{IEEEeqnarray} 
with $m_{\emptyset}=1$ and 
$m_{\max} \triangleq \max \left\{m_{\mathcal{T}}:\mathcal{T} \in \mathcal{P}([K])\right\}$.

\begin{defn}[RASC] \label{def-rasc} 
An $\left(\overline{\mathbf{m}}_{K}, 
\mathbf{Q}_{[K]}, \overline{\boldsymbol{\epsilon}}_{K}\right)$ 
RASC for sources $X_{[K]}$ on source alphabet $\mathcal{X}_{[K]}$ 
comprises a collection of encoding and decoding functions 
	\begin{IEEEeqnarray}{rcl}
	\mathsf{f}_i &\ :\ &  \mathcal{X}_i \rightarrow [Q_i]^{m_{\max}}, \, i \in [K], \\
	\mathsf{g}_{\mathcal{T}} & : &  \prod\limits_{i \in \mathcal{T}} [Q_i]^{m_{\mathcal{T}}}
	 \rightarrow \mathcal{X}_{\mathcal{T}}, \, \mathcal{T} \in \mathcal{P}([K]),
	\end{IEEEeqnarray}
	where $\mathsf{f}_i$ is the encoding function for source $X_i$ 
	and $\mathsf{g}_{\mathcal{T}}$ is the decoding function 
	for active coder set $\mathcal{T}$.
	For each $\mathcal{T} \in \mathcal{P}([K])$, 
	source vector $\mathbf{{X}}_{\mathcal{T}}$ 
	is decoded at time $m_{\mathcal{T}}$ with error probability 
	$\mathbb{P}\big[\mathsf{g}_{\mathcal{T}}\left(\mathsf{f}_i(X_i)_{[m_\mathcal{T}]}, 
	i \in \mathcal{T} \right) \neq \mathbf{{X}}_{\mathcal{T}}\big]
	\leq \epsilon_{\mathcal{T}}$, 
	where $\mathsf{f}_i(x_i)_{[m]}$ denotes the first $m$ code symbols 
	of $\mathsf{f}_i(x_i)$.
\end{defn}

Definition~\ref{def-rasc-block} particularizes Definition \ref{def-rasc} to the block setting.

\begin{defn}[Block RASC]\label{def-rasc-block}
	An $\left(n, \overline{\mathbf{m}}_{K}, \mathbf{Q}_{[K]},
	\overline{\boldsymbol{\epsilon}}_{K}\right)$ RASC 
	is an RASC for an $n$-block of source outcomes. 
	The parameter $n$, called the \emph{encoding blocklength} 
	does not vary with $\mathcal{T}$.
\end{defn}

An $\left(\overline{\mathbf{m}}_{K}, \mathbf{Q}_{[K]}, \overline{\boldsymbol{\epsilon}}_{K}
\right)$ RASC behaves, for each $\mathcal{T}$, like a 
$\big((Q_i^{m_{\mathcal{T}}},\, i \in \mathcal{T}), \epsilon_{\mathcal{T}}\big)$ 
MASC (see Definition~\ref{def-sw}) 
with a finite number $\left|\left\{m \in \mathcal{M}: m \leq m_{\mathcal{T}}\right\}\right|$ 
of feedback bits. 
However, the RASC is \emph{one code}.
Its descriptions are nested (i.e., for each $x_i \in \mathcal{X}_i$, 
if $m_{\mathcal{T}'} < m_{\mathcal{T}}$, then 
$\mathsf{f}_i(x_i)_{[m_{\mathcal{T}'}]}$ is a prefix of 
$\mathsf{f}_i(x_i)_{[m_{\mathcal{T}}]}$).
It simultaneously satisfies the error constraints 
for all $\mathcal{T} \in \mathcal{P}([K])$.
And, since the code symbol alphabet sizes $\mathbf{Q}_{[K]}$ are fixed, 
its rate vectors are coupled.  See Figure~\ref{fig-rasc-rate}.

\begin{figure}[!t]
	\centering
	\includegraphics[width=0.35\textwidth]{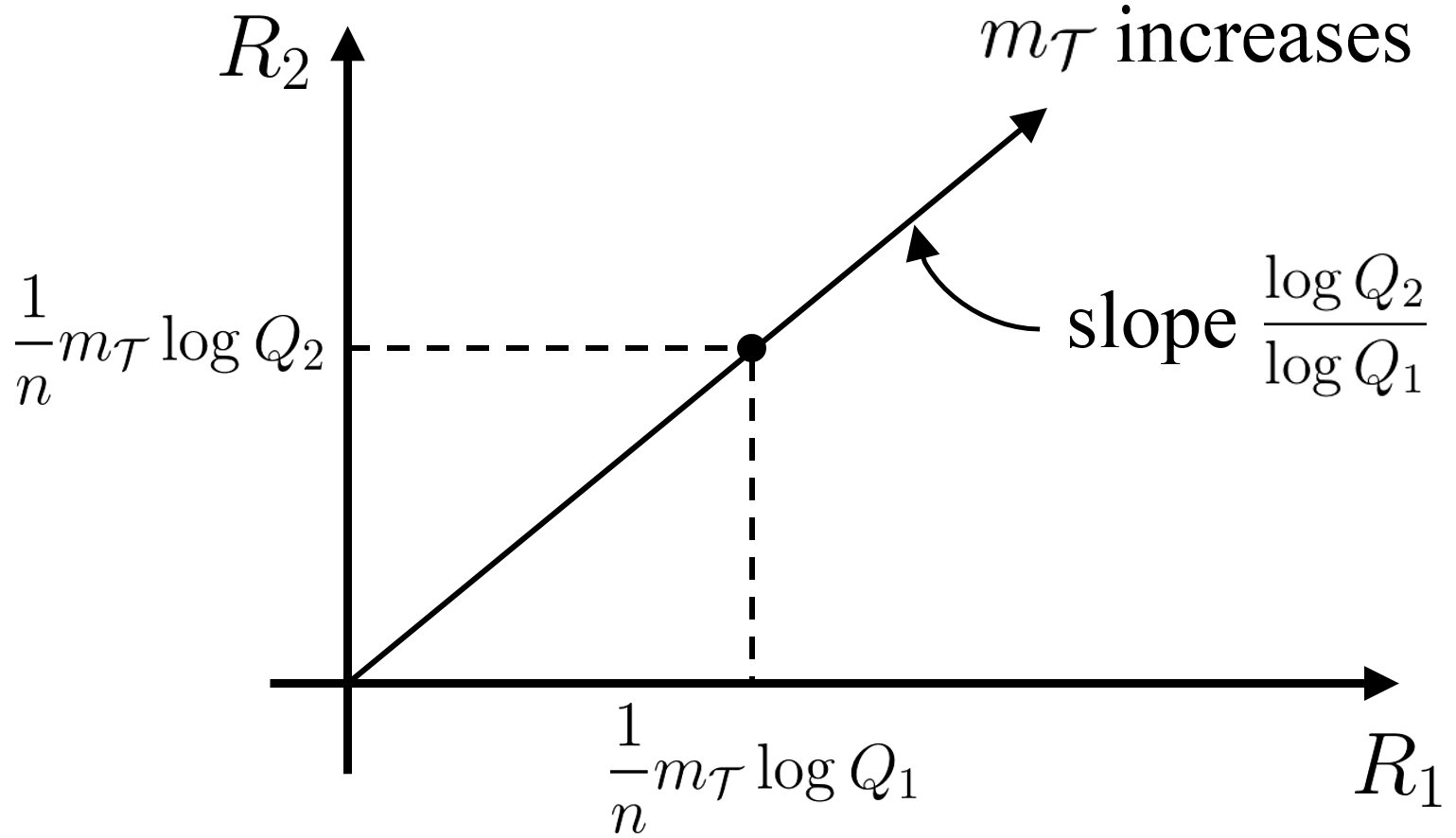}
	\caption{The relationship between decoding blocklength $m_{\mathcal{T}}$, 
	code symbol alphabet sizes $(Q_1, Q_2)$, 
	and source coding rate vector $\mathbf{R}_{\mathcal{T}}$, illustrated for $\mathcal{T} =\left\{1, 2\right\}$.}
	\label{fig-rasc-rate}
\end{figure} 

The following definitions build toward 
the non-asymptotic fundamental limit of RASCs.

\begin{defn}[$n$-Valid and $\left(n, \overline{\boldsymbol{\epsilon}}_{K}\right)$-Rate 
	sets] \label{def-valid-rate}
	A collection $\left(\mathbf{R}_{\mathcal{T}}\right)_{\mathcal{T} \in \mathcal{P}([K])}$ 
	of rate vectors is $n$-valid if $\exists$ 
	$\left(\overline{\mathbf{m}}_{K}, \mathbf{Q}_{[K]}\right)$ s.t.
	\begin{equation}
	\mathbf{R}_\mathcal{T} = \frac{1}{n}\left(m_\mathcal{T}\log Q_i, \, i \in \mathcal{T}\right),\, \forall \, \mathcal{T} \in \mathcal{P}([K]). \label{eq-def-rate-vector}
	\end{equation}
	The set $\mathcal{R}_{\rm valid}(n)$ is the set of $n$-valid rate collections.
	The collection is $\left(n, \overline{\boldsymbol{\epsilon}}_{K}\right)$-achievable 
	if there exists an $\left(n,\overline{\mathbf{m}}_{K}, \mathbf{Q}_{[K]}, 
	\overline{\boldsymbol{\epsilon}}_{K}\right)$ RASC. 
	The $\left(n, \overline{\boldsymbol{\epsilon}}_{K}\right)$-rate set 
	$\mathcal{R}^{*}\left(n, \overline{\boldsymbol{\epsilon}}_{K}\right)$ 
	is the set of $\left(n,\overline{\boldsymbol{\epsilon}}_{K}\right)$-achievable 
	rate collections.
\end{defn}

\subsection{Background} \label{sec-ra-b}
While the concept of an RASC is new, 
the RASC problem is related to the universal MASC problem. 
Like a universal MASC, 
the RASC is designed for an unknown distribution 
from a known collection of possible distributions. 
In this case, the possible distributions are 
$\left\{P_{\mathbf{X}_\mathcal{T}}: \mathcal{T} \in \mathcal{P}([K]) \right\}$.
The RASC differs, however, 
from universal MASCs 
since even the set of active encoders is unknown \emph{a priori}. 

A short summary of prior universal MASCs follows.
\begin{enumerate}
	\item For a fixed-rate MASC and finite source alphabets, 
	universal decoding can be realized using type methods. 
	(See \cite{tan-kosut}, \cite{oohama-han}, \cite{csiszar-korner}.) 
	Such strategies achieve optimal performance 
	only when the source's MASC rate region 
	matches the code's fixed rate. 
	\item Oohama~\cite{oohama} and Jaggi and Effros~\cite{jaggi-effros} 
	study the effect of limited encoder cooperation 
	on the asymptotically universally achievable rate region. 
	Rate-zero cooperation between encoders 
	suffices to achieve universality in the asymptotic regime. 
	Oohama characterizes the optimal error exponents in~\cite{oohama}.
	\item Yang et al.~\cite{yang-yeung} study a 
	block MASC with progressive encoding; 
	the code uses zero-rate feedback 
	to universally achieve the asymptotic MASC rate region.
	Sarvotham et al.~\cite{sarvotham-baraniuk} 
	propose a variable-rate block sequential coding scheme 
	with zero-rate feedback for binary symmetric sources, 
	showing that at blocklength $n$ and target error probability $\epsilon$, 
	the backoff from the asymptotic MASC rate due to universality 
	is $O\left(\frac 1 {\sqrt{n}}Q^{-1}(\epsilon)\right)$. 
	\item In~\cite{draper}, Draper introduces a rateless MASC 
	with single-bit feedback. 
	Draper's algorithm asymptotically achieves the optimal coding rates 
	for sources with unknown joint distributions but known finite alphabet sizes. 
	See~\cite{eckford-yu} for a practical rateless MASC. 
\end{enumerate}

\subsection{Asymptotics: Third-Order Performance of the RASC}
\label{sec-result-rasc}

In this section, we analyze the performance of an 
$\left(n, \overline{\mathbf{m}}_{K}, 
\mathbf{Q}_{[K]}, \overline{\boldsymbol{\epsilon}}_{K}\right)$ RASC 
for stationary, memoryless sources. 
Results include both achievability and converse characterizations 
of the $\left(n, \overline{\boldsymbol{\epsilon}}_{K}\right)$-rate set 
$\mathcal{R}^{*}\left(n, \overline{\boldsymbol{\epsilon}}_{K}\right)$
under the assumption that 
the single-letter joint source distribution $P_{\mathbf{X}_{[K]}}$ satisfies
\begin{IEEEeqnarray}{rCl}
&\mathbb E\left[V_c\big(\mathbf{X}_{\hat{\mathcal{T}}} |
 \mathbf{X}_{\mathcal{T}\backslash \hat{\mathcal{T}}}\big)\right] > 0 \ \ \ \forall \, 
 \hat{\mathcal{T}} \subseteq \mathcal{T} \subseteq [K], \, \hat{\mathcal{T}}, 
 \mathcal{T} \neq \emptyset&  \label{assump-c1} \\
&T\big(\mathbf{X}_{\hat{\mathcal{T}}} | 
\mathbf{X}_{\mathcal{T}\backslash \hat{\mathcal{T}}}\big) < \infty \ \ \ \ \ \ \ \ \forall \, 
\hat{\mathcal{T}} \subseteq \mathcal{T} \subseteq [K], \, \hat{\mathcal{T}}, 
\mathcal{T} \neq \emptyset. & \label{assump-c2} \\
&\mathbb E\left[T_c^2\big(\mathbf{X}_{\hat{\mathcal{T}}} |
 \mathbf{X}_{\mathcal{T}\backslash \hat{\mathcal{T}}}\big)\right] < \infty \ \ \forall \, 
 \hat{\mathcal{T}} \subset \mathcal{T} \subseteq [K], \, \hat{\mathcal{T}}, 
 \mathcal{T} \neq \emptyset.&  \label{assump-c3} 
 \IEEEeqnarraynumspace
\end{IEEEeqnarray}
Constraints \eqref{assump-c1}--\eqref{assump-c3} 
enable us to use Berry-Esseen bounds.  
The resulting characterization is tight up to the third-order term. 
While the existence of an 
$\left(n,\overline{\mathbf{m}}_{K}, \mathbf{Q}_{[K]}, 
\overline{\boldsymbol{\epsilon}}_{K}\right)$ RASC 
implies the existence of an 
$\big(n,(Q_i^{m_{\mathcal{T}}},\, i \in \mathcal{T}), \epsilon_{\mathcal{T}}\big)$ MASC 
for each $\mathcal{T} \in \mathcal{P}([K])$, 
the existence of individual MASCs 
does not imply the existence of a single RASC 
that simultaneously satisfies the error probability constraints 
for all possible configurations of active encoders. 
Indeed, the existence of a single RASC 
that simultaneously performs as well (up to the third-order term) 
as the optimal MASC for each $\mathcal{T} \in \mathcal{P}([K])$ 
is one of our most surprising results.

Define the inner and outer bounding sets
\begin{IEEEeqnarray}{rCl}
\lefteqn{\mathcal{R}_{\rm in}^{*}\left(n, \overline{\boldsymbol{\epsilon}}_{K}\right) 
\triangleq \Big\{\left(\mathbf{R}_{\mathcal{T}}\right)_{\mathcal{T} \in \mathcal{P}([K])} 
\in \mathcal{R}_{\rm valid}(n):} \nonumber \\
&&\qquad \qquad \mathbf{R}_\mathcal{T} \in \mathscr{R}_{{\rm in}, 
\mathcal{T}}^{*}(n,\epsilon_{\mathcal{T}}) \ \ \forall\, 
\mathcal{T} \in \mathcal{P}([K]) \Big\} \label{eq-rasc-in} \\
\lefteqn{\mathcal{R}_{\rm out}^{*}\left(n, \overline{\boldsymbol{\epsilon}}_{K}\right) 
\triangleq \Big\{\left(\mathbf{R}_{\mathcal{T}}\right)_{\mathcal{T} \in \mathcal{P}([K])} 
\in \mathcal{R}_{\rm valid}(n):} \nonumber \\
&&\qquad \qquad \mathbf{R}_\mathcal{T} \in \mathscr{R}_{{\rm out}, 
\mathcal{T}}^{*}(n,\epsilon_{\mathcal{T}}) \ \ \forall\, \mathcal{T} \in \mathcal{P}([K]) 
\Big\}, \label{eq-rasc-out} \IEEEeqnarraynumspace
\end{IEEEeqnarray} 
where $\mathscr{R}_{{\rm in}, \mathcal{T}}^{*}(n,\epsilon)$ 
and $\mathscr{R}_{{\rm out}, \mathcal{T}}^{*}(n,\epsilon)$ 
are the third-order MASC bounding sets for distribution $P_{\mathbf{X}_\mathcal{T}}$. 
(See \eqref{eq-def-sw-in} and \eqref{eq-def-sw-out}.)

\begin{thm}[Third-order RASC performance] \label{thm-rasc}
	For any $K < \infty$, consider stationary, memoryless sources 
	specified by a single-letter joint distribution $P_{\mathbf{X}_{[K]}}$ 
	satisfying~\eqref{assump-c1}--\eqref{assump-c3}. 
	For any $\mathbf{0} < \overline{\boldsymbol{\epsilon}}_{K} < \mathbf{1}$, 
	\begin{equation}
	\mathcal{R}_{\rm in}^{*}\left(n, \overline{\boldsymbol{\epsilon}}_{K}\right) 
	\subseteq \mathcal{R}^{*}\left(n, \overline{\boldsymbol{\epsilon}}_{K}\right) 
	\subseteq \mathcal{R}_{\rm out}^{*}\left(n, \overline{\boldsymbol{\epsilon}}_{K}\right).
	\end{equation}
\end{thm}

The converse and achievability proofs follow.

\begin{proof}[Proof of Theorem \ref{thm-rasc}: converse]
As shown in Section~\ref{sec-sw-feedback} (Theorem~\ref{thm-cf-rasc-conv}), 
even with {\em a priori} knowledge of the encoder set 
$\mathcal{T}\in\mathcal{P}([K])$ 
and $2^K$ bits of feedback, 
a MASC for the encoders in set $\mathcal{T}$ 
cannot achieve performance outside of 
the third-order MASC outer bounding set 
$\mathscr{R}_{{\rm out},\mathcal{T}}^{*}(n,\epsilon_\mathcal{T})$.
\end{proof}
The achievability part of Theorem~\ref{thm-rasc} 
provides a sufficient condition 
for the existence of a \emph{single} RASC 
that is simultaneously good for all $\mathcal{T} \in \mathcal{P}([K])$. 
To prove this, 
we first derive an achievability result 
assuming that the encoders and decoder share the common randomness 
used to generate a random code (Theorem~\ref{thm-rasc-achiev-bound}). 
Unfortunately, 
the existence of a random code ensemble 
with expected error probability satisfying the error probability constraint 
for each $\mathcal{T}\in\mathcal{P}([K])$ 
does not guarantee the existence of a {\em single} deterministic code 
satisfying those constraints simultaneously. 
We therefore take a different approach, 
which, unexpectedly, combines a converse bound on error probability 
and a random coding argument to show achievability.


The following refinement of the random coding argument 
provides a bound on the probability (with respect to the random code choice) 
that the error probability of a randomly chosen code 
exceeds a certain threshold. 
The code of interest here can be any type of source or channel code.
\begin{lem} \label{lem-bad-code}
	Let $\mathcal{C}$ be any class of codes 
	with a corresponding error probability $P_e(\mathsf{c})$ 
	for each $\mathsf{c} \in \mathcal{C}$. 
	Let 
	\begin{equation}
	\epsilon^{*}(\mathcal{C}) = \min\limits_{\mathsf{c} \in \mathcal{C}} P_e(\mathsf{c})
	\end{equation} 
	denote the error probability of the best code in $\mathcal{C}$. 
	Then any random code ensemble\footnote{A random code ensemble 
	is a random variable $\mathsf{C}$ 
	defined on code set $\mathcal{C}$.} 
	$\mathsf{C}$ defined over $\mathcal{C}$ satisfies 
	\begin{equation}
	\mathbb{P}\left[P_e(\mathsf{C}) > \epsilon \right] 
	\leq \frac{\mathbb{E}\left[P_e(\mathsf{C})\right] - \epsilon^{*}(\mathcal{C})}
	{\epsilon - \epsilon^{*}(\mathcal{C})}, \; \forall \, \epsilon > \epsilon^{*}(\mathcal{C}). 
	\label{eq-rand-lem}
	\end{equation}
\end{lem}

\begin{proof}
	Let $Y$ be any non-negative random variable 
	and define $y_{\min} \triangleq {\rm ess} \inf Y$; 
	that is, $y_{\min}$ is the largest constant $y \in \mathcal{Y}$ 
	for which $Y \geq y$ almost surely. By Markov's inequality, 
	\begin{IEEEeqnarray*}{rcl}
	\mathbb{P}\left[Y \geq y \right] 
	&\,=\,& \mathbb{P}\left[Y - y_{\rm min} \geq y - y_{\rm min} \right] \\
	&\,\leq\,& \frac{\mathbb{E}\left[Y\right] - y_{\min}}{y - y_{\min}}\ \ 
	\forall \, y > y_{\min}.
	\end{IEEEeqnarray*}
	Taking $Y = P_e(\mathsf{C})$ and $y = \epsilon$ 
	yields the desired result.
\end{proof}
In the regime of interest $\mathbb{E}\left[P_e(\mathsf{C})\right] < \epsilon$.
Therefore, the right side of \eqref{eq-rand-lem} is decreasing as a function of $\epsilon^{*}(\mathcal{C})$, and replacing  $\epsilon^{*}(\mathcal{C})$ by any converse on $\epsilon^{*}(\mathcal{C})$ yields a valid achievability bound. Thus Lemma~\ref{lem-bad-code} provides a means to leverage a converse to prove achievability. 

Given any RASC $\mathsf{c}$, for each $\mathcal{T} \in \mathcal{P}([K])$
let $P_{e,\mathcal{T}}(\mathsf{c})$ denote the error probability of code $\mathsf{c}$ 
under active encoder set $\mathcal{T}$. 
The RASC achievability proof 
applies Lemma~\ref{lem-bad-code} 
with error probability $P_{e,\mathcal{T}}(\mathsf{c})$ 
for each $\mathcal{T}\in\mathcal{P}([K])$. 
Before proceeding to that proof, 
we use Theorem~\ref{thm-rasc-achiev-bound}, below, 
to define a random code ensemble 
and calculate its expected error probability.

\begin{thm}[Random code] \label{thm-rasc-achiev-bound}
	For any $K < \infty$, 
	consider a source distribution $P_{\mathbf{X}_{[K]}}$ 
	defined on countable alphabet $\mathcal{X}_{[K]}$. 
	There exists a random code ensemble $\mathsf{C}$ 
	defined on the set of all RASCs 
	with decoding blocklengths $\overline{\mathbf{m}}_{K}$ 
	and code alphabets $\mathbf{Q}_{[K]}$ 
	for which the following inequalities hold simultaneously 
	for all $\mathcal{T} \in \mathcal{P}([K])$:
	\begin{IEEEeqnarray}{rCl}
	\mathbb{E}\left[P_{e,\mathcal{T}}(\mathsf{C})\right] 
	&\leq& \mathbb{E} \Bigg[\min \Bigg\{1, \, 
	\sum\limits_{{\hat{\mathcal{T}}} \in \mathcal{P}(\mathcal{T})} 
	\exp\big(-m_{\mathcal{T}}\cdot \overline{Q}({\hat{\mathcal{T}}}) \big)   A_{\hat {\mathcal{T}}} 
	\Bigg\} \Bigg], \label{eq-rasc-achiev-b} \nonumber \\*
	\end{IEEEeqnarray} 
	where 

\begin{IEEEeqnarray}{rCl}
	\overline{Q}({\hat{\mathcal{T}}}) & \triangleq & 
	\sum\limits_{i \in {\hat{\mathcal{T}}}} \log Q_i \\
 A_{\hat {\mathcal{T}}} &\triangleq& 
 	\mathbb{E} \big[\exp\big(\imath\big(\bar{\mathbf{X}}_{{\hat{\mathcal{T}}}}|
	\mathbf{X}_{\mathcal{T} \backslash {\hat{\mathcal{T}}}}\big)\big) \cdot \nonumber \\ 
	&&1\big\{\imath\big(\bar{\mathbf{X}}_{{\hat{\mathcal{T}}}}|
	\mathbf{X}_{\mathcal{T} \backslash {\hat{\mathcal{T}}}}\big) \leq 
	\imath\big(\mathbf{X}_{{\hat{\mathcal{T}}}}\big|\mathbf{X}_{\mathcal{T} 
	\backslash {\hat{\mathcal{T}}}}\big) \big\} \big| \mathbf{X}_{\mathcal{T}} \big]\label{eq-AhatT}
\end{IEEEeqnarray}
	and the expectation in \eqref{eq-AhatT} is with respect to the conditional distribution
	\begin{equation}
P_{\bar{\mathbf{X}}_{{\hat{\mathcal{T}}}} | \mathbf{X}_{\mathcal{T}}} =  P_{{\mathbf{X}}_{{\hat{\mathcal{T}}}} |  \mathbf{X}_{\mathcal{T} \backslash {\hat{\mathcal{T}}}}}.
	\label{eq-rasc-achiev-dis}
	\end{equation}
\end{thm}

\begin{proof}
	We construct the random code ensemble $\mathsf{C}$ as follows. 
	
	\textit{Random Encoding Map}: 
	For every $i \in [K]$, 	
	draw encoder outputs $\mathsf{F}_{i}(x_i)$ for all $x_i \in \mathcal{X}_i$ 
	i.i.d. uniformly at random from $[Q_i]^{m_{\rm max}}$, 
	where $m_{\max} \triangleq \max \left\{m_\mathcal{T}: 
	\mathcal{T} \in \mathcal{P}([K])\right\}$. 
	
	\textit{Maximum Likelihood Decoder}: 
	For any $m \in [m_{\max}]$, $x_i \in \mathcal{X}_i$, and $i \in [K]$, 
	denote the first $m$ symbols of $\mathsf{F}_i(x_i)$
	by $\mathsf{F}_i(x_i)_{[m]}$.  
	For each $\mathcal{T}\in\mathcal{P}([K])$, 
	the maximum likelihood decoder $\mathsf{g}_\mathcal{T}$ 
	for $\mathcal{T}$ observes the first $m_{\mathcal{T}}$ symbols 
	from the encoders in $\mathcal{T}$, 
	here denoted by 
	\begin{equation}
	\mathsf{F}(\mathbf{x}_\mathcal{T})_{[m_{\mathcal{T}}]} 
	\triangleq \left(\mathsf{F}_i(x_i)_{[m_{\mathcal{T}}]}\right)_{i \in \mathcal{T}}, 
	\end{equation} 
	and, for each $\textbf{c}_{\mathcal{T}}=(\textbf{c}_i)_{i \in \mathcal{T}}
	\in \prod\limits_{i \in \mathcal{T}} [Q_i]^{m_{\mathcal{T}}}$, 
	produces the output
	\begin{equation}
	\mathsf{g}_\mathcal{T}(\textbf{c}_{\mathcal{T}}) = 
	\arg \underset{\substack{{\mathbf{x}_\mathcal{T} \in \mathcal{X}_\mathcal{T}:} \\ 
	{\mathsf{F}(\mathbf{x}_\mathcal{T})_{[m_\mathcal{T}]} 
	= \textbf{c}_{\mathcal{T}}}}}{\min} \imath(\mathbf{x}_\mathcal{T}).
	\end{equation} 
	
	\textit{Expected Error Analysis}: 
	The expected error probability $\mathbb{E}\left[P_{e,\mathcal{T}}(\mathsf{C})\right]$ 
	over the random code ensemble is bounded above by the probability of event 
	\begin{IEEEeqnarray}{rCl}
	\mathcal{E}_\mathcal{T} 
	&\triangleq& \left\{\exists \, \bar{\mathbf{x}}_{\mathcal{T}} \in 
	\mathcal{X}_\mathcal{T} \backslash \{\mathbf{X}_{\mathcal{T}}\}: \right. \nonumber \\
	&& \left. \imath\big(\bar{\mathbf{x}}_\mathcal{T}\big) 
	\leq \imath\big(\mathbf{X}_\mathcal{T}\big), \,
	\mathsf{F}(\bar{\mathbf{x}}_\mathcal{T})_{[m_\mathcal{T}]} = 
	\mathsf{F}(\mathbf{X}_\mathcal{T})_{[m_\mathcal{T}]}  \right\}. 
	\IEEEeqnarraynumspace
	\end{IEEEeqnarray}
	It follows that
	\begin{IEEEeqnarray}{rCl}
	\lefteqn{\mathbb{P} \left[\mathcal{E}_{\mathcal{T}}\right]} \nonumber \\
	&=& \mathbb{P} \left[\bigcup\limits_{\bar{\mathbf{x}}_{\mathcal{T}} 
	\in \mathcal{X}_\mathcal{T} \backslash \{\mathbf{X}_{\mathcal{T}}\}}  
	\left\{ \imath(\bar{\mathbf{x}}_\mathcal{T}) \leq \imath(\mathbf{X}_\mathcal{T}), 
	\right. \right. \nonumber \\
	&& \left. \left. 
	\phantom{\mathbb P \big[ }
	\mathsf{F}(\bar{\mathbf{x}}_\mathcal{T})_{[m_\mathcal{T}]} 
	= \mathsf{F}(\mathbf{X}_\mathcal{T})_{[m_\mathcal{T}]} \right\} 
	\vphantom{\bigcup\limits_{\mathcal{T} \in \mathcal{P}(\mathcal{T})}}  \right] \\
	&=& \mathbb{P} \left[\bigcup\limits_{{\hat{\mathcal{T}}} \in \mathcal{P}(\mathcal{T})} 
	\left\{\bigcup\limits_{
	\substack{{\bar{\mathbf{x}}_{{\hat{\mathcal{T}}}}\in\mathcal{X}_{\hat{\mathcal{T}}}:}\\
	{\bar{\mathbf{x}}_i\neq\mathbf{X}_i\ \forall i\in\hat{\mathcal{T}}}}}
	\left\{ \imath\big(\bar{\mathbf{x}}_{\hat{\mathcal{T}}}, \mathbf{X}_{\mathcal{T} 
	\backslash {\hat{\mathcal{T}}}}\big) \leq \imath(\mathbf{X}_\mathcal{T}), 
	\right.\right. \right. \nonumber \\
	&& \left. \left. \left. 
	\phantom{\mathbb P \big[ }
	\mathsf{F}(\bar{\mathbf{x}}_{\hat{\mathcal{T}}})_{[m_\mathcal{T}]} 
	= \mathsf{F}(\mathbf{X}_{\hat{\mathcal{T}}})_{[m_\mathcal{T}]} \right\} 
	\vphantom{\bigcup\limits_{{\hat{\mathcal{T}}} \in \mathcal{P}(\mathcal{T})}} \right\} 
	\right] \label{eq-rasc-achiev-b1} \\
	&=& \mathbb{P} \left[\bigcup\limits_{{\hat{\mathcal{T}}} 
	\in \mathcal{P}(\mathcal{T})} \left\{\qquad \;\; 
	\mathclap{\bigcup\limits_{
	\substack{{\bar{\mathbf{x}}_{{\hat{\mathcal{T}}}}\in\mathcal{X}_{\hat{\mathcal{T}}}:}\\
	{\bar{\mathbf{x}}_i\neq\mathbf{X}_i\ \forall i\in\hat{\mathcal{T}}}}}}
	\qquad \left\{ \imath\big(\bar{\mathbf{x}}_{\hat{\mathcal{T}}}|
	\mathbf{X}_{\mathcal{T} \backslash {\hat{\mathcal{T}}}}\big) \leq 
	\imath\big(\mathbf{X}_{\hat{\mathcal{T}}}|\mathbf{X}_{\mathcal{T} 
	\backslash {\hat{\mathcal{T}}}}\big), \right. \right. \right. \nonumber \\
	&& \left. \left. \left.
	\phantom{\mathbb P \big[ }
	 \mathsf{F}(\bar{\mathbf{x}}_{\hat{\mathcal{T}}})_{[m_\mathcal{T}]} 
	= \mathsf{F}(\mathbf{X}_{\hat{\mathcal{T}}})_{[m_\mathcal{T}]} \right\} 
	\vphantom{\bigcup\limits_{{\hat{\mathcal{T}}} \in \mathcal{P}(\mathcal{T})}} \right\} 
	\right] \label{eq-rasc-achiev-b2} \\
	&\leq& \mathbb{E} \Bigg[\min \Bigg\{1, \, \sum\limits_{{\hat{\mathcal{T}}} \in 
	\mathcal{P}(\mathcal{T})} \exp\big(-m_{\mathcal{T}} \cdot \overline{Q}
	({\hat{\mathcal{T}}}) \big) \nonumber \\
	&&  \left.\left. \phantom{\mathbb P \big[ }
	\cdot \sum_{\bar{\mathbf{x}}_{{\hat{\mathcal{T}}}}\in\mathcal{X}_{\hat{\mathcal{T}}}} 1 \left\{ \imath\big(\bar{\mathbf{x}}_{\hat{\mathcal{T}}}|
	\mathbf{X}_{\mathcal{T} \backslash {\hat{\mathcal{T}}}}\big) \leq 
	\imath\big(\mathbf{X}_{\hat{\mathcal{T}}}|\mathbf{X}_{\mathcal{T} 
	\backslash {\hat{\mathcal{T}}}}\big) \right\}
	 \right\} \right], 
	\label{eq-rasc-achiev-b3}
	\end{IEEEeqnarray} 
and \eqref{eq-rasc-achiev-b3} is equal to the right-hand-side of \eqref{eq-rasc-achiev-b}. 
	Here, \eqref{eq-rasc-achiev-b1} considers the case 
	where source symbols in set $\hat{\mathcal{T}}$ are decoded incorrectly 
	for each $\hat{\mathcal{T}} \in \mathcal{P}(\mathcal{T})$.  
	The derivation of \eqref{eq-rasc-achiev-b3} from \eqref{eq-rasc-achiev-b2} 
	follows the argument in \eqref{eq-sw-rcu-11}--\eqref{eq-sw-rcu-7}. 
	Specifically, since each component of $\bar{\mathbf{x}}_{{\hat{\mathcal{T}}}}$ 
	differs from the corresponding component of $\mathbf{X}_{{\hat{\mathcal{T}}}}$ 
	and since the encoder output for each is drawn independently and uniformly at random 
	from $[Q_i]^{m_{\max}}$, 
	\begin{IEEEeqnarray}{rCl}
	\lefteqn{\left. \mathbb{P}\left[ 
	\mathsf{F}(\bar{\mathbf{x}}_{\hat{\mathcal{T}}})_{[m_\mathcal{T}]} = 
	\mathsf{F}(\mathbf{X}_{\hat{\mathcal{T}}})_{[m_\mathcal{T}]} \right| 
	\mathbf{X}_{\mathcal{T}}\right]} \nonumber \\
	&=& \prod\limits_{i \in {\hat{\mathcal{T}}}} \frac{1}{Q_i^{m_\mathcal{T}}} \\
	&=& \exp\big(-m_\mathcal{T}\cdot \overline{Q}({\hat{\mathcal{T}}}) \big) 
	\label{eq-rasc-achiev-rep}
	\end{IEEEeqnarray} 
	for any $\bar{\mathbf{x}}_{{\hat{\mathcal{T}}}} \in 
	\mathcal{X}_{{\hat{\mathcal{T}}}}\backslash \{\mathbf{X}_{{\hat{\mathcal{T}}}}\}$. 
\end{proof}

We now prove the achievability part of Theorem~\ref{thm-rasc} 
by applying Lemma~\ref{lem-bad-code} 
to the random code in Theorem~\ref{thm-rasc-achiev-bound}.
\begin{proof}[Proof of Theorem \ref{thm-rasc}: achievability]
The probability that random RASC $\mathsf{C}$ 
has error probability $P_{e,\mathcal{T}}(\mathsf{C})$ 
greater than $\epsilon_{\mathcal{T}}$ for some possible set $\mathcal{T}\in\mathcal{P}([K])$ of active encoders
is 
\begin{equation}
\mathbb{P}\left[\bigcup\limits_{\mathcal{T} \in \mathcal{P}([K])} 
\left\{P_{e,\mathcal{T}}(\mathsf{C}) > \epsilon_\mathcal{T}\right\} \right] 
\leq \sum\limits_{\mathcal{T} \in \mathcal{P}([K])} \mathbb{P} 
\left[P_{e,\mathcal{T}}(\mathsf{C}) > \epsilon_\mathcal{T} \right]. \label{eq-rand-code}
\end{equation}
To bound each term 
$\mathbb{P} \left[P_{e,\mathcal{T}}(\mathsf{C}) > \epsilon_\mathcal{T} \right]$ 
using Lemma~\ref{lem-bad-code}, 
we next bound the expected error probability 
$\mathbb{E}\left[P_{e, \mathcal T}(\mathsf{C})\right]$ 
and the error probability $\epsilon^{*}(\mathcal{C}_{\mathcal T})$ 
for the best code in $\mathcal{C}_\mathcal{T}$, where $\mathcal{C}_\mathcal{T}$  is the set of 
$\big(n,(Q_i^{m_{\mathcal{T}}},\, i \in \mathcal{T}), \epsilon_{\mathcal{T}}\big)$ 
MASCs with $m_\mathcal{T}$ set as in \eqref{eq-def-mt} below. 

To find $\mathbb{E}\left[P_{e, \mathcal T}(\mathsf{C})\right]$, 
we apply Theorem~\ref{thm-rasc-achiev-bound} to our 
stationary, memoryless sources 
with $n$-symbol distribution $P_{\mathbf{X}_{[K]}^n} = P_{\mathbf{X}_{[K]}}^{n}$. 
Given any $\mathcal{T} \in \mathcal{P}([K])$ 
and ${\hat{\mathcal{T}}} \in \mathcal{P}(\mathcal{T})$, let 
	\begin{equation}
	I_{\mathcal{T}, {\hat{\mathcal{T}}}} 
	\triangleq \imath\big(\mathbf{X}^n_{{\hat{\mathcal{T}}}}|\mathbf{X}^n_{\mathcal{T} 
	\backslash {\hat{\mathcal{T}}}}\big). \label{Itau}
	\end{equation} 
	Under moment assumptions \eqref{assump-c1}--\eqref{assump-c3}, one can generalize the argument in \eqref{eq-def-I1n}--\eqref{eq:subR}	to $|\mathcal T|$ active encoders to obtain 
	\begin{align}
\!\!\! & \mathbb{E}\left[P_{e, \mathcal{T}}(\mathsf{C})\right]  
 	\leq \mathbb{P}\left[\bigcup\limits_{{\hat{\mathcal{T}}} \in \mathcal{P}
 	(\mathcal{T})} \left\{I_{\mathcal{T},{\hat{\mathcal{T}}}} > m_\mathcal{T}
 	\overline{Q}({\hat{\mathcal{T}}}) + \frac{\log n}{2} \right.\right.\label{eq-rasc-achiev-b7}  \\
\!\!\!	&  \left.\left. - \log \left( \bar K_{\mathcal{T},{\hat{\mathcal{T}}}} 
	\left(2^{|\mathcal{T}|}-1\right) \right) \vphantom{\frac{\log n}{2}}\right\} 
	\vphantom{\bigcap\limits_{{\hat{\mathcal{T}}}}} \right] + 
	\sum_{\hat {\mathcal T} \in \mathcal P(\mathcal T)} \left( \frac{K_{\mathcal{T}, \hat{\mathcal{T}} }}{\sqrt{n}} + \frac{S_{\mathcal{T}, \hat{\mathcal{T}} }}{{n}}\right), \notag
	\end{align} 
	where $\bar K_{\mathcal{T},{\hat{\mathcal{T}}}}$, $K_{\mathcal{T},{\hat{\mathcal{T}}}}$ and  $S_{\mathcal{T},{\hat{\mathcal{T}}}}$
	are finite positive constants. 

	Fix any $\mathbf{Q}_{[K]}$. 
	By the definition of $\overline{\mathbf{R}}_{\mathcal{T}}$ 
	in \eqref{eq-def-RPT} and the relation in \eqref{eq-def-rate-vector}, 
	we see that
	\begin{equation}
	\overline{\mathbf{R}}_{\mathcal{T}} 
	= \frac{1}{n}\big(m_\mathcal{T}\, \overline{Q}({\hat{\mathcal{T}}}), \, 
	{\hat{\mathcal{T}}} \in \mathcal{P}(\mathcal{T})\big). \label{eq:RQ}
	\end{equation} 
	For brevity, define constant vector
	\begin{equation}
	{\mathbf{C}}_{\mathcal{T}} 
	\triangleq \left(\log\left(\bar K_{\mathcal{T},{\hat{\mathcal{T}}}}\left(2^{|\mathcal{T}|}
	-1 \right)\right), \, {\hat{\mathcal{T}}} \in \mathcal{P}(\mathcal{T}) \right).
	\end{equation} 
	and the almost-constant error thresholds
\begin{equation}
\epsilon_\mathcal{T}^\prime \triangleq  \epsilon_\mathcal{T} - \frac {B}{\sqrt n}- \sum_{\hat {\mathcal T} \in \mathcal P(\mathcal T)} \left( \frac{K_{\mathcal{T}, \hat{\mathcal{T}} }}{\sqrt{n}} + \frac{S_{\mathcal{T}, \hat{\mathcal{T}} }}{{n}} \right)
\end{equation}
where $B$ is the Bentkus constant \eqref{eq-lem-b-e} for the vector of information densities \eqref{Itau}.
	We choose the decoding blocklength $m_\mathcal{T}$ as
	\begin{align}
 	& m_\mathcal{T} = \min \left\{m_\mathcal{T}: \overline{\mathbf{R}}_{\mathcal{T}} \in 
    \overline{\mathscr{R}}_{\mathcal{T}}^{*}\left(n, \epsilon_\mathcal{T}^\prime 
    - \delta_\mathcal{T}  \right) 
    + \frac{1}{n} {\mathbf{C}}_{\mathcal{T}} \right\} , \label{eq-def-mt}
    \end{align} 
    where $\delta_\mathcal{T}$ (which may be a function of $n$) satisfying
    $0 \leq \delta_\mathcal{T} < \epsilon_\mathcal{T}$ will be determined in the sequel, and $\overline{\mathscr{R}}_{\mathcal{T}}^{*}(n, \epsilon)$ 
    is defined in \eqref{eq-def-R3PT}. 
    Applying Lemma \ref{lem-b-e} to \eqref{eq-rasc-achiev-b7} with $m_\mathcal{T}$ in \eqref{eq-def-mt} yields
	\begin{equation}
	\mathbb{E}\left[P_{e, \mathcal{T}}(\mathsf{C})\right] 
	\leq  \epsilon_\mathcal{T} - \delta_\mathcal{T}. \label{eq-rasc-ave-err}
	\end{equation} 	
	To lower-bound $\epsilon^{*}(\mathcal{C})$, 
	for each $n$ and $\epsilon$ define
	\begin{align}
	m_{\mathcal{T}}^{*}(n, \epsilon) &\triangleq 
	\min\left\{m_\mathcal{T}: \mathbf{R}_\mathcal{T} 
	\in \mathscr{R}_{\mathcal{T}}^{*}(n, \epsilon) \right\} \\
	&\geq \min\left\{m_\mathcal{T}: \mathbf{R}_\mathcal{T} 
	\in \mathscr{R}_{{\rm out}, \mathcal{T}}^{*}\left(n, \epsilon \right) \right\}, \label{eq:convm}
	\end{align} 
	where $\mathscr{R}_{\mathcal{T}}^{*}(n, \epsilon)$ 
	is the $(n,\epsilon)$-MASC rate region (see Remark~\ref{rem-sw-general}), $\mathscr{R}_{{\rm out}, \mathcal{T}}^{*}\left(n, \epsilon \right)$ is defined in \eqref{eq-def-sw-out}, and \eqref{eq:convm} is by the converse (Theorem~\ref{thm-cf-rasc-conv}).  By Lemma~\ref{lem-sve}-\ref{part2}, one can always choose $\Delta_\mathcal{T} = O \left( \frac 1 {\sqrt n} \right) $ such that for $n$ sufficiently large
	\begin{align}
	\!\!\!\!\!\!\!\! \mathscr{R}_{{\rm out}, \mathcal{T}}^{*}\left(n, \epsilon_\mathcal{T}^\prime - \delta_{\mathcal T}
	- \Delta_\mathcal{T} \right) 
	\subseteq&~ \overline{\mathscr{R}}_{\mathcal{T}}^{*}\left(n, \epsilon_\mathcal{T}^\prime - \delta_{\mathcal T}\right) 
+ \frac 1 n {\mathbf{C}}_{\mathcal{T}} 
. \!\!\!\!\!\!
\label{eq-rasc-achiev-11}
	\end{align} 
It follows that
\begin{equation}
m_{\mathcal{T}}^{*}\left(n, \epsilon_\mathcal{T}^\prime - \delta_\mathcal{T} 
- \Delta_\mathcal{T} \right) 
\geq m_\mathcal{T}. \label{eq-rasc-mt-relation}
\end{equation}

Equation~\eqref{eq-rasc-mt-relation} and the converse (Theorem~\ref{thm-cf-rasc-conv}) imply 
that the minimal error probability over $\mathcal{C}_\mathcal{T}$ satisfies
\begin{equation}
\epsilon^{*}(\mathcal{C}_\mathcal{T}) 
\geq \epsilon_\mathcal{T}^\prime - \delta_\mathcal{T} 
- \Delta_\mathcal{T}. \label{eq-rasc-mk-1}
\end{equation}

Plugging \eqref{eq-rasc-ave-err} and \eqref{eq-rasc-mk-1} 
into Lemma~\ref{lem-bad-code} and noting the monotonicity of the bound in Lemma~\ref{lem-bad-code}  gives 
\begin{IEEEeqnarray}{rCl}
\mathbb{P}\left[P_{e,\mathcal{T}}(\mathsf{C}) > \epsilon_\mathcal{T} \right] 
&\leq& \frac{\mathbb{E}\left[P_{e,\mathcal{T}}(\mathsf{C})\right] 
- \epsilon^{*}(\mathcal{C}_\mathcal{T})}{\epsilon_\mathcal{T} 
- \epsilon^{*}(\mathcal{C}_\mathcal{T})} \\
&\leq& \frac{ \epsilon_\mathcal{T} - \epsilon_\mathcal{T}^\prime 
+ \Delta_\mathcal{T}}{\epsilon_\mathcal{T} - \epsilon_\mathcal{T}^\prime  + \delta_{\mathcal T}}, 
  \label{eq-rasc-mk-2}
\end{IEEEeqnarray} 
We may choose $\delta_\mathcal{T} = O\left(\frac{1}{\sqrt{n}}\right)$ 
to ensure that the right-hand side of \eqref{eq-rasc-mk-2} 
is as small a constant as desired. Specifically, we choose constants
$(\lambda_\mathcal{T})_{\mathcal{T} \in \mathcal{P}([K])}$ to satisfy 
\begin{equation}
\sum\limits_{\mathcal{T} \in \mathcal{P}([K])} \frac{1}{\lambda_\mathcal{T} + 1} < 1,
\label{eq-rand-code-1}
\end{equation} 
and put
\begin{equation}
\delta_\mathcal{T} = \lambda_\mathcal{T} \left( \epsilon_\mathcal{T} - \epsilon_\mathcal{T}^\prime + \Delta_\mathcal{T} \right) \label{eq-rasc-deltak}.
\end{equation} 
With \eqref{eq-rasc-mk-2} and \eqref{eq-rasc-deltak}, 
we bound the right-hand side of \eqref{eq-rand-code} as
\begin{align}
\sum\limits_{\mathcal{T} \in \mathcal{P}([K])} \mathbb{P}\left[P_{e,\mathcal{T}}(\mathsf{C}) 
> \epsilon_\mathcal{T} \right] \leq \sum\limits_{\mathcal{T} 
\in \mathcal{P}([K])} \frac{1}{\lambda_\mathcal{T} + 1} < 1,
\end{align} 
which implies the existence of a deterministic 
$\left(n, \overline{\mathbf{m}}_{K}, \mathbf{Q}_{[K]}, 
\overline{\boldsymbol{\epsilon}}_{K}\right)$ 
RASC with $m_\mathcal{T}$ in \eqref{eq-def-mt}, $\overline{Q}({\hat{\mathcal{T}}})$ in \eqref{eq:RQ},  and $\mathbf{R}_\mathcal{T}$ in \eqref{eq-def-mt}. 
\end{proof}

\begin{remark} \label{rem-rasc-tradeoff}
When parameters 
$\left(n,\mathbf{Q}_{[K]}, \overline{\boldsymbol{\epsilon}}_{K}\right)$ are fixed, 
increasing $\lambda_\mathcal{T}$ 
yields larger decoding blocklengths $m_\mathcal{T}$. 
Therefore, 
the choice of $(\lambda_\mathcal{T})_{\mathcal{T} \in \mathcal{P}([K])}$ 
to satisfy \eqref{eq-rand-code-1} 
controls the RASC performance trade-off 
across different active encoder sets. 
This trade-off affects the performance of the RASC 
in the fourth- or higher-order terms.  
\end{remark}

\subsection{RASC for Permutation-Invariant Sources}
\label{sec-rasc-perminv}

A 
\emph{permutation-invariant}\footnote{Polyanskiy~\cite{ra-polyanskiy}
introduces a similar notion of permutation invariance 
for multiple access channel coding in~\cite{ra-polyanskiy}.} source 
is defined by the constraint 
\begin{equation}
P_{\mathbf{X}_{[K]}}\left(\mathbf{x}_{[K]}\right) 
= P_{\mathbf{X}_{[K]}}\left(\mathbf{x}_{\pi([K])}\right) \label{eq-rasc-perm}
\end{equation}
for all permutations $\pi$ on $[K]$ 
and all $\mathbf{x}_{[K]} \in \mathcal{X}_{[K]}$. 
For example, given any $P_S$ and $P_{X|S}$, 
the marginal $P_{\mathbf{X}_{[K]}}$ 
of $P_{\mathbf{X}_{[K]}S} = (P_{X|S})^K P_S$ 
satisfies \eqref{eq-rasc-perm}. 
Such ``hidden variable'' models have applications in statistics, science, and economics, 
where latent variables 
(e.g., the health of the world economy or the state of the atmosphere) 
influence observables (e.g., stock prices or climates).
Figure~\ref{fig-ras} shows an example 
with $K$ sensors reading measurements of a common hidden state $S$.

\begin{figure}[!t]
	\centering
	\includegraphics[width=0.25\textwidth]{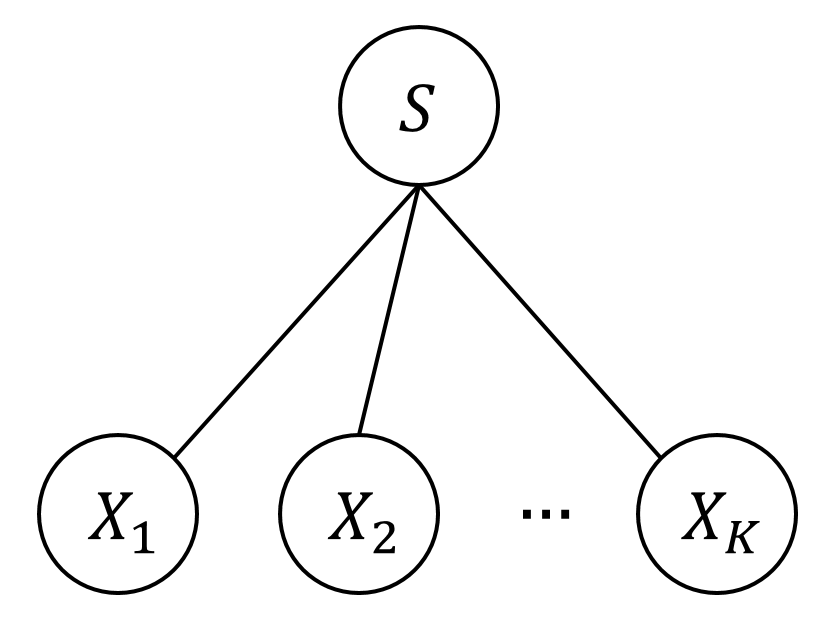}
	\caption{A graphical model of a common distributed sensing scenario.}
	\label{fig-ras}
\end{figure} 

Permutation-invariant source models interest us 
both because of their wide applicability 
and because they present an opportunity 
for code simplification through \emph{identical encoding}, 
where all encoders employ the same encoding map. 
For any permutation-invariant source, 
\eqref{eq-rasc-red} and \eqref{eq-rasc-perm} imply
that $\mathcal{X}_i=\mathcal{X}$ for all $i\in[K]$ and,
for any $\mathcal{T} \in \mathcal{P}([K])$ with $|\mathcal{T}| = k$,  
\begin{align}
P_{\mathbf{X}_\mathcal{T}} &= P_{\mathbf{X}_{[k]}}. \label{eq-rasc-p1}
\end{align}
Thus,  $P_{\mathbf{X}_\mathcal{T}}$ is permutation-invariant for every $\mathcal{T}$ 
and the joint source distribution depends on the number of active encoders 
but not their identities. 
Assuming that we further employ the same error probability $\epsilon_k$ 
for all $\mathcal{T}\in\mathcal{P}([K])$ with $|\mathcal{T}|=k$, 
we can fix a single decoding blocklength 
for each number $k\in[K]$ of active encoders 
and use identical encoders at all transmitters,  
allowing us to accommodate 
an arbitrarily large number of encoders 
without designing a unique encoder for each. 
A similar phenomenon arises for RA channel coding~\cite{recep}.



In analyzing RASC performance 
with identical encoders on a permutation-invariant source, 
we assume in addition to \eqref{assump-c1} and \eqref{assump-c2} 
that no two sources are identical, i.e., 
\begin{equation}
\mathbb{P}\left[ \bigcup\limits_{i,j\in[K], \, i \neq j} \{X_i = X_j\} \right] < 1. \label{eq-no-rep}
\end{equation}
This is important since using identical encoders on identical sources 
yields identical descriptions, 
in which case descriptions from multiple encoders 
are no better than descriptions from a single encoder. 
Under these assumptions, 
Theorem~\ref{thm-rasc} continues to hold. 
In the analysis, 
we modify the decoder to output the most probable source vector 
$\mathbf{x}_{\mathcal{T}} \in \mathcal{X}_\mathcal{T}$ 
that contains no repeated symbols 
(see the proof of Theorem~\ref{thm-rasc-achiev-bound}), 
treating the case where $\mathbf{X}_{\mathcal{T}}$ 
contains repeated symbols as an error. 
In the asymptotic analysis 
for stationary, memoryless sources, 
the probability of this error event is bounded by 
\begin{align}
&~\mathbb{P}\left[ \bigcup\limits_{i,j\in[K], \, i \neq j} \{\mathbf{X}_i^n 
= \mathbf{X}_j^n\} \right] \nonumber \\ 
\leq &~\left(\mathbb{P}\left[ \bigcup\limits_{i,j\in[K], \, i \neq j} \{X_i = X_j\} \right]\right)^n, 
\end{align} 
which decays exponentially in $n$ by \eqref{eq-no-rep}. 
Therefore, under the assumption in \eqref{eq-no-rep}, 
identical encoding does not incur a first-, second-, or third-order 
performance penalty.

\section{Concluding Remarks} \label{sec-conclusion}
This paper studies finite-blocklength lossless source coding 
in three scenarios. 
 
We derive a new non-asymptotic achievability (RCU) bound 
(Theorem \ref{thm-rcu-bound}) 
and use it to show that for point-to-point coding 
on stationary, memoryless sources, 
random code design with maximum likelihood decoding 
achieves the same coding rate 
up to the third-order as the optimal code from~\cite{kontoyiannis-verdu}. 
The RCU bound generalizes to the MASC scenario 
(Theorem~\ref{thm-sw-rcu}). 

A new HT converse (Theorem~\ref{thm-sw-cht-conv}) 
extends the channel coding meta-converse~\cite{pol-poo-ver} 
to an MASC 
and suggests the possibility of using composite hypothesis testing 
to derive converses for other multi-terminal scenarios. 
Our analysis of composite hypothesis testing provides general tools 
(Lemmas~\ref{lem-cht-achiev}, \ref{lem-cht-conv}, and \ref{lem-cht-var}) 
for use in other related problems. 
Just as the meta-converse for channel coding 
recovers previously known converses, 
our HT converse recovers Han's MASC converse~\cite[Lemma.~7.2.2]{han}.
Just as the HT converse for lossy source coding \cite[Th.~8]{kostina-verdu} 
is equivalent to the LP-based converse for that setting (see \cite[Cor.~3]{jose-k}), 
our MASC HT converse is equivalent 
to the MASC LP-based converse~\cite[Th.~12]{jose-k}. 

We give the first third-order characterization 
of the MASC rate region for stationary, memoryless sources, 
tightening prior second-order characterizations 
from~\cite{tan-kosut} and~\cite{nomura-han} 
and replacing the $2^k-1$ thresholds used there to decode for $k$ users 
by a maximum likelihood decoder 
that chooses the jointly most probable source realizations 
consistent with the received codewords. 
We show that 
for rate points converging to a non-corner point on the asymptotic sum-rate boundary, 
separate encoding does not compromise the performance 
in lossless data compression up to the third-order term. 
Numerical comparison of the new HT converse 
and the optimal performance of point-to-point source coding 
in Figure~\ref{fig-bounds} 
allows one to bound from below the small gap between joint and separate encoding, 
which is not captured in the first three terms of the asymptotic expansion. 
For independent sources, there are no non-corner points, 
and MASC separate encoding
incurs a positive penalty in the second-order term 
relative to joint encoding with a point-to-point code. 
When two sources have the same marginals, 
this penalty equals the penalty for using two independent blocklength-$n$ codes 
rather than a single blocklength-$2n$ point-to-point code 
for encoding $2n$ samples.

Our proposed RASC works universally 
for all possible encoder activity patterns. 
The nested structure of the RASC demonstrates 
that there is no need for the encoders to know the set of active encoders {\em a priori}.
The third-order-optimal MASC performance 
is achievable even when the only information the encoders receive 
is the acknowledgment that tells them when to stop transmitting (Theorem \ref{thm-rasc}). 
   
Our refinement of the traditional random coding argument 
(Lemma \ref{lem-bad-code} and \eqref{eq-rand-code}) 
uses bounds on the minimal (converse) and expected (achievability) error probabilities 
for {\em each} possible active encoder set 
to show the existence of a single code that is good 
for {\em all} possible active encoder sets. 
This argument is likely to be useful for other information-theoretic problems.


\appendices

\numberwithin{equation}{section}
\section{Proof of Theorem~\ref{thm-achiev-dt}}
\label{append-thm-achiev-dt}
\renewcommand{\theequation}{\thesection.\arabic{equation}}
Following \cite[Eq.~(68)]{pol-poo-ver}, 
note that for $z > 0$ and $\gamma > 0$
\begin{equation}
\exp \left\{-\left|\log \frac{\gamma}{z}\right|_{+} \right\} 
= 1\left\{z > \gamma \right\} + \frac{z}{\gamma}1\left\{z \leq \gamma \right\}. 
\label{eq-lem-dt-2}
\end{equation} 
Let $z = \frac{1}{P_X(X)}$ and $\gamma = M$. 
Then taking the expectation of both sides of \eqref{eq-lem-dt-2} 
with respect to $P_X$ gives
\begin{align}
&~{\mathbb{E}\left[\exp \left\{-\left|\log M - \imath(X)\right|_{+} \right\}\right]} 
\nonumber \\
=&~ \mathbb{P}\left[\imath(X) > \log M \right] + \frac{1}{M} \mathbb{U}
\left[\imath(X) \leq \log M\right], \label{eq-achiev-dt-1}
\end{align} 
where $\mathbb{P}\left[\cdot\right]$ denotes a probability 
with respect to $P_X$ 
and $\mathbb{U}\left[\cdot\right]$ denotes a mass 
with respect to the counting measure $U_X$ on $\mathcal{X}$, 
which assigns unit weight to each $x \in \mathcal{X}$. 
In light of \eqref{eq-achiev-dt-1}, 
we can prove \eqref{eq-achiev-dt} 
by demonstrating the existence of an $(M,\epsilon)$ code 
for which the right-hand side of \eqref{eq-achiev-dt-1} 
exceeds $\epsilon$. 
We prove a slightly stronger result, 
showing that there exists an $(M,\epsilon)$ code 
with a threshold decoder such that
\begin{equation} 
\epsilon \leq \mathbb{P}\left[\imath(X) > \log \gamma \right] 
+ \frac{1}{M} \mathbb{U}\left[\imath(X) \leq \log \gamma\right] 
\label{eq-lem-dt-1}
\end{equation} 
for all $\gamma > 0$. 
Setting $\gamma = M$ in \eqref{eq-lem-dt-1} yields the desired bound.

Fix $\gamma > 0$. For each $x \in \mathcal{X}$, 
randomly and independently draw 
each encoder output $\mathsf{F}(x)$ 
from the uniform distribution on $[M]$. 
Define the threshold decoder 
\begin{equation}
g(c) = \begin{cases}
x, &\text{if } \exists \, \text{unique } x \in \mathcal{X} \\
&\text{ s.t. } \mathsf{F}(x) = c, \, \imath(x) \leq \log \gamma \\
\text{error,} &\text{otherwise.}
\end{cases}
\end{equation}
We capture all errors 
using a union of error events 
\begin{IEEEeqnarray}{rCl}
\mathcal{E}_1 &\triangleq& \left\{\imath(X) > \log \gamma \right\} \\
\mathcal{E}_2 &\triangleq& 
\left\{\exists\, \bar{x} \in \mathcal{X}\backslash\{X\} \text{ s.t. } 
\mathsf{F}(\bar{x}) = \mathsf{F}(X),\; \imath(\bar{x}) \leq \log \gamma \right\}. 
\IEEEeqnarraynumspace
\end{IEEEeqnarray}
By the random coding argument and the union bound, 
there exists an $(M,\epsilon)$ code such that 
\begin{equation}
\epsilon \leq \mathbb{P}\left[\mathcal{E}_1 \cup \mathcal{E}_2 \right] 
\leq \mathbb{P}\left[\mathcal{E}_1\right] +  \mathbb{P}\left[\mathcal{E}_2 \right].
\end{equation}	
Here,
\begin{IEEEeqnarray}{rCl}
\mathbb{P}\left[\mathcal{E}_1 \right] 
&=& \mathbb{P}\left[\imath(X) > \log \gamma \right]\\
\mathbb{P}\left[\mathcal{E}_2 \right] 
&=& \mathbb{P} \left[\underset{\mathclap{\quad\;\; 
\bar{x}\in \mathcal{X}\backslash\{X\}}}{\;\; 
\bigcup \;}\left\{\mathsf{F}(\bar{x}) = \mathsf{F}(X), \; 
\imath(\bar{x}) \leq \log \gamma \right\} \right] \\
&\leq& \sum\limits_{\bar{x}\in \mathcal{X}\backslash\{X\}} \mathbb{P}\left[\mathsf{F}(\bar{x}) = \mathsf{F}(X) \right]1\left\{\imath(\bar{x}) \leq \log \gamma \right\} \label{eq-achiev-dt-8} \IEEEeqnarraynumspace\\
&\leq& \frac{1}{M}\sum\limits_{\bar{x}\in \mathcal{X}}1\left\{\imath(\bar{x}) \leq \log \gamma \right\} \label{eq-achiev-dt-9} \\
&=& \frac{1}{M} \mathbb{U}\left[\imath(X) \leq \log \gamma \right],
\end{IEEEeqnarray}
where \eqref{eq-achiev-dt-8} applies the union bound 
and \eqref{eq-achiev-dt-9} holds since the encoder outputs 
are i.i.d.\ and uniformly distributed. \qed

\numberwithin{equation}{section}
\section{Proof of Lemma \ref{lem-cht-conv}}
\label{append-lem-cht-conv}
\renewcommand{\theequation}{\thesection.\arabic{equation}}
The proof extends the proof of \cite[Eq.~(102)]{pol-poo-ver} (e.g., \cite{verdu-notes}). We show that for any test $P_{Z|X}$ that decides between $P$ vs. $\{Q_j\}_{j=1}^{k}$, 
\begin{align}
&~{\mathbb{P}\left[Z=1\right] - \sum\limits_{j=1}^{k}\gamma_j\mathbb{Q}_j\left[Z=1\right]} \nonumber \\
\leq&~ \mathbb{P}\left[\bigcap\limits_{j \in [k]} \left\{\frac{P(X)}{Q_j(X)} > \gamma_j\right\}\right], \IEEEeqnarraynumspace \label{eq-app-cht-1}
\end{align} where $\gamma_j \geq 0$, $j \in [k]$ are arbitrary constants. Then Lemma~\ref{lem-cht-conv} follows immediately by definition of $\beta_\alpha \left(P, \{Q_j\}_{j=1}^{k}\right)$. 

To prove \eqref{eq-app-cht-1}, fix a $\gamma_j \geq 0$ for each $j \in [k]$. We then have
\begin{IEEEeqnarray}{rCl}
\IEEEeqnarraymulticol{3}{l}{\ \ \ \  \mathbb{P}\left[Z=1\right] - \sum\limits_{j=1}^{k}\gamma_j\mathbb{Q}_j\left[Z=1\right]} \nonumber \\
&=& \sum\limits_{x \in \mathcal{X}} P_{Z|X}(1|x) \left(P(x) - \sum\limits_{j=1}^{k}\gamma_j Q_j(x) \right) \\
&\leq& \sum\limits_{x \in \mathcal{X}} P_{Z|X}(1|x) \left(P(x) - \sum\limits_{j=1}^{k}\gamma_j Q_j(x) \right) \nonumber \\
&&1\left\{P(x) > \sum\limits_{j = 1}^{k} \left\{\gamma_j Q_j(x)\right\} \right\} \\
&\leq& \sum\limits_{x \in \mathcal{X}} P_{Z|X}(1|x) P(x) 1\left\{P(x) > \sum\limits_{j=1}^{k}\gamma_j Q_j(x) \right\} \label{eq-app-cht-2} \IEEEeqnarraynumspace \\
&=& \mathbb{P} \left[Z=1, \, P(X) > \sum\limits_{j=1}^{k}\gamma_j Q_j(X) \right] \\
&\leq& \mathbb{P} \left[P(X) > \sum\limits_{j=1}^{k}\gamma_j Q_j(X) \right] \\
&\leq& \mathbb{P} \left[\bigcap\limits_{j\in [k]} \left\{P(X) > \gamma_j Q_j(X)\right\}\right], \label{eq-app-cht-3}
\end{IEEEeqnarray} where \eqref{eq-app-cht-2} follows from the non-negativity of probability and each $\gamma_j$. The proof is complete since \eqref{eq-app-cht-3} equals the right-hand-side of \eqref{eq-app-cht-1}. \qed

\numberwithin{equation}{section}
\section{Proof of Lemma \ref{lem-cht-var}}
\label{append-lem-cht-var}
\renewcommand{\theequation}{\thesection.\arabic{equation}}
For any test $P_{Z|X}$ deciding between $P$ vs. $\{Q_j\}_{j=1}^{k}$, 
we show that 
\begin{IEEEeqnarray}{rCl}
\IEEEeqnarraymulticol{3}{l}{\ \ \ \  \mathbb{P}\left[Z=1\right] - \sum\limits_{j=1}^{k}\gamma_j\mathbb{Q}_j\left[Z=1\right]} \nonumber \\
&\leq& 1 - \sum\limits_{x \in \mathcal{X}} \min \left\{P(x),\, \sum\limits_{j=1}^{k}\gamma_j Q_j(x) \right\}, \IEEEeqnarraynumspace \label{eq-app-cht2-1}
\end{IEEEeqnarray} where $\gamma_j \geq 0$, $j \in [k]$ are arbitrary constants. Fix a $\gamma_j \geq 0$ for each $j \in [k]$. For notational brevity, define sets
\begin{IEEEeqnarray}{rCl}
\mathcal{X}^{(<)} &\triangleq& \left\{x \in \mathcal{X}: P(x) < \sum\limits_{j=1}^{k}\gamma_jQ_j(x) \right\} \\
\mathcal{X}^{(=)} &\triangleq& \left\{x \in \mathcal{X}: P(x) = \sum\limits_{j=1}^{k}\gamma_jQ_j(x) \right\} \\
\mathcal{X}^{(>)} &\triangleq& \left\{x \in \mathcal{X}: P(x) > \sum\limits_{j=1}^{k}\gamma_jQ_j(x) \right\}.
\end{IEEEeqnarray} For any test $P_{Z|X}$, we have
\begin{IEEEeqnarray}{rCl}
\IEEEeqnarraymulticol{3}{l}{\mathbb{P}\left[Z=1\right] - \sum\limits_{j=1}^{k}\gamma_j\mathbb{Q}_j\left[Z=1\right]} \nonumber \\
&+& \sum\limits_{x \in \mathcal{X}} \min \left\{P(x),\, \sum\limits_{j=1}^{k}\gamma_j Q_j(x) \right\} \nonumber \\
&=& \sum\limits_{x \in \mathcal{X}} \left(P_{Z|X}(1|x)\left(P(x) - \sum\limits_{j=1}^{k}\gamma_j Q_j(x)\right) \right. \nonumber \\
&& \left. + \min \left\{P(x),\, \sum\limits_{j=1}^{k}\gamma_j Q_j(x) \right\} \right) \\
&=& \sum\limits_{x \in \mathcal{X}^{(<)}} \left(P_{Z|X}(1|x)\left(P(x) - \sum\limits_{j=1}^{k}\gamma_j Q_j(x)\right) + P(x) \right) \nonumber \\
&&+ \sum\limits_{x \in \mathcal{X}^{(>)}} \left(P_{Z|X}(1|x)\left(P(x) - \sum\limits_{j=1}^{k}\gamma_j Q_j(x)\right) \right. \nonumber \\
&& \left. + \sum\limits_{j=1}^{k}\gamma_j Q_j(x)\right) + \sum\limits_{x \in \mathcal{X}^{(=)}} P(x) \\
&\leq& \sum\limits_{x \in \mathcal{X}^{(<)}} P(x) + \sum\limits_{x \in \mathcal{X}^{(>)}} P(x) + \sum\limits_{x \in \mathcal{X}^{(=)}} P(x) \label{eq-app-cht2-2} \\
&=& 1. \label{eq-app-cht2-3}
\end{IEEEeqnarray} 
The equality in \eqref{eq-app-cht2-2} is achieved by test 
\begin{equation}
P_{Z|X}(1|x) = 
\begin{cases}
1 &\text{ for } x \in \mathcal{X}^{(>)} \\
0 &\text{ for } x \in \mathcal{X}^{(<)} \\
\lambda &\text{ for } x \in \mathcal{X}^{(=)}
\end{cases}
\end{equation} for any $\lambda \in [0,1]$. Rearranging \eqref{eq-app-cht2-3} yields \eqref{eq-app-cht2-1}. Choosing the unique $\lambda \in [0,1]$ to satisfy $\mathbb{P}\left[Z=1\right] = \alpha$, we obtain Lemma~\ref{lem-cht-var} by the definition of $\beta_\alpha \left(P, \{Q_j\}_{j=1}^{k}\right)$. \qed

\numberwithin{equation}{section}
\section{Proof of Lemma \ref{lem-b-e}}
\label{append-lem-b-e}
\renewcommand{\theequation}{\thesection.\arabic{equation}}
Recall that $\mathsf{T}$ is composed of the $r$ normalized eigenvectors corresponding to the non-zero eigenvalues of covariance matrix $\textsf{V}$ and $\mathbf{U}_i = \mathsf{T}\mathbf{W}_i$, where $\mathbf{W}_i \in \mathbb{R}^r$ for $i = 1,\ldots,n$. Thus $\mathsf{V} = \mathsf{T}\mathsf{V}_r\mathsf{T}^T$, where $\mathsf{V}_r \triangleq \text{Cov}\left[\mathbf{W}_1\right]$ is non-singular.

For each $\mathbf{z} \in \mathbb{R}^d$, define
\begin{equation}
\mathscr{A}_r(\mathbf{z}) \triangleq \{\mathbf{x} \in \mathbb{R}^r: \mathsf{T}\mathbf{x} \leq \mathbf{z} \},
\end{equation}
which is a convex subset of $\mathbb{R}^r$. Let $\mathbf{Z}_r \sim \mathcal{N}(\mathbf{0},\mathsf{V}_r) \in \mathbb{R}^r$. Applying \cite[Cor.~8]{tan-kosut} to the i.i.d. random vectors $\mathbf{W}_1,\ldots,\mathbf{W}_n$, we obtain
\begin{align}
&~\underset{\mathbf{z} \in \mathbb{R}^d}{\sup} \left|\mathbb{P}\left[\frac{1}{\sqrt{n}}\sum\limits_{i=1}^{n}\mathbf{W}_i \in \mathscr{A}_r(\mathbf{z})\right] - \mathbb{P}[\mathbf{Z}_r \in \mathscr{A}_r(\mathbf{z})]\right| \notag \\
\leq&~ \frac{400 r^{1/4}\beta_r}{\lambda_{\min}(\mathsf{V}_r)^{3/2}\sqrt{n}},
\end{align} which is equivalent to \eqref{eq-lem-b-e} by the definition of $\mathscr{A}_r(\mathbf{z})$. \qed

\numberwithin{equation}{section}
\section{Proof of Lemma~\ref{lem-sve}}
\label{append-sve}
\renewcommand{\theequation}{\thesection.\arabic{equation}}
For simplicity, we assume that $\mathsf{V}$ is non-singular. When $\mathsf{V}$ is singular, a similar analysis can be applied with $\mathsf{V}$ replaced by $\mathsf{V}_r$ defined in Lemma \ref{lem-b-e}. 

Let $\mathbf{Z} \sim \mathcal{N}(\mathbf{0},\mathsf{V})$ be a $d$-dimensional multivariate Gaussian with covariance matrix $\mathsf{V}$. Recall from \eqref{eq-def-sve} that $\mathscr{Q}_{\rm inv}(\mathsf{V},\epsilon)$ is defined as \begin{equation}
\mathscr{Q}_{\rm inv}(\mathsf{V},\epsilon) \triangleq \{\mathbf{z} \in \mathbb{R}^d: \mathbb{P}[\mathbf{Z} \leq \mathbf{z}] \geq 1-\epsilon \}.
\end{equation}
By the definition of $\mathscr{Q}_{\rm inv}(\mathsf{V},\epsilon)$ and the definition of $\Phi(\mathsf{V}; \mathbf{z})$ in \eqref{eq-def-gaussian},
$\Phi(\mathsf{V}; \mathbf{z}) = 1-\epsilon$ if and only if $\mathbf{z}$ lies on the boundary of $\mathscr{Q}_{\rm inv}(\mathsf{V},\epsilon)$, and $\Phi(\mathsf{V}; \mathbf{z}) > 1-\epsilon$ if and only if $\mathbf{z}$ lies in the interior of $\mathscr{Q}_{\rm inv}(\mathsf{V},\epsilon)$.

\begin{proof}[Proof of Lemma \ref{lem-sve}] To prove \eqref{eq-lem-sve-1}, consider any $D_1 > 0$ and $\delta \geq 0$. Since $\Phi(\mathsf{V}; \mathbf{z})$ is continuously differentiable everywhere provided that $\mathsf{V}$ is non-singular, we can apply the multivariate Taylor's theorem to expand $\Phi(\mathsf{V}; \mathbf{z} + D_1\delta\mathbf{1})$ as
\begin{equation}
\Phi(\mathsf{V}; \mathbf{z} + D_1\delta\mathbf{1}) = \Phi(\mathsf{V}; \mathbf{z}) + D_1\delta\sum\limits_{i=1}^{d} \frac{\partial \Phi(\mathsf{V}; \mathbf{z})}{\partial z_i} + \xi(\mathbf{z}, D_1\delta). \label{eq-bound-s-achiev-2}
\end{equation} The second-order residual term $\xi(\mathbf{z}, D_1\delta)$ can be bounded as
\begin{equation}
|\xi(\mathbf{z}, D_1\delta)| \leq \frac{\xi_{\max}}{2}(d\cdot D_1\delta)^2,
\end{equation} where
\begin{equation}
\xi_{\max} \triangleq \max\limits_{\delta' \in [0,D_1\delta]} \left\|\nabla^2 \Phi(\mathsf{V}; \mathbf{z} + \delta'\mathbf{1}) \right\|_{\max}
\end{equation} and $\|\cdot\|_{\max}$ denotes the max norm of a matrix.

Denote 
\begin{equation}
D' \triangleq \sum\limits_{i=1}^{d} \frac{\partial \Phi(\mathsf{V}; \mathbf{z})}{\partial z_i}.
\end{equation}
Since $\Phi(\mathsf{V}; \mathbf{z})$ is increasing in any coordinate of $\mathbf{z}$, $D' > 0$. Then, for any $\mathbf{z} \in \mathscr{Q}_{\rm inv}(\mathsf{V},\epsilon)$, we have
\begin{IEEEeqnarray}{rCl}
\Phi\left(\mathsf{V}; \mathbf{z} + D_1\delta\mathbf{1}\right) &\geq& \Phi(\mathsf{V}; \mathbf{z}) + D'D_1\delta - \frac{\xi_{\max}}{2}(d\cdot D_1\delta)^2 \IEEEeqnarraynumspace \\
&\geq&
1-\epsilon + \delta\left(D'D_1 - \frac{\xi_{\max}}{2}d^2 D_1^2\delta\right).
\end{IEEEeqnarray} 
We note that for any finite positive $D_1$, $\xi_{\max}$ approaches $\|\nabla^2 \Phi(\mathsf{V}; \mathbf{z}) \|_{\max}$ as $\delta \rightarrow 0$. Thus, for any finite positive $D_1$ that satisfies $D'D_1 > 1$, there exists some $\delta_1 > 0$ such that for all $0 \leq \delta < \delta_1$,
\begin{equation}
D'D_1 - \frac{\xi_{\max}}{2}d^2 D_1^2\delta \geq 1,
\end{equation} which yields
\begin{align}
\Phi(\mathsf{V}; \mathbf{z} + D_1\delta\mathbf{1}) \geq 1-\epsilon + \delta. \label{eq-bound-s-achiev-3} 
\end{align} By the definitions of $\Phi(\mathsf{V}; \mathbf{z})$ and $\mathscr{Q}_{\rm inv}(\mathsf{V},\epsilon)$, \eqref{eq-bound-s-achiev-3} implies
\begin{equation}
\mathbf{z} + D_1\delta\mathbf{1} \in \mathscr{Q}_{\rm inv}(\mathsf{V},\epsilon - \delta),
\end{equation} and consequently
\begin{equation}
\mathscr{Q}_{\rm inv}(\mathsf{V},\epsilon) + D_1\delta\mathbf{1} \subseteq \mathscr{Q}_{\rm inv}(\mathsf{V},\epsilon - \delta),
\end{equation} which proves \eqref{eq-lem-sve-1}.

Eq. \eqref{eq-lem-sve-2} can be proved in a similar way.

\end{proof}


\numberwithin{equation}{section}
\section{Equivalence between HT and LP-Based Converses for the MASC}
\label{append-equivalence}
\renewcommand{\theequation}{\thesection.\arabic{equation}}
In this appendix, we establish the equivalence between the HT converse and the LP-based converse by showing that the bounds in \eqref{eq-lem-jk} and \eqref{eq-sw-cht-9} are equivalent. According to \cite[Eq.~(31)]{jose-k}, \eqref{eq-lem-jk} is equivalent to the following converse
\begin{IEEEeqnarray}{rCl}
\epsilon &\geq& \sup\limits_{\eta_1, \eta_2, \eta_3 \in \mathcal Z} \left\{\sum\limits_{\substack{{x_1 \in \mathcal{X}_1} \\ {x_2 \in \mathcal{X}_2}}} \min \left\{P_{X_1X_2}(x_1, x_2), \, \sum\limits_{j=1}^{3} \eta_j(x_1,x_2) \right\} \right.  \nonumber \\
&& - M_1 \sum\limits_{x_2 \in \mathcal{X}_2} \max\limits_{\hat{x}_1 \in \mathcal{X}_1} \min\left\{P_{X_1X_2}(\hat{x}_1,x_2), \, \eta_1(\hat{x}_1,x_2) \right\} \nonumber \\
&& - M_2 \sum\limits_{x_1 \in \mathcal{X}_1} \max\limits_{\hat{x}_2 \in \mathcal{X}_2} \min\left\{P_{X_1X_2}(x_1,\hat{x}_2), \, \eta_2(x_1,\hat{x}_2) \right\} \nonumber \\
&& \left. - M_1M_2 \max\limits_{\substack{{\hat{x}_1 \in \mathcal{X}_1} \\ {\hat{x}_2 \in \mathcal{X}_2}} } \min \left\{P_{X_1X_2}(\hat{x}_1,\hat{x}_2),\, \eta_3(\hat{x}_1,\hat{x}_2) \right\} \vphantom{\sum\limits_{\substack{{x_1 \in \mathcal{X}_1} \\ {x_2 \in \mathcal{X}_2}}}} \right\}, \label{eq-append-equiv-1} 
\end{IEEEeqnarray} where the supremum is over
\begin{equation}
\mathcal{Z} \triangleq \{z: \mathcal{X}_1 \times \mathcal{X}_2 \rightarrow [0, \infty)\}.
\end{equation} 
Therefore, we show that \eqref{eq-append-equiv-1} is equivalent to \eqref{eq-sw-cht-9}.

We first demonstrate that \eqref{eq-append-equiv-1} implies \eqref{eq-sw-cht-9}. 
Set 
\[
\eta_i = \gamma_i Q^{(i)}_{X_1X_2} 
\]
for any $\sigma$-finite $Q^{(i)}_{X_1X_2}$ and $\gamma_i\geq 0$, $i\in[3]$.  
Since 
\begin{IEEEeqnarray*}{rcl}
\min\left\{P_{X_1X_2}(\hat{x}_1,x_2), \, \eta_1(\hat{x}_1,x_2) \right\} & \,\leq\, & \eta_1(\hat{x}_1,x_2)\\
\min\left\{P_{X_1X_2}(x_1,\hat{x}_2), \, \eta_2(x_1,\hat{x}_2) \right\}&\,\leq\,&\eta_2(x_1,\hat{x}_2) \\
\min \left\{P_{X_1X_2}(\hat{x}_1,\hat{x}_2),\, \eta_3(\hat{x}_1,\hat{x}_2) \right\}&\,\leq\,&
\eta_3(\hat{x}_1,\hat{x}_2)
\end{IEEEeqnarray*}
in \eqref{eq-append-equiv-1}, we obtain \eqref{eq-sw-cht-9}.

To prove the other direction, we substitute $z_1 = \gamma_1Q^{(1)}_{X_1X_2}$, $z_2 = \gamma_2Q^{(2)}_{X_1X_2}$, and $z_3 = \gamma_3Q^{(3)}_{X_1X_2}$ in the right-hand side of \eqref{eq-sw-cht-9} to obtain
\begin{IEEEeqnarray}{rCl}
\epsilon &\geq& \sup\limits_{z_1, z_2, z_3 \in \mathcal{Z}} \left\{ \sum\limits_{\substack{{x_1 \in \mathcal{X}_1} \\ {x_2 \in \mathcal{X}_2}}} \min \left\{P_{X_1X_2}(x_1,x_2), \, \sum\limits_{j=1}^{3}z_j(x_1,x_2) \right\} \right. \nonumber \\
&& - M_1\sum\limits_{x_2 \in \mathcal{X}_2}\max\limits_{\hat{x}_1 \in \mathcal{X}_1} z_1(\hat{x}_1,x_2) - M_2\sum\limits_{x_1 \in \mathcal{X}_1}\max\limits_{\hat{x}_2 \in \mathcal{X}_2} z_2(x_1,\hat{x}_2) \nonumber \\ 
&& \left. - M_1M_2 \max\limits_{\substack{{\hat{x}_1 \in \mathcal{X}_1} \\ {\hat{x}_2 \in \mathcal{X}_2}}}z_3(\hat{x}_1, \hat{x}_2) \vphantom{\sum\limits_{\substack{{x_1 \in \mathcal{X}_1} \\ {x_2 \in \mathcal{X}_2}}}}\right\} \label{eq-append-equiv-2},
\end{IEEEeqnarray}  
Take a supremum of the right-hand side of \eqref{eq-append-equiv-2} over $\eta_1, \eta_2, \eta_3 \in \mathcal Z$. Since \eqref{eq-append-equiv-2} does not contain $\eta_1, \eta_2, \eta_3$, this does not change anything. Now, weaken (i.e., lower-bound) the inner supremum over $z_1, z_2, z_3 \in \mathcal{Z}$ by setting 
\begin{equation}
{z}_j(x_1,x_2) = \min \left\{P_{X_1X_2}(x_1,x_2),\, \eta_j(x_1,x_2) \right\}.
\end{equation}
Observing that
\begin{IEEEeqnarray}{rCl}
\IEEEeqnarraymulticol{3}{l}{\min \left\{P_{X_1X_2}(x_1,x_2), \, \sum\limits_{j=1}^{3}\min\left\{P_{X_1X_2}(x_1,x_2), \, \eta_j(x_1,x_2) \right\} \right\}} \nonumber \\
&=& \min \left\{P_{X_1X_2}(x_1,x_2), \, \sum\limits_{j=1}^{3} \eta_j(x_1,x_2) \right\},
\end{IEEEeqnarray} 
we see that the result of our weakening is exactly the right-hand side of \eqref{eq-append-equiv-1}, as desired.  \qed

\numberwithin{equation}{section}
\section{MASCs for Sources with Less Redundancy}
\label{append-zero-var}
\renewcommand{\theequation}{\thesection.\arabic{equation}}

Applying Lemma~\ref{lem-p} to get the asymptotic achievability result in Theorem~\ref{thm-sw} 
requires that all $V(X_1,X_2)$, $V(X_1 |X_2)$, and $V(X_2 |X_1)$ are strictly positive 
(as an implication of assumption \eqref{assump-b1}). 
Thus, the analysis in Section \ref{sec-result-SW} breaks down 
when any of these varentropies is equal to zero. (We refer to such a source as being \emph{less redundant}.) In this appendix, we analyze the performance of the MASC for less redundant sources. Specifically, we consider a pair of stationary, memoryless sources and analyze the following three cases: 
\begin{enumerate}
\item[1)] all \emph{three} varentropies are equal to zero;
\item[2)] exactly \emph{two} of the varentropies are equal to zero;
\item[3)] exactly \emph{one} of the varentropies is equal to zero. 
\end{enumerate}

We continue to assume that the joint distribution $P_{X_1 X_2}$ satisfies \eqref{assump-b2} and \eqref{assump-b3}. 
For those cases in which $V(X_2 |X_1) > 0$, we continue to assume $\mathbb E \left[ V_c (X_2 | X_1) \right]  \!>\! 0$. Likewise, if $V(X_1 |X_2) > 0$, we continue to assume $\mathbb E \left[ V_c (X_1 | X_2) \right]  \!>\! 0$.

In point-to-point almost-lossless source coding, the optimal code for a non-redundant source is easy to find (see Remark~\ref{rem-zero-var}). 
When the encoders are required to operate independently in a MASC, 
we know no easy way to find the optimal codes in general. 
In Section~A below, we give characterizations of the $(n,\epsilon)$-rate region in the three general cases listed above using the techniques developed in Section~\ref{sec-result-SW}. Then, in Section~B, we restrict attention to the case where $P_{X_1X_2}(x_1,x_2) > 0$ for every $(x_1,x_2) \in \mathcal{X}_1 \times \mathcal{X}_2$; under this condition, the optimal codes can be found and analyzed directly.

\subsection{General Characterizations of the $(n,\epsilon)$-Rate Region}
\label{append-zero-var-a}
We first list our results in the three general cases below.

\textit{Case 1)}: Suppose that $V(X_1|X_2) = 0$, $V(X_2|X_1) = 0$, and $V(X_1,X_2) = 0$. For any $\delta_1$, $\delta_2$, $\delta_{12} > 0$, let 
\begin{IEEEeqnarray}{rCl}
\hat{\mathscr{R}}_{\rm in}^{(1)}(n,\delta_1, \delta_2, \delta_{12}) &\triangleq& \bigg\{(R_1,R_2) \in \mathbb{R}^2: \nonumber \\
R_1 &\geq& H(X_1|X_2) + \frac{1}{n}\log \frac{1}{\delta_1} \nonumber \\
R_2 &\geq& H(X_2|X_1) + \frac{1}{n}\log \frac{1}{\delta_2} \nonumber \\
R_1+R_2 &\geq& H(X_1,X_2) + \frac{1}{n}\log \frac{1}{\delta_{12}} \bigg\}. \IEEEeqnarraynumspace
\end{IEEEeqnarray} Define
\begin{IEEEeqnarray}{rCl}
\mathscr{R}_{\rm in}^{(1)}(n,\epsilon) &\triangleq& \bigcup\limits_{\substack{{\delta_1, \delta_2, \delta_{12} > 0}\\{\delta_1+\delta_2+\delta_{12} = \epsilon}}} \hat{\mathscr{R}}_{\rm in}^{(1)}(n, \delta_1, \delta_2, \delta_{12}) \IEEEeqnarraynumspace \\
\mathscr{R}_{\rm out}^{(1)}(n,\epsilon) &\triangleq& \bigg\{(R_1,R_2) \in \mathbb{R}^2: \nonumber \\
R_1 &\geq& H(X_1|X_2) - \frac{1}{n}\log \frac{1}{1-\epsilon} \nonumber \\
R_2 &\geq& H(X_2|X_1) - \frac{1}{n}\log \frac{1}{1-\epsilon} \nonumber \\
R_1+R_2 &\geq& H(X_1, X_2) - \frac{1}{n}\log \frac{1}{1-\epsilon} \bigg\}. 
\label{case-1-conv-1}
\end{IEEEeqnarray}
\begin{thm} \label{thm-zero-var-1}
When $V(X_1|X_2) = 0$, $V(X_2|X_1) = 0$, and $V(X_1,X_2) = 0$, the $(n,\epsilon)$-rate region $\mathscr{R}^{*}(n,\epsilon)$ satisfies
\begin{equation}
\mathscr{R}_{\rm in}^{(1)}(n,\epsilon) \subseteq \mathscr{R}^{*}(n,\epsilon) \subseteq \mathscr{R}_{\rm out}^{(1)}(n,\epsilon).
\end{equation}
\end{thm}

As in the point-to-point scenario, there are no second-order dispersion terms or $-\frac{\log n}{2n}$ third-order terms in the characterization of $\mathscr{R}^{*}(n,\epsilon)$ in this case. For any $n$ and $\epsilon$, the achievable region $\mathscr{R}_{\rm in}^{(1)}(n,\epsilon)$ has a curved boundary due to the trade-off in the $O\left(\frac{1}{n}\right)$ fourth-order terms, while the converse region $\mathscr{R}_{\rm out}^{(1)}(n,\epsilon)$ has three linear boundaries.

\textit{Case 2)}: There are three possible cases where exactly two of the three varentropies are equal to zero. Here, we suppose that $V(X_1|X_2) > 0$ while $V(X_2|X_1) = V(X_1,X_2) = 0$. The other two cases can be analyzed in the same way. Let $B_1$ denote the Berry-Esseen constant for the random variable $\imath(X_1|X_2)$, and let $S_2$, $K_1$, $\bar K_1$ 
be the finite positive constants defined in \eqref{eq:T2}, \eqref{eq-def-K1}, and \eqref{eq-def-K1bar}, respectively. For any $\delta_1$, $\delta_2$, $\delta_{12} > 0$, let
\begin{IEEEeqnarray}{rCl}
\IEEEeqnarraymulticol{3}{l}{\hat{\mathscr{R}}_{\rm in}^{(2)}(n,\delta_1, \delta_2, \delta_{12}) \triangleq \bigg\{(R_1,R_2) \in \mathbb{R}^2:} \nonumber \\
R_1 &\geq& H(X_1|X_2) \notag\\
&&+ \sqrt{\frac{V(X_1|X_2)}{n}}Q^{-1}\left(\delta_1 - \frac{B_1+K_1}{\sqrt{n}} - \frac{S_2}{n} \right) \nonumber \\
&& - \frac{\log n}{2n} + \frac{1}{n}\log \frac{\bar K_1}{1-\delta_2-\delta_{12}} \nonumber \\
R_2 &\geq& H(X_2|X_1) + \frac{1}{n}\log \frac{1}{\delta_2} \nonumber \\
R_1+R_2 &\geq& H(X_1,X_2) + \frac{1}{n}\log \frac{1}{\delta_{12}}  \bigg\}. \label{case-2-achiev-2}
\end{IEEEeqnarray}
Define
\begin{IEEEeqnarray}{rCl}
\mathscr{R}_{\rm in}^{(2)}(n,\epsilon) &\triangleq& \bigcup\limits_{\substack{{\delta_1, \delta_2, \delta_{12} > 0}\\{\delta_1+\delta_2+\delta_{12} = \epsilon}}} \hat{\mathscr{R}}_{\rm in}^{(2)}(n, \delta_1, \delta_2, \delta_{12}) \label{case-2-achiev-2in}\\
\mathscr{R}_{\rm out}^{(2)}(n,\epsilon) &\triangleq& \bigg\{(R_1,R_2) \in \mathbb{R}^2: \vphantom{\frac{1}{n}} \nonumber \\
R_1 &\geq& H(X_1|X_2) + \sqrt{\frac{V(X_1|X_2)}{n}}Q^{-1}\left(\epsilon + \frac{B_1+1}{\sqrt{n}}\right) \nonumber \\
&&- \frac{\log n}{2n} \nonumber \\
R_2 &\geq& H(X_2|X_1) - \frac{1}{n}\log \frac{1}{1-\epsilon} \nonumber \\
R_1+R_2 &\geq& H(X_1, X_2) - \frac{1}{n}\log \frac{1}{1-\epsilon} \bigg\}.
\label{case-2-conv-5}
\end{IEEEeqnarray}
\begin{thm} \label{thm-zero-var-2}
When $V(X_1|X_2) > 0$, $V(X_2|X_1) = 0$, and $V(X_1,X_2) = 0$, the $(n,\epsilon)$-rate region $\mathscr{R}^{*}(n,\epsilon)$ satisfies
\begin{equation}
\mathscr{R}_{\rm in}^{(2)}(n,\epsilon) \subseteq \mathscr{R}^{*}(n,\epsilon) \subseteq \mathscr{R}_{\rm out}^{(2)}(n,\epsilon).
\end{equation}
\end{thm}

The achievable region $\mathscr{R}_{\rm in}^{(2)}(n,\epsilon)$ has a curved boundary due to the trade-off in $\delta_1$, $\delta_2$, and $\delta_{12}$. If we let
\begin{align}
\delta_1 = \epsilon - \frac{2}{\sqrt{n}}, \, \delta_2 = \frac{1}{\sqrt{n}}, \, \delta_{12} = \frac{1}{\sqrt{n}},
\end{align} then it is apparent that the dispersion corresponding to $R_1$ is $V(X_1|X_2)$ with a $-\frac{\log n}{2n}$ third-order term, while the dispersions of $R_2$ and $R_1 + R_2$ are zero. 

\textit{Case 3)}: Similar to Case 2), there are three possible cases where exactly one of the three varentropies is equal to zero. Here, we consider the case where $V(X_1|X_2) = 0$ while $V(X_2|X_1) > 0$ and $V(X_1,X_2) > 0$. Let $S_1$, $K_2$, $\bar K_2$, and $K_{12}$ be the finite positive constants defined in \eqref{eq:T1}, \eqref{eq-def-K2}, \eqref{eq-def-K2bar}, and \eqref{eq-def-K12}, respectively, and let $B$ be the Bentkus constant \eqref{eq-lem-b-e} for the vector $(I_2, I_{12})$. For any $\delta \in (0, \epsilon)$, let 
\begin{align}
& \hat{\mathscr{R}}_{\rm in}^{(3)}(n, \delta) \triangleq \bigg\{(R_1, R_2) \in \mathbb{R}^2:
R_1 \geq H(X_1|X_2) + \frac{1}{n}\log \frac{1}{\delta} \nonumber \\
& \begin{bmatrix}
R_2 \\
R_1 + R_2
\end{bmatrix} \in 
\begin{bmatrix}
H(X_2|X_1) \\
H(X_1,X_2)
\end{bmatrix}  + \frac{1}{\sqrt{n}}\mathscr{Q}_{\rm inv}\left(\mathsf{V}_2,\epsilon-\delta-\frac{C_{\rm in}}{\sqrt{n}}\right) \notag\\
&- \frac{\log n}{2n}\mathbf{1} + \frac{1}{n} \log \frac{1}{1-\delta}\mathbf{1} + \frac{1}{n}\begin{bmatrix}
\log 2\bar K_2 \\
\log 2K_{12}
\end{bmatrix} \bigg\},  \label{case-3-achiev-4} 
\end{align} where $C_{\rm in} \triangleq K_2 + K_{12} + B + \frac{S_1}{\sqrt n}$, and $\mathsf{V}_2$ is the covariance matrix of the random vector $(\imath(X_2|X_1), \imath(X_1, X_2))$. Define
\begin{IEEEeqnarray}{rCl}
\mathscr{R}_{\rm in}^{(3)}(n, \epsilon) &\triangleq& \bigcup\limits_{\delta \in (0, \epsilon)}\hat{\mathscr{R}}_{\rm in}^{(3)}(n, \delta) \label{case-3-achiev-7} \\
\mathscr{R}_{\rm out}^{(3)}(n, \epsilon) &\triangleq& \bigg\{(R_1, R_2) \in \mathbb{R}^2: \nonumber \\
R_1 &\geq& H(X_1|X_2) + \frac{1}{n}\log \frac{1}{1-\epsilon} \nonumber \\
\begin{bmatrix}
R_2 \\
R_1 + R_2
\end{bmatrix} &\in&
\begin{bmatrix}
H(X_2|X_1) \\
H(X_1,X_2)
\end{bmatrix} + \frac{1}{\sqrt{n}}\mathscr{Q}_{\rm inv}\left(\textsf{V}_2,\epsilon+\frac{B + 2}{\sqrt{n}}\right) \nonumber \\
&&- \frac{\log n}{2n}\mathbf{1}\bigg\}.
\label{case-3-conv-3}
\end{IEEEeqnarray} 
\begin{thm} \label{thm-zero-var-3}
When $V(X_1|X_2) = 0$, $V(X_2|X_1) > 0$, and $V(X_1,X_2) > 0$, the $(n,\epsilon)$-rate region $\mathscr{R}^{*}(n,\epsilon)$ satisfies
\begin{equation}
\mathscr{R}_{\rm in}^{(3)}(n,\epsilon) \subseteq \mathscr{R}^{*}(n,\epsilon) \subseteq \mathscr{R}_{\rm out}^{(3)}(n,\epsilon).
\end{equation}
\end{thm}

For any $n$ and $\epsilon$, the achievable region $\mathscr{R}_{\rm in}^{(3)}(n, \epsilon)$ has a curved boundary that is characterized by the trade-off between a separate bound on $R_1$ and a region in $\mathbb{R}^2$ that bounds $(R_2, R_1+R_2)$ jointly. The converse region $\mathscr{R}_{\rm out}^{(3)}(n, \epsilon)$ is the intersection of a region with a linear boundary that bounds $R_1$ only and a region with a curved boundary that bounds $(R_2, R_1+R_2)$ jointly. If we let
\begin{align}
\delta = \frac{1}{\sqrt{n}},
\end{align} then it is apparent that the dispersion corresponding to $R_2$ and $R_1 + R_2$ is given by $\mathsf{V}_2$ with a $-\frac{\log n}{2n}$ third-order term, while the dispersion of $R_1$ is zero.

A less redundant stationary, memoryless source has some useful properties. When $V(X_1,X_2) = 0$, \begin{equation}
P_{X_1^n X_2^n}(x_1^n, x_2^n) \in \left\{0, \, \exp\left(-nH(X_1,X_2)\right) \right\},
\end{equation} for every $(x_1^n, x_2^n) \in \mathcal{X}_1^n \times \mathcal{X}_2^n$; in other words, $(X_1,X_2)$ is uniformly distributed over its support in $\mathcal{X}_1 \times \mathcal{X}_2$. 
When $V(X_1|X_2) = 0$, 
\begin{equation}
P_{X_1^n |X_2^n}(x_1^n |x_2^n) = 
\begin{cases}
\exp\left(-nH(X_1|X_2)\right), \\
\qquad \; \text{ if } P_{X_1^n X_2^n}(x_1^n, x_2^n) > 0 \\
0, \quad\;\;\, \text{otherwise}; \label{eq-zero-var}
\end{cases}
\end{equation} in other words, $X_1$ is uniformly distributed over its conditional support for each $x_2 \in \mathcal{X}_2$. When $V(X_2|X_1) = 0$, a result analogous to \eqref{eq-zero-var} holds. These properties do not reduce the difficulty of characterizing the optimal MASCs in general. As a result, we continue to employ the random coding techniques from Section \ref{sec-result-SW} in our analysis here. For the achievability argument, we invoke the MASC RCU bound (Theorem \ref{thm-sw-rcu}); 
for the converse, we appeal to a modified version of \cite[Lemma~7.2.2]{han}, as stated below. 
\begin{lem}[Modified {\cite[Lemma~7.2.2]{han}}] \label{lem-h-modified}
	Any $\left(n,\exp\left(nR_1\right),\exp\left(nR_2\right),\epsilon'\right)$ MASC satisfies 
	\begin{IEEEeqnarray}{rCl}
	\epsilon' &\geq& \mathbb{P} \bigg[\left\{\frac{1}{n}I_{1} \geq R_1 + \gamma_1\right\} \cup \left\{\frac{1}{n}I_{2} \geq R_2 + \gamma_2 \right\} \cup \label{eq-lem-h-modified}  \\ 
	&& \phantom{\mathbb P \Bigg[}\left\{\frac{1}{n}I_{12} \geq R_1 + R_2 +\gamma_{12}\right\}\bigg] \nonumber \\
	&& - \min\left\{\mathbb{P}\left[\frac{1}{n}I_{1} \geq R_1 + \gamma_1\right], \, \exp\left(-n\gamma_1\right) \right\} \nonumber \\
	&& - \min\left\{\mathbb{P}\left[\frac{1}{n}I_{2} \geq R_2 + \gamma_2\right], \, \exp\left(-n\gamma_2\right) \right\} \nonumber\\
	&& - \min\left\{\mathbb{P}\left[\frac{1}{n}I_{12} \geq R_1 + R_2 + \gamma_{12}\right], \, \exp\left(-n\gamma_{12}\right) \right\},
	\nonumber
	\end{IEEEeqnarray} for any $\gamma_1$, $\gamma_2$, $\gamma_{12} > 0$, where $I_{1}$, $I_{2}$ and $I_{12}$ are defined in \eqref{eq-def-I1n}--\eqref{eq-def-In}.
\end{lem}

We next prove Theorems~\ref{thm-zero-var-1}, \ref{thm-zero-var-2}, and \ref{thm-zero-var-3}.

\begin{proof}[Proof of Theorem~\ref{thm-zero-var-1}]
\emph{Achievability}: We employ the RCU bound in \eqref{eq-sw-achiev-1}. To evaluate the terms in \eqref{eq-sw-achiev-1}, note that the uniformity over the distribution's support that results from $V(X_1,X_2) = V(X_1|X_2) = V(X_2|X_1) = 0$ implies that for any $(x_1^n, x_2^n)$ such that $P_{X_1^nX_2^n}(x_1^n, x_2^n) > 0$,
\begin{align}
A_1 M_1 = \exp\left(nH(X_1|X_2)\right) \quad \text{a.s.} \label{case-1-achiev-4}
\end{align} 
Similar equalities hold for $A_2$ and $A_{12}$, and for any  $(R_1, R_2) \in \mathscr{R}_{\rm in}^{(1)}(n,\epsilon)$, \eqref{eq-sw-achiev-1} gives
\begin{equation}
\epsilon' \leq \delta_1 + \delta_2 + \delta_{12} = \epsilon,
\end{equation} implying that such a rate pair $(R_1, R_2)$ is achievable. Therefore, the $(n,\epsilon)$-rate region satisfies
\begin{equation}
\mathscr{R}^{*}(n,\epsilon) \supseteq \mathscr{R}_{\rm in}^{(1)}(n,\epsilon).
\end{equation} 

\emph{Converse}: Consider any $(R_1, R_2)$ with $R_1 < H(X_1|X_2) - \frac{1}{n}\log \frac{1}{1-\epsilon}$. Since the bound in \eqref{eq-lem-h-modified} holds for any $\gamma_1$, $\gamma_2$, $\gamma_{12} > 0$, we take 
\begin{equation}
\gamma_1 = H(X_1|X_2) - R_1 > \frac{1}{n}\log \frac{1}{1-\epsilon},
\end{equation} which, under the given uniformity, implies
\begin{equation}
\mathbb{P} \left[\frac{1}{n}I_{1} \geq R_1 + \gamma_1\right] = 1.
\end{equation} We take $\gamma_2$ and $\gamma_{12}$ sufficiently large so that 
\begin{IEEEeqnarray}{rCl}
R_2 + \gamma_2 &>& H(X_2|X_1) \\
R_1 + R_2 + \gamma_{12} &>& H(X_1, X_2)
\end{IEEEeqnarray} and hence
\begin{equation}
\mathbb{P}\left[\frac{1}{n}I_{2} \geq R_2 + \gamma_{2}\right] = \mathbb{P}\left[\frac{1}{n}I_{12} \geq R_1 + R_2 + \gamma_{12}\right] = 0. 
\end{equation} Under these conditions, \eqref{eq-lem-h-modified} gives 
\begin{equation}
\epsilon' \geq 1 - \exp\left(-n\gamma_1\right) > 1 - (1-\epsilon) = \epsilon.
\end{equation} Therefore, any achievable rate pair $(R_1, R_2)$ must satisfy \begin{equation}
R_1 \geq H(X_1|X_2) - \frac{1}{n}\log \frac{1}{1-\epsilon}.
\end{equation} The same analysis applies to $R_2$ and $R_1 + R_2$. We then conclude that any achievable rate pair $(R_1, R_2)$ must satisfy $(R_1, R_2) \in \mathscr{R}_{\rm out}^{(1)}(n,\epsilon)$. 
Thus, 
\begin{equation}
\mathscr{R}^{*}(n,\epsilon) \subseteq \mathscr{R}_{\rm out}^{(1)}(n,\epsilon).
\end{equation}
\end{proof}

\begin{proof}[Proof of Theorem~\ref{thm-zero-var-2}]
\emph{Achievability}: Take $(R_1,R_2) \in \mathscr{R}_{\rm in}^{(2)}(n,\epsilon)$ \eqref{case-2-achiev-2in} satisfying the inequalities in \eqref{case-2-achiev-2} with some $\delta_1$, $\delta_2$, $\delta_{12} > 0$ such that $\delta_1 + \delta_2 + \delta_{12} = \epsilon$.  We again employ the RCU bound from \eqref{eq-sw-achiev-1}. Since $\mathbb E \left[V_c(X_1|X_2)\right] > 0$, we use \eqref{eq:A1baru} to bound $A_1$. The terms $A_2$ and $A_{12}$ are constants (cf.~\eqref{case-1-achiev-4}). With these observations, we weaken \eqref{eq-sw-achiev-2} as 
\begin{align}
\epsilon' &\leq \mathbb E \left[   A_1 1 \left\{  A_1 \leq 1 \right\} \right] + A_2 + A_{12} +  \mathbb P \left[  \bar {\bar A}_1   > 1 -  A_2 -  A_{12} \right] \notag\\
&\phantom{=}+ \mathbb P \left[ \mathcal S_2^c \right] \\
&\leq \mathbb E \left[   A_1 1 \left\{  A_1 \leq 1 \right\} \right] + \delta_2 + \delta_{12} + \mathbb P \left[  \bar {\bar A}_1   > 1 -  \delta_2 -  \delta_{12} \right] \notag\\
&\phantom{=}+ \mathbb P \left[ \mathcal S_2^c \right] \label{eq:zero-var-2a}\\
&\leq \delta_1 +\delta_2 + \delta_{12} , \label{eq:zero-var-2b}
\end{align}
where \eqref{eq:zero-var-2a} is by our choice of $(R_1,R_2)$, and \eqref{eq:zero-var-2b} applies \eqref{eq:sec1}, the Berry-Esseen inequality (Theorem \ref{thm-berry-esseen}), and \eqref{eq:cheb} to bound the three probability terms. 
Therefore, $(R_1, R_2)$ is achievable at blocklength $n$ and error probability $\epsilon$, implying  
\begin{equation}
\mathscr{R}^{*}(n,\epsilon) \supseteq \mathscr{R}_{\rm in}^{(2)}(n,\epsilon).
\end{equation} 

\emph{Converse}: We next apply Lemma~\ref{lem-h-modified} 
to derive a converse result. Recall that under our assumptions $V(X_2|X_1) = V(X_1,X_2) = 0$, $\imath(X_2|X_1) = H(X_2|X_1)$ and $\imath(X_1,X_2) = H(X_1,X_2)$ almost surely. Consider any $(R_1, R_2)$ such that $R_2 < H(X_2|X_1) - \frac{1}{n}\log\frac{1}{1-\epsilon}$. Since the bound in \eqref{eq-lem-h-modified} holds for any $\gamma_1$, $\gamma_2$, $\gamma_{12} > 0$, we can take 
\begin{equation}
\gamma_2 = H(X_2|X_1) - R_2 > \frac{1}{n}\log\frac{1}{1-\epsilon}
\end{equation}
so that
\begin{equation}
\mathbb{P} \left[\frac{1}{n}I_{2} \geq R_2 + \gamma_2\right] = 0.
\end{equation} By this choice of $\gamma_2$, 
$1 - \epsilon - \exp\left(-n\gamma_2\right) > 0$. Thus, we can take $\gamma_1$ and $\gamma_{12}$ sufficiently large such that 
\begin{equation}
\exp\left(-n\gamma_1\right) + \exp\left(-n\gamma_{12}\right) < 1-\epsilon- \exp\left(-n\gamma_2\right).
\end{equation} 
By the above choices of $\gamma_1$, $\gamma_2$, and $\gamma_{12}$, \eqref{eq-lem-h-modified} gives
\begin{equation}
\epsilon' \geq 1 - \exp\left(-n\gamma_1\right) - \exp\left(-n\gamma_2\right) - \exp\left(-n\gamma_{12}\right) > \epsilon.
\end{equation} Therefore, any achievable rate pair $(R_1, R_2)$ must satisfy 
\begin{equation}
R_2 \geq H(X_2|X_1) - \frac{1}{n}\log \frac{1}{1-\epsilon}. \label{case-2-conv-2}
\end{equation} The same analysis applies to $R_1+R_2$, and we conclude that any achievable rate pair $(R_1, R_2)$ must also satisfy 
\begin{equation}
R_1+R_2 \geq H(X_1,X_2) - \frac{1}{n}\log \frac{1}{1-\epsilon}.\label{case-2-conv-3}
\end{equation}
Given \eqref{case-2-conv-2} and \eqref{case-2-conv-3}, 
we re-evaluate the bound in \eqref{eq-lem-h-modified} by taking 
\begin{IEEEeqnarray}{rCl}
\!\!\!\!\!\!\!\!\!\!\!\! \gamma_1 &=& \frac{\log n}{2n},~~ \gamma_2 > \frac{1}{n}\log \frac{1}{1-\epsilon},~~ \gamma_{12} > \frac{1}{n}\log \frac{1}{1-\epsilon}.
\end{IEEEeqnarray} Under these conditions, we have
\begin{equation}
\!\!\!\! \mathbb{P}\left[\frac{1}{n}I_{2} \geq R_2 + \gamma_2\right] = \mathbb{P}\left[\frac{1}{n}I_{12} \geq R_1 + R_2 + \gamma_{12}\right] = 0,
\end{equation} and the bound in \eqref{eq-lem-h-modified} becomes
\begin{equation}
\epsilon' \geq \mathbb{P}\left[\frac{1}{n}I_{1} \geq R_1 + \frac{\log n}{2n} \right] - \frac{1}{\sqrt{n}}. \label{case-2-conv-4}
\end{equation} Then, by the Berry-Esseen inequality (Theorem \ref{thm-berry-esseen}), taking 
\begin{equation}
R_1 = H(X_1|X_2) + \sqrt{\frac{V(X_2|X_1)}{n}}Q^{-1}\left(\epsilon + \frac{B_1+1}{\sqrt{n}}\right) - \frac{\log n}{2n}
\end{equation} in \eqref{case-2-conv-4} yields $\epsilon' \geq \epsilon$. 
Therefore, any achievable rate pair $(R_1, R_2)$ 
must satisfy $(R_1, R_2) \in \mathscr{R}_{\rm out}^{(2)}(n,\epsilon)$. Thus, 
\begin{equation}
\mathscr{R}^{*}(n,\epsilon) \subseteq \mathscr{R}_{\rm out}^{(2)}(n,\epsilon).
\end{equation} 
\end{proof}

\begin{proof}[Proof of Theorem~\ref{thm-zero-var-3}]
\emph{Achievability}: Take any $(R_1,R_2) \in \mathscr{R}_{\rm in}^{(3)}(n, \epsilon)$ satisfying the inequalities in \eqref{case-3-achiev-4} with some $\delta \leq \epsilon$. We employ the RCU bound in \eqref{eq-sw-achiev-1}. Since $\mathbb E \left[V_c(X_2|X_1)\right] > 0 $ and $V(X_1,X_2) > 0$, we use \eqref{eq:A2baru} and \eqref{eq:A12bar} to bound $A_2$ and $A_{12}$, respectively; $A_1$ is the constant in \eqref{case-1-achiev-4}. With these observations, we weaken \eqref{eq-sw-achiev-2} as 
\begin{align}
\epsilon' &\leq 
A_1 + \mathbb E \left[   A_2 1 \left\{  A_2 \leq 1 \right\} \right] + \mathbb E \left[   A_{12} 1 \left\{  A_{12} \leq 1 \right\} \right] \notag\\
&\phantom{=} + \mathbb P \left[ 2  \bar {\bar A}_2 > 1 - A_1 \cup  2  \bar A_{12}  > 1 - A_1 \right] + \mathbb P[\mathcal S_1^c] \\
&\leq \delta + \mathbb E \left[   A_2 1 \left\{  A_2 \leq 1 \right\} \right] + \mathbb E \left[   A_{12} 1 \left\{  A_{12} \leq 1 \right\} \right] \notag\\
&\phantom{=} + \mathbb P \left[ 2  \bar {\bar A}_2 > 1 - \delta \cup  2  A_{12}  > 1 - \delta \right]  + \mathbb P[\mathcal S_1^c] \label{eq:zero-var-3a}\\
&\leq \epsilon, \label{eq:zero-var-3b}
\end{align}
where \eqref{eq:zero-var-3a} is by our choice of $(R_1,R_2)$, and \eqref{eq:zero-var-3b} applies \eqref{eq:sec2}, \eqref{eq:sec12}, Lemma \ref{lem-b-e} (multidimensional Berry-Esseen Theorem), and \eqref{eq:cheb} to bound the four probability terms. Therefore,  $(R_1, R_2)$  is achievable at blocklength $n$ and error probability $\epsilon$, implying that 
\begin{equation}
\mathscr{R}^{*}(n, \epsilon) \supseteq \mathscr{R}_{\rm in}^{(3)}(n, \epsilon). 
\end{equation} 

\emph{Converse}:
We employ Lemma \ref{lem-h-modified} to derive a converse. 
Recall that in this case, $\imath(X_1|X_2) = H(X_1|X_2)$ almost surely. Consider any $(R_1, R_2)$ such that $R_1 < H(X_1|X_2) - \frac{1}{n}\log \frac{1}{1-\epsilon}$. Since the bound in \eqref{eq-lem-h-modified} holds for any $\gamma_1$, $\gamma_2$, $\gamma_{12} > 0$, we can set
\begin{equation}
\gamma_{1} = H(X_1|X_2) - R_1 >  \frac{1}{n}\log \frac{1}{1-\epsilon}
\end{equation} so that
\begin{equation}
\mathbb{P} \left[\frac{1}{n}I_{1} \geq R_1 + \gamma_{1} \right] = 1.
\end{equation} 
By this choice of $\gamma_1$, $1-\epsilon-\exp\left(-n\gamma_1\right) > 0$. 
Thus, we can take $\gamma_2$ and $\gamma_{12}$ sufficiently large such that
\begin{equation}
\exp\left(-n\gamma_2\right) + \exp\left(-n\gamma_{12}\right) < 1 - \epsilon - \exp\left(-n\gamma_1\right).
\end{equation} By the above choices of $\gamma_1$, $\gamma_2$, and $\gamma_{12}$, \eqref{eq-lem-h-modified} gives
\begin{equation}
\epsilon' \geq 1 - \exp\left(-n\gamma_1\right) - \exp\left(-n\gamma_2\right) - \exp\left(-n\gamma_{12}\right) > \epsilon. 
\end{equation} Therefore, any achievable rate pair $(R_1, R_2)$ must satisfy
\begin{equation}
R_1 \geq H(X_1|X_2) - \frac{1}{n}\log \frac{1}{1-\epsilon}. \label{case-3-conv-1}
\end{equation} Given that \eqref{case-3-conv-1} holds, we re-evaluate the bound in \eqref{eq-lem-h-modified} by taking
\begin{IEEEeqnarray}{rCl}
\gamma_1 &>& \frac{1}{n}\log \frac{1}{1-\epsilon}, ~~
\gamma_2 =\frac{\log n}{2n}, ~~
\gamma_{12} = \frac{\log n}{2n}.
\end{IEEEeqnarray} Under these conditions, the bound in \eqref{eq-lem-h-modified} becomes
\begin{IEEEeqnarray}{rCl}
\epsilon' &\geq& \mathbb{P}\bigg[\left\{\frac{1}{n}I_{2} \geq R_2 + \frac{\log n}{2n}\right\} \cup \nonumber \\
&& \phantom{\mathbb P \bigg[}\left\{\frac{1}{n}I_{12} \geq R_1 + R_2 + \frac{\log n}{2n}\right\} \bigg] - \frac{2}{\sqrt{n}}. \label{case-3-conv-2} 
\end{IEEEeqnarray} Applying Lemmas~\ref{lem-b-e} and \ref{lem-sve} to \eqref{case-3-conv-2}, we conclude that any $(R_1, R_2)$ in the $(n,\epsilon)$-rate region must satisfy $(R_1, R_2) \in \mathscr{R}_{\rm out}^{(3)}(n, \epsilon)$. Thus,
\begin{equation}
\mathscr{R}^{*}(n, \epsilon) \subseteq \mathscr{R}_{\rm out}^{(3)}(n, \epsilon).
\end{equation} 
\end{proof}

\subsection{Two Special Cases}
\label{append-zero-var-b}
The analysis in Section \ref{append-zero-var-a} above applies to any stationary, memoryless source with single-letter distribution $P_{X_1X_2}$ that satisfies \eqref{assump-b2}. In such a general setting, it is hard to find an optimal code. However, there are some special cases in which the optimal codes for a less redundant source can be characterized. 

To enable the following analysis on these special cases, we assume that $P_{X_1X_2}(x_1,x_2) > 0$ for every $(x_1,x_2) \in \mathcal{X}_1 \times \mathcal{X}_2$. 
Under this assumption, 
$V(X_1,X_2) = 0$ if and only if $V(X_1 |X_2) = V(X_2 |X_1) = 0$. As a result, the three cases discussed in Section A reduce to only \emph{two} possible scenarios:
\begin{enumerate}
\item[1)] $V(X_1,X_2) = V(X_1 |X_2) = V(X_2 |X_1) = 0$;
\item[2)] $V(X_1,X_2) >0$, and either $V(X_1 |X_2) = 0$ or $V(X_2 |X_1) = 0$.
\end{enumerate}
Note that $X_1$ and $X_2$ are independent in both of these scenarios.

We first summarize the results below.

\textit{Special Case 1)}:
\begin{thm}\label{thm-zero-var-s-1}
Suppose that $V(X_1 |X_2) = 0$, $V(X_2 |X_1) = 0$, and $V(X_1,X_2) = 0$. If $P_{X_1X_2}$ satisfies $P_{X_1X_2}(x_1,x_2) > 0$ for every $(x_1,x_2) \in \mathcal{X}_1 \times \mathcal{X}_2$, 
then
\begin{equation}
\mathscr{R}^{*}(n,\epsilon) = \mathscr{R}_{\rm out}^{(1)}(n,\epsilon),
\end{equation} where $\mathscr{R}_{\rm out}^{(1)}(n,\epsilon)$ is defined in \eqref{case-1-conv-1}.
\end{thm}

This scenario is a special example of Case 1) discussed in Section A above. The $(n,\epsilon)$-rate region here coincides with the converse region $\mathscr{R}_{\rm out}^{(1)}(n,\epsilon)$ presented in \eqref{case-1-conv-1} for general source distributions. See Figure \ref{fig-sw}\protect\subref{fig-sw-1} for a comparison among $\mathscr{R}_{\rm in}^{(1)}(n,\epsilon)$, $\mathscr{R}_{\rm out}^{(1)}(n,\epsilon)$, and $\mathscr{R}^{*}(n,\epsilon)$ in this special case.

\textit{Special Case 2)}: With $V(X_1, X_2) > 0$, we here assume that $V(X_1|X_2) = 0$ and $V(X_2|X_1) > 0$. The other case can be analyzed similarly. For any $\delta \in [0, \epsilon)$, we define
\begin{IEEEeqnarray}{rCl}
&& \hat{\mathscr{R}}_{\rm in}^{s}(n,\delta) \triangleq \bigg\{(R_1,R_2) \in \mathbb{R}^2:  \label{b-case-2-1} \\
R_1 &\geq& H(X_1) - \frac{1}{n}\log \frac{1}{1-\delta} \nonumber \\
 R_2 &\geq& H(X_2) + \sqrt{\frac{V(X_2)}{n}}Q^{-1}\left(\frac{\epsilon-\delta}{1-\delta}\right) - \frac{\log n}{2n} + \xi_{\rm in}(\epsilon,\delta,n) \bigg\} \notag\\
&&\hat{\mathscr{R}}_{\rm out}^{s}(n,\delta) \triangleq \bigg\{(R_1,R_2) \in \mathbb{R}^2: \\
R_1 &\geq& H(X_1) - \frac{1}{n}\log \frac{1}{1-\delta} \nonumber \\
R_2 &\geq& H(X_2) + \sqrt{\frac{V(X_2)}{n}}Q^{-1}\left(\frac{\epsilon-\delta}{1-\delta}\right)- \frac{\log n}{2n} \notag\\
&& - \xi_{\rm out}(\epsilon,\delta,n) \bigg\}, \notag
\end{IEEEeqnarray} 
where the functions $\xi_{\rm in}(\epsilon,\delta,n)$ and $\xi_{\rm out}(\epsilon,\delta,n)$ are characterized as follows.  
For any fixed $\delta$, $\xi_{\rm out}(\epsilon,\delta,n) = O(\frac{1}{n})$ 
and $\xi_{\rm in}(\epsilon,\delta,n) = O(\frac{1}{n})$. 
For any fixed $n$, both $\xi_{\rm out}(\epsilon,\delta,n)$ and $\xi_{\rm in}(\epsilon,\delta,n)$ 
blow up as $\delta$ approaches $\epsilon$. 
(These bounds are applications of the point-to-point results in Theorem~\ref{thm-k-v}.)
Also define
\begin{IEEEeqnarray}{rCl}
\mathscr{R}_{\rm in}^{s}(n,\epsilon) &\triangleq& \bigcup\limits_{\delta \in [0, \epsilon)} \hat{\mathscr{R}}_{\rm in}^{s}(n,\delta) \label{b-case-2-2} \\
\mathscr{R}_{\rm out}^{s}(n,\epsilon) &\triangleq& \bigcup\limits_{\delta \in [0, \epsilon)} \hat{\mathscr{R}}_{\rm out}^{s}(n,\delta).
\end{IEEEeqnarray} 
\begin{thm}\label{thm-zero-var-s-2}
Suppose that $V(X_1|X_2) = 0$, $V(X_2|X_1) > 0$, and $V(X_1,X_2) > 0$. If $P_{X_1X_2}$ satisfies $P_{X_1X_2}(x_1,x_2) > 0$ for every $(x_1,x_2) \in \mathcal{X}_1 \times \mathcal{X}_2$, then 
\begin{equation}
\mathscr{R}_{\rm in}^{s}(n,\epsilon) \subseteq \mathscr{R}^{*}(n,\epsilon) \subseteq \mathscr{R}_{\rm out}^{s}(n,\epsilon). \label{eq-case-2-2}
\end{equation}
\end{thm}

This scenario is a special example of Case 3) discussed in Section A of this appendix. The $(n,\epsilon)$-rate region characterized in \eqref{eq-case-2-2} is sandwiched between the achievable region presented in \eqref{case-3-achiev-7} and the converse region presented in \eqref{case-3-conv-3}. To compare these regions, we note that the bounds on $R_1 + R_2$ in \eqref{case-3-achiev-7} and \eqref{case-3-conv-3} become inactive in this special scenario where $X_1$ and $X_2$ are independent. As a result, the achievable region in \eqref{case-3-achiev-7} becomes
\begin{IEEEeqnarray}{rCl}
\mathscr{R}_{\rm in}^{(3)}(n,\epsilon) &=& \bigcup\limits_{\delta \in (0, \epsilon)} \bigg\{(R_1,R_2) \in \mathbb{R}^2:  \\
R_1 &\geq& H(X_1) + \frac{1}{n}\log \frac{1}{\delta} \nonumber \\
R_2 &\geq& H(X_2) + \sqrt{\frac{V(X_2)}{n}}Q^{-1}\left(\epsilon-\delta - \frac{C_{\rm in}}{\sqrt{n}}\right) \nonumber \\
&& - \frac{\log n}{2n} + \frac{1}{n}\log \frac{1}{1-\delta} \bigg\},\notag
\end{IEEEeqnarray} and the converse region in \eqref{case-3-conv-3} becomes
\begin{IEEEeqnarray}{rCl}
\mathscr{R}_{\rm out}^{(3)}(n,\epsilon) &=& \bigg\{(R_1,R_2) \in \mathbb{R}^2:  \\
R_1 &\geq& H(X_1) - \frac{1}{n}\log \frac{1}{1-\epsilon} \nonumber \\
R_2 &\geq& H(X_2) + \sqrt{\frac{V(X_2)}{n}}Q^{-1}(\epsilon) - \frac{\log n}{2n} - O\left(\frac{1}{n}\right) \bigg\}. \nonumber 
\end{IEEEeqnarray} 
As $\delta$ approaches $\epsilon$, 
the boundary of the $(n,\epsilon)$-rate region given in \eqref{b-case-2-2} 
approaches the line $R_1 = H(X_1) - \frac{1}{n}\log \frac{1}{1-\epsilon}$, 
which matches the vertical segment of the boundary of the converse region 
$\mathscr{R}_{\rm out}^{(3)}(n,\epsilon)$. 
See Figure \ref{fig-sw}\protect\subref{fig-sw-2} for a comparison 
of $\mathscr{R}_{\rm in}^{(3)}(n,\epsilon)$, $\mathscr{R}_{\rm out}^{(3)}(n,\epsilon)$, and $\mathscr{R}^{*}(n,\epsilon)$ in this case.

\begin{figure}[!t]
	\centering
	\subfloat[]{\includegraphics[width=0.48\textwidth]{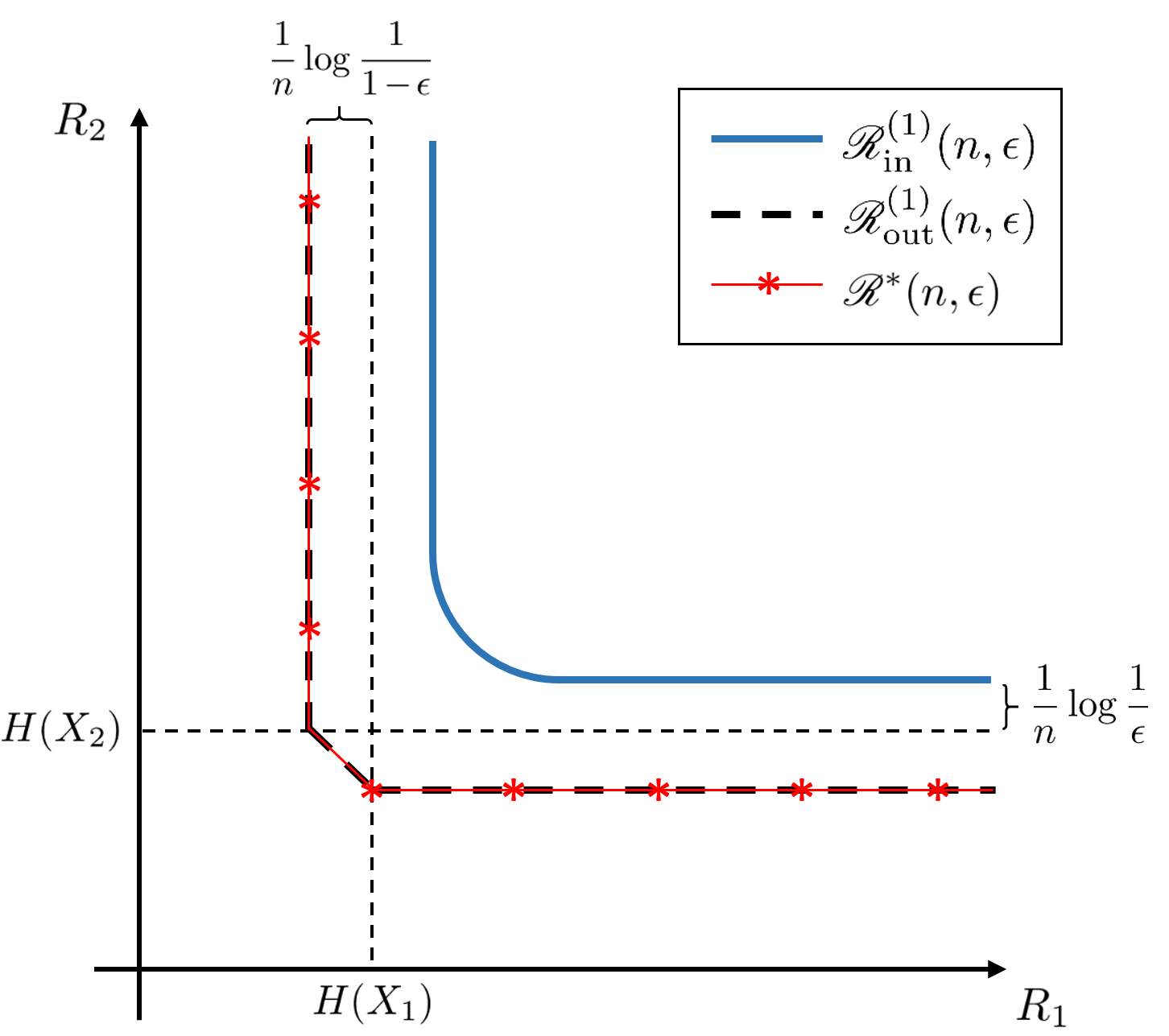}
		\label{fig-sw-1}}
	\hfil
	\subfloat[]{\includegraphics[width=0.48\textwidth]{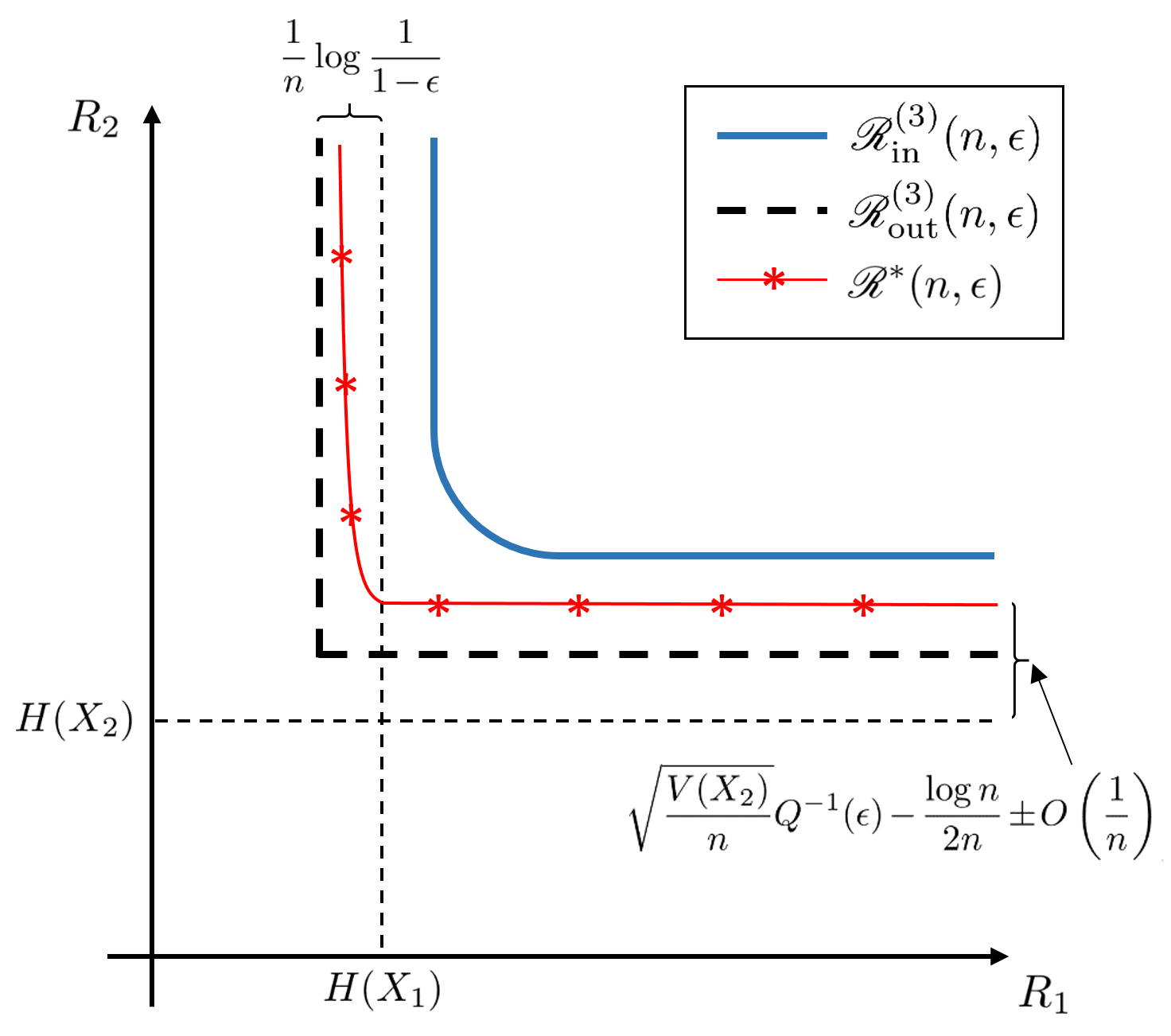}
		\label{fig-sw-2}}
	\caption{Schematic illustrations of the MASC rate regions for a less redundant source. The drawing in \protect\subref{fig-sw-1} illustrates the achievable and converse regions in Case 1) ($V(X_1,X_2) = V(X_1|X_2) = V(X_2|X_1) = 0$) and the $(n,\epsilon)$-rate region $\mathscr{R}^{*}(n,\epsilon)$ when $P_{X_1X_2}$ is assumed to have no zeros (Special Case 1)). 
	The drawing in \protect\subref{fig-sw-2} illustrates the achievable and converse regions in Case 3) ($V(X_1|X_2) = 0$, $V(X_1,X_2),V(X_2|X_1) > 0$) and the $(n,\epsilon)$-rate region $\mathscr{R}^{*}(n,\epsilon)$ when $P_{X_1X_2}$ is assumed to have no zeros (Special Case 2)). }
	\label{fig-sw}
\end{figure}

We next give proofs for Theorems~\ref{thm-zero-var-s-1} and \ref{thm-zero-var-s-2}.

\begin{proof}[Proof of Theorem~\ref{thm-zero-var-s-1}]
When $V(X_1 |X_2) = V(X_2 |X_1) = V(X_1,X_2) = 0$, $(X_1,X_2)$ is uniformly distributed over $\mathcal{X}_1 \times \mathcal{X}_2$, 
which restricts $\mathcal{X}_1$ and $\mathcal{X}_2$ to be finite and $X_1$ and $X_2$ to be independent. 
The MASC problem reduces to independent (point-to-point) almost-lossless source coding problems for the two sources with a single compound error probability. As a result, the optimal MASC with blocklength $n$ and code sizes $(M_1,M_2)$ has an error probability given by
\begin{equation}
1 - \min \left\{1, \, \frac{M_1}{|\mathcal{X}_1|^n} \right\} \cdot \min \left\{1, \, \frac{M_2}{|\mathcal{X}_2|^n} \right\}.
\end{equation}
Therefore, for any $0 < \epsilon < 1$, there exists an $(n,M_1,M_2,\epsilon)$ MASC if and only if 
\begin{align}
\min \left\{1, \, \frac{M_1}{|\mathcal{X}_1|^n} \right\} \cdot \min \left\{1, \, \frac{M_2}{|\mathcal{X}_2|^n} \right\} \geq 1 - \epsilon. \label{b-eq-zero-var-1}
\end{align}
In this case, $H(X_1) = \log |\mathcal{X}_1|$ 
and $H(X_2) = \log |\mathcal{X}_2|$.

\noindent $\bullet$  For $R_1 < H(X_1)$ and $R_2 < H(X_2)$, 
\eqref{b-eq-zero-var-1} becomes 
\begin{align}
M_1 M_2 \geq (1-\epsilon)|\mathcal{X}_1|^n |\mathcal{X}_2|^n, \label{b-eq-zero-var-2}
\end{align} 
which is equivalent to
\begin{equation}
R_1 + R_2 \geq H(X_1) + H(X_2) - \frac{1}{n}\log \frac{1}{1-\epsilon}.
\end{equation}

\noindent $\bullet$ For $R_1 \geq H(X_1)$, \eqref{b-eq-zero-var-1} becomes \begin{align}
M_2 \geq (1-\epsilon)|\mathcal{X}_2|^n,
\end{align} 
which is equivalent to
\begin{align}
R_2 \geq H(X_2) - \frac{1}{n}\log \frac{1}{1-\epsilon}.
\end{align}

\noindent $\bullet$ For $R_2 \geq H(X_2)$, \eqref{b-eq-zero-var-1} gives
\begin{equation}
R_1 \geq H(X_1) - \frac{1}{n}\log \frac{1}{1-\epsilon}.
\end{equation}
For all $0 < \epsilon < 1$ and $n \geq 1$, 
\begin{IEEEeqnarray}{rCl}
\mathscr{R}^{*}(n,\epsilon) &=& \bigg\{(R_1,R_2) \in \mathbb{R}^2:  \\
R_1 &\geq& H(X_1) - \frac{1}{n}\log \frac{1}{1-\epsilon} \nonumber \\
R_2 &\geq& H(X_2) - \frac{1}{n}\log \frac{1}{1-\epsilon} \nonumber \\
R_{1} + R_{2} &\geq& H(X_1) + H(X_2) - \frac{1}{n}\log \frac{1}{1-\epsilon} \bigg\} = \mathscr{R}_{\rm out}^{(1)}(n,\epsilon). \nonumber 
\end{IEEEeqnarray}
\end{proof}

\begin{proof}[Proof of Theorem~\ref{thm-zero-var-s-2}]
When $V(X_1|X_2) = 0$ and $V(X_2|X_1),V(X_1,X_2) > 0$, 
$X_1$ is uniformly distributed over $\mathcal{X}_1$, 
which implies that $\mathcal{X}_1$ is finite and $H(X_1) = \log |\mathcal{X}_1|$.  
In contrast, $X_2$ is not uniform over $\mathcal{X}_2$. 
Moreover, $X_1$ and $X_2$ are independent. 
The MASC problem in this case can also be resolved via independent point-to-point source coding for each of the two sources. The optimal code with blocklength $n$ and code sizes $(M_1,M_2)$ encodes $M_1$ arbitrary symbols in $\mathcal{X}_{1}^n$ and a cardinality-$M_2$ subset of $\mathcal{X}_{2}^n$ that has the largest probability with respect to $P_{X_{2}^n}$. As a result, for any $0 < \epsilon < 1$, there exists an $(M_1,M_2,\epsilon)$ MASC if and only if
\begin{equation}
(1-\delta)\delta' \geq 1-\epsilon, \label{eq-zero-var-2}
\end{equation}
where $\delta = 1- \min \left\{1, \, \frac{M_1}{|\mathcal{X}_1|^n} \right\}$ is the total marginal probability of symbols that are not encoded in $\mathcal{X}_1^n$, and $\delta'$ is the total marginal probability (with respect to $P_{X_{2}^n}$) of the encoded symbols in $\mathcal{X}_{2}^n$. Eq. \eqref{eq-zero-var-2} implicitly requires $\delta \in [0, \epsilon]$ and $\delta' \in [1-\epsilon, 1]$.

\noindent $\bullet$ For $\delta = 0$, we have 
\begin{equation}
R_1 \geq H(X_1).
\end{equation}
In this case, \eqref{eq-zero-var-2} gives 
\begin{equation}
1 - \delta' \leq \epsilon.
\end{equation}
We can apply the point-to-point almost-lossless source coding results 
from Theorem \ref{thm-k-v} to obtain
\begin{align}
&~H(X_2) + \sqrt{\frac{V(X_2)}{n}}Q^{-1}(\epsilon) - \frac{\log n}{2n} - O\left(\frac{1}{n}\right) \nonumber \\
\leq&~ R_2 \nonumber \\
\leq&~ H(X_2) + \sqrt{\frac{V(X_2)}{n}}Q^{-1}(\epsilon) - \frac{\log n}{2n} + O\left(\frac{1}{n}\right). 
\end{align}

\noindent $\bullet$ For  $0 < \delta \leq \epsilon$, we have 
\begin{equation}
R_1 = H(X_1)-\frac{1}{n}\log \frac{1}{\delta}.
\end{equation} In this case, \eqref{eq-zero-var-2} gives 
\begin{equation}
1 - \delta' \leq \frac{\epsilon-\delta}{1-\delta}.
\end{equation} We can also apply the point-to-point results to get 
\begin{align}
&~{H(X_2) + \sqrt{\frac{V(X_2)}{n}}Q^{-1}\left(\frac{\epsilon-\delta}{1-\delta}\right) - \frac{\log n}{2n} - \xi_{\rm out}(\epsilon,\delta,n)} \nonumber \\
\leq&~ R_2 \label{eq-zero-var-3} \\
\leq&~ H(X_2) + \sqrt{\frac{V(X_2)}{n}}Q^{-1}\left(\frac{\epsilon-\delta}{1-\delta}\right) - \frac{\log n}{2n} + \xi_{\rm in}(\epsilon,\delta,n), \nonumber 
\end{align} where for any fixed $\delta$, $\xi_{\rm out}(\epsilon,\delta,n) = O\left(\frac{1}{n}\right)$ and $\xi_{\rm in}(\epsilon,\delta,n) = O\left(\frac{1}{n}\right)$; for any fixed $n$, both $\xi_{\rm out}(\epsilon,\delta,n)$ and $\xi_{\rm in}(\epsilon,\delta,n)$ blow up as $\delta$ approaches $\epsilon$ (see Theorem \ref{thm-k-v} for the case where $\epsilon$ approaches $0$).
\end{proof}

\numberwithin{equation}{section}
\section{Proof of Corollary \ref{cor-sw-dep}}
\label{append-cor-sw-dep}
\renewcommand{\theequation}{\thesection.\arabic{equation}}
1) When $X_1$ and $X_2$ are dependent, our choice of $\mathbf{R} = (R_1, R_2)$ in \eqref{eq-sum-rate-1}--\eqref{eq-sum-rate-2} implies that 
\begin{align}
R_1 &\geq H(X_1|X_2) + \delta_2 - \frac{\log n}{2n} \\
R_2 &\geq H(X_2|X_1) + \delta_1 - \frac{\log n}{2n}.
\end{align} Define
\begin{align}
\mathbf{a} \triangleq 
\begin{bmatrix}
a_1 \\ a_2 \\ a_3
\end{bmatrix} \triangleq \overline{\mathbf{R}} - \overline{\mathbf{H}} + \frac{\log n}{2n}\mathbf{1}. \label{eq-def-a}
\end{align} We have
\begin{align}
a_1 &= R_1 - H(X_1|X_2) + \frac{\log n}{2n} \geq \delta_2 \\
a_2 &= R_2 - H(X_2|X_1) + \frac{\log n}{2n} \geq \delta_1 \\
a_3 &= \sqrt{\frac{V(X_1,X_2)}{n}}Q^{-1}\left(\epsilon - \frac{G}{\sqrt{n}}\right).
\end{align}
Let $\mathbf{Z} \triangleq (Z_1, Z_2, Z_3) \sim \mathcal{N}(\mathbf{0},\mathsf{V})$ be a multivariate Gaussian in $\mathbb{R}^3$, where $\mathsf{V}$ is the entropy dispersion matrix (see Definition~\ref{def-dis-matrix}). Then 
\begin{align}
&~{\mathbb{P} \left[\mathbf{Z} \leq \sqrt{n} \mathbf{a}\right]} \label{eq-append-dep-1}  \\
\geq&~ 1 - (\mathbb{P}\left[Z_1 > a_1\sqrt{n}\right] + \mathbb{P}\left[Z_2 > a_2\sqrt{n}\right] + \mathbb{P}\left[Z_3 > a_3\sqrt{n}\right]), \nonumber 
\end{align} where \eqref{eq-append-dep-1} holds by the union bound. It follows that
\begin{IEEEeqnarray}{rCl}
\mathbb{P}\left[Z_1 > a_1\sqrt{n}\right] &=& \mathbb{P}\left[Z_1 \geq a_1\sqrt{n}\right] \\
&=& Q\left(\frac{a_1\sqrt{n}}{\sqrt{V(X_1|X_2)}}\right) \\
&\leq& e^{-na_1^2/(2V(X_1|X_2))} \label{eq-append-dep-2}\\
&\leq& e^{-n\delta_2^2/(2V(X_1|X_2))}, \label{eq-append-dep-3} \IEEEeqnarraynumspace
\end{IEEEeqnarray} where \eqref{eq-append-dep-2} applies the Chernoff bound of the Q-function, and \eqref{eq-append-dep-3} holds since $a_1 \geq \delta_2 > 0$. Similarly,
\begin{equation}
\mathbb{P}\left[Z_2 > a_2\sqrt{n}\right] \leq e^{-n\delta_1^2/(2V(X_2|X_1))}.
\end{equation} In contrast,
\begin{align}
\mathbb{P}\left[Z_3 > a_3\sqrt{n}\right] = \epsilon - \frac{G}{\sqrt{n}}. \label{eq-append-dep-4}
\end{align} Plugging \eqref{eq-append-dep-3}--\eqref{eq-append-dep-4} into \eqref{eq-append-dep-1}, we conclude that for all $n$ sufficiently large such that
\begin{align}
e^{-n\delta_2^2/(2V(X_1|X_2))} + e^{-n\delta_1^2/(2V(X_2|X_1))} \leq \frac{G}{\sqrt{n}}, 
\end{align}
the bound
\begin{align}
\mathbb{P} \left[\mathbf{Z} \leq \sqrt{n} \mathbf{a}\right] \geq 1 - \epsilon 
\end{align} 
holds.  
Therefore, $\sqrt{n}\mathbf{a} \in \mathscr{Q}_{\rm inv}(\mathsf{V},\epsilon)$, 
and hence $\overline{\mathbf{R}} \in  \overline{\mathscr{R}}^{*}(n,\epsilon)$ \eqref{eq-def-third-order}.

2) Recall vector $\mathbf{a}$ defined in \eqref{eq-def-a}. 
With $R_1 = H(X_1)$, 
\begin{align}
a_1 &= H(X_1) - H(X_1|X_2) + \frac{\log n}{2n} \\
a_2 &= R_2 - H(X_2|X_1) + \frac{\log n}{2n} \\
a_3 &= R_2 - H(X_2|X_1) + \frac{\log n}{2n}.
\end{align} Note that 
\begin{align}
&~{\mathbb{P}\left[\mathbf{Z} \leq \sqrt{n}\mathbf{a} \right]} \label{eq-append-dep-5}
 \\
\geq&~ \mathbb{P} \left[Z_2 \leq a_2\sqrt{n}, \, Z_3 \leq a_3\sqrt{n}\right] - \mathbb{P} \left[ Z_1 > a_1\sqrt{n} \right]. \notag
\end{align} Since $H(X_1) - H(X_1|X_2) > 0$, $\mathbb{P} \left[ Z_1 > a_1\sqrt{n} \right]$ decays exponentially in $n$. Therefore, by the definition of $r^*$ in \eqref{eq-def-r*} and a first-order multivariate Taylor bound, $G > 0$ in \eqref{eq:R2cor} can be chosen so that the right side of \eqref{eq-append-dep-5} is equal to $1 - \epsilon$, which implies that $\overline{\mathbf{R}} \in  \overline{\mathscr{R}}^{*}(n,\epsilon)$ \eqref{eq-def-third-order}.

Conversely, for any $R_2$ such that 
$\overline{\mathbf{R}} \in \overline{\mathscr{R}}^{*}(n,\epsilon)$, 
\begin{align}
\mathbb{P}\left[\mathbf{Z} \leq \sqrt{n}\mathbf{a} \right] \geq 1 - \epsilon,
\end{align} which further implies
\begin{align}
\mathbb{P} \left[Z_2 \leq a_2\sqrt{n}, \, Z_3 \leq a_3\sqrt{n}\right] \geq 1 - \epsilon.
\end{align} Thus, by the definition of $r^*$, 
\begin{align}
\sqrt{n}\left(R_2 - H(X_2|X_1) + \frac{\log n}{2n}\right) \geq r^*,
\end{align} which is equivalent to \eqref{eq:R2corc}. \qed

\numberwithin{equation}{section}
\section{Proof of Corollary \ref{cor-sw-indep}}
\label{append-cor-sw-indep}
\renewcommand{\theequation}{\thesection.\arabic{equation}}
Fix any $\lambda \in [0,1]$. Define
\begin{align}
\mathbf{a} \triangleq 
\begin{bmatrix}
a_1 \\ a_2 \\ a_3
\end{bmatrix} \triangleq \overline{\mathbf{R}} - \overline{\mathbf{H}} + \frac{\log n}{2n}
\begin{bmatrix}
\lambda \\ 1-\lambda \\ 1
\end{bmatrix}.
\end{align} By the assumption that $X_1$ and $X_2$ are independent, we have
\begin{align}
a_3 = a_1 + a_2.
\end{align} Denote 
\begin{align}
r_1 \triangleq \frac{a_1\sqrt{n}}{\sqrt{V(X_1)}}, \; r_2 \triangleq \frac{a_2\sqrt{n}}{\sqrt{V(X_2)}}.
\end{align}
Let $\mathbf{Z} \triangleq (Z_1, Z_2, Z_3) \sim \mathcal{N}(\mathbf{0},\mathsf{V})$ be a multivariate Gaussian in $\mathbb{R}^3$, where $\mathsf{V}$ is the entropy dispersion matrix of the independent sources $X_1$ and $X_2$. It follows in this case that $Z_1$ and $Z_2$ are independent and $Z_3 = Z_1 + Z_2$. We then have
\begin{IEEEeqnarray}{rCl}
\IEEEeqnarraymulticol{3}{l}{\mathbb{P} \left[\mathbf{Z} \leq \sqrt{n} \mathbf{a}\right]} \nonumber \\
&=& \mathbb{P}\left[Z_1 \leq a_1\sqrt{n}\right] \mathbb{P}\left[Z_2 \leq a_2\sqrt{n}\right] \nonumber \\
&& \left.\mathbb{P}\left[Z_3 \leq a_3\sqrt{n} \, \right|Z_1 \leq a_1\sqrt{n}, Z_2 \leq a_2\sqrt{n} \right] \\
&=& \mathbb{P}\left[Z_1 \leq a_1\sqrt{n}\right] \mathbb{P}\left[Z_2 \leq a_2\sqrt{n}\right] \\
&=& \Phi(r_1)\Phi(r_2).
\end{IEEEeqnarray} 
Thus, for any $r_1$, $r_2$ such that
\begin{align}
\Phi(r_1)\Phi(r_2) \geq 1-\epsilon, 
\end{align}
$\mathbf{a} \in \frac{\mathscr{Q}_{\rm inv}(\mathsf{V},\epsilon)}{\sqrt{n}}$ 
and hence
\begin{IEEEeqnarray}{rCl}
\overline{\mathbf{R}} &\in& \overline{\mathbf{H}} + \frac{\mathscr{Q}_{\rm inv}(\mathsf{V},\epsilon)}{\sqrt{n}} - \frac{\log n}{2n}
\begin{bmatrix}
\lambda \\ 1-\lambda \\ 1
\end{bmatrix} 
\subseteq
 \overline{\mathscr{R}}^{*}(n,\epsilon). 
\end{IEEEeqnarray} 
Therefore,
\begin{IEEEeqnarray}{rCl}
\IEEEeqnarraymulticol{3}{l}{\overline{R}^*_{\rm sum}(n,\epsilon) \leq H(X_1) + H(X_2) + } \label{eq:sumrateind} \\
&&\min\limits_{\substack{{(r_1,r_2):} \\ {\Phi(r_1)\Phi(r_2) \geq 1- \epsilon}}} \left(\sqrt{\frac{V(X_1)}{n}}r_1 +  \sqrt{\frac{V(X_2)}{n}}r_2\right) - \frac{\log n}{2n}. \nonumber 
\end{IEEEeqnarray}
On the other hand, for any $r_1$, $r_2$ such that
\begin{equation}
\Phi(r_1)\Phi(r_2) < 1 - \epsilon, 
\end{equation} 
$\mathbf{a} \not\in \frac{\mathscr{Q}_{\rm inv}(\mathsf{V},\epsilon)}{\sqrt{n}}$ 
and hence $\overline{\mathbf{R}} \notin  \overline{\mathscr{R}}^{*}(n,\epsilon)$ 
\eqref{eq-def-third-order}.  
Thus, \eqref{eq:sumrateind} holds with equality. 

\section{Proof of Theorem~\ref{thm-cf-rasc-conv}}
\label{append-cf-rasc-conv}
\renewcommand{\theequation}{\thesection.\arabic{equation}}

	The proof employs an extension of Han's MASC converse~\cite[Lemma~7.2.2]{han}.
	
	Given an $(L, M_1, M_2,\epsilon)$ CF-MASC 
	$(\mathsf{L}, \mathsf{f}_1, \mathsf{f}_2, \mathsf{g})$, 
	let 
	\begin{IEEEeqnarray*}{rCl}	
	\mathcal{S} & \triangleq & 	\big\{(x_1,x_2) \in \mathcal{X}_1\times\mathcal{X}_2: \\
	&&	(x_1,x_2)= \mathsf{g} \left(\mathsf{f}_1\left(\mathsf{L}(x_1,x_2), x_1\right), 
	\mathsf{f}_2\left(\mathsf{L}(x_1,x_2), x_2\right) \right) \big\} \\
	\mathcal{S}_{1}(x_2) & \triangleq & \big\{x_1 \in \mathcal{X}_1: 
	(x_1,x_2) \in \mathcal{S} \big\}\ \ \forall x_2\in{\cal X}_2 \\
	\mathcal{S}_{2}(x_1) & \triangleq & \big\{x_2 \in \mathcal{X}_2: 
	(x_1,x_2) \in \mathcal{S} \big\}\ \ \forall x_1\in{\cal X}_1.
	\end{IEEEeqnarray*}
	Then $\mathbb{P}[\mathcal{S}^c]$ equals the code's error probability, and 
	\begin{IEEEeqnarray}{rcl}
	|\mathcal{S}| &\, \leq \,& M_1 M_2 \\
	|\mathcal{S}_1(x_2)| 
	&\, \leq \,& L M_1, \text{ for any } x_2 \in \mathcal{X}_2 \label{eq-cf-rasc-conv-5}\\
	|\mathcal{S}_2(x_1)| 
	&\, \leq \,& L M_2, \text{ for any } x_1 \in \mathcal{X}_1, 
	\end{IEEEeqnarray}
	where the bound on $|\mathcal{S}|$ is the number of distinct decoder inputs 
	and the bounds on $|\mathcal{S}_1(x_2)|$ and $|\mathcal{S}_2(x_1)|$ 
	are the number of distinct decoder inputs 
	under fixed values of $x_2$ and $x_1$ and an $\ell$-bit CF.  
	Fix $\gamma > 0$. Define sets
	\begin{IEEEeqnarray}{rCl}
	\mathcal{U} &\triangleq& 
	\left\{(x_1,x_2) \in \mathcal{X}_1 \times \mathcal{X}_2: \right. \nonumber \\
	&& \left. \imath(x_1, x_2) \geq \log M_1 + \log M_2 +  \gamma \right\} \\
	\mathcal{U}_1 &\triangleq& 
	\left\{(x_1,x_2) \in \mathcal{X}_1 \times \mathcal{X}_2: \; 
	\imath(x_1|x_2) \geq \log (L M_1 ) + \gamma \right\} \\
	\mathcal{U}_2 &\triangleq& 
	\left\{(x_1,x_2) \in \mathcal{X}_1 \times \mathcal{X}_2: \; 
	\imath(x_2|x_1) \geq \log (L M_2 ) + \gamma \right\}. \IEEEeqnarraynumspace
	\end{IEEEeqnarray}
	Then, 
	\begin{IEEEeqnarray}{rCl}
	\lefteqn{\ \ \ \ \mathbb{P}\left[\mathcal{U}_1 \cap \mathcal{S}\right]}  \label{eq-cf-rasc-conv-0}  \\
	&=& \mathbb{E} \left[1\left\{P_{X_1|X_2}(X_1|X_2) 
	\leq \frac{\exp(-\gamma) }{L M_1  } \right\} 
	1 \left\{(X_1,X_2) \in \mathcal{S} \right\} \right] \notag \\
	&\leq& \sum\limits_{x_2 \in \mathcal{X}_2} 
	P_{X_2}(x_2)  |\mathcal{S}_1(x_2)| \frac{\exp(-\gamma)}{L M_1  } 
	\label{eq-cf-rasc-conv-1} \IEEEeqnarraynumspace \\
	&\leq& \exp(-\gamma), \label{eq-cf-rasc-conv-2}
	\end{IEEEeqnarray} 
	where \eqref{eq-cf-rasc-conv-0} follows the definition of $\mathcal{U}_1$, 
	\eqref{eq-cf-rasc-conv-1} applies $1\{Z \leq z\} \leq z$,
	and \eqref{eq-cf-rasc-conv-2} holds by \eqref{eq-cf-rasc-conv-5}. Similarly,
	\begin{IEEEeqnarray}{rCl}
	\mathbb{P}\left[\mathcal{U}_2 \cap \mathcal{S}\right] &\leq& \exp(-\gamma) \\
	\mathbb{P}\left[\mathcal{U} \cap \mathcal{S}\right] &\leq& \exp(-\gamma).
	\end{IEEEeqnarray}
	Thus, 
	\begin{IEEEeqnarray}{rCl}
	\!\!\!\!\!\!\!\!\!\!\!\!\! \mathbb{P}\left[\mathcal{U}_1 \cup \mathcal{U}_2 \cup \mathcal{U} \right] 
	\!&\leq&\! \mathbb{P}\left[\mathcal{U}_1 \cap \mathcal{S}\right] 
	+ \mathbb{P}\left[\mathcal{U}_2 \cap \mathcal{S}\right] 
	+ \mathbb{P}\left[\mathcal{U} \cap \mathcal{S}\right] 
	+  \mathbb{P}\left[\mathcal{S}^c\right] \\
	&\leq& 3 \exp(-\gamma) 
	+ \mathbb{P}\left[\mathcal{S}^c\right]. \label{eq-cf-rasc-conv-3}
	\end{IEEEeqnarray}
	Rearranging \eqref{eq-cf-rasc-conv-3} 
	gives a lower bound on the error probability 
	$\epsilon = \mathbb{P} \left[\mathcal{S}^c\right]$. 
	Thus, any $(L, M_1, M_2, \epsilon)$ CF-MASC must satisfy
	\begin{IEEEeqnarray}{rCl}
	\epsilon &\geq& \mathbb{P} \left[\left\{\imath(X_1|X_2) \geq \log (L M_1) + \gamma \right\} \cup \nonumber \right.\\
	&& \phantom{\mathbb{P} [} \left\{\imath(X_2|X_1) \geq \log (L M_2) + \gamma \right\} \cup \notag \\
	&& \phantom{\mathbb{P} [} \left. \left\{\imath(X_1, X_2) \geq \log (M_1 M_2) + \gamma \right\}\right] \notag \\
	&& - 3 \exp(-\gamma) . \label{eq-cf-rasc-conv-4}
	\end{IEEEeqnarray}
	Particularizing \eqref{eq-cf-rasc-conv-4} 
	to stationary, memoryless sources 
	with single-letter distribution $P_{X_1X_2}$ 
	satisfying \eqref{assump-b1} and \eqref{assump-b2} 
	shows that any $(n,L, M_1, M_2, \epsilon)$ CF-MASC must satisfy
	\begin{IEEEeqnarray}{rCl}
	\epsilon &\geq& \mathbb{P} \left[\left\{I_{1} \geq \log (L M_1) 
	+ \gamma \right\} \cup \left\{I_{2} \geq \log (L M_2)  + \ell+ \gamma \right\} \right. \nonumber \\
	&& \phantom{\mathbb{P} [} \cup \left. \left\{I_{12} \geq \log ( M_1 M_2)  
	+ \gamma \right\}\right] - 3 \exp(-\gamma) 
	\IEEEeqnarraynumspace \\
	&=& 1 - \mathbb{P}\left[\sum_{i = 1}^n \mathbf{U}_i  < n \overline{\mathbf{R}} 
	- n \overline{\mathbf{H}} + \gamma\mathbf{1}\right] - 3 \exp\left(-\gamma\right), \label{eq-cf-rasc-conv-6} 
	\end{IEEEeqnarray} 
	where $\gamma > 0$ is an arbitrary constant, $I_{1}$, $I_{2}$, and $I_{12}$ 
	are defined in \eqref{eq-def-I1n}, \eqref{eq-def-I2n}, and \eqref{eq-def-In}, 
	and $\mathbf{U}_i$ is defined in \eqref{vardef-ui}. 
	Let $L$ be a finite constant that does not grow with $n$ and let $\gamma = \frac{\log n}{2} - \log L$.  
	Applying Lemma~\ref{lem-b-e} and Lemma~\ref{lem-sve}-\ref{part2} to bound the probability \eqref{eq-cf-rasc-conv-6} in a manner similar to \eqref{eq:MASCrates}--\eqref{eq-sw-achiev-3}, 
	we conclude that any $(n,\ell,\epsilon)$-achievable rate pair $(R_1, R_2)$ must be in $\mathscr{R}_{\rm out}^{*}(n,\epsilon)$ \eqref{eq-def-sw-2-out}. \qed

\begin{remark}
One could also prove Theorem~\ref{thm-cf-rasc-conv} by extending our HT converse (Theorem~\ref{thm-sw-cht-conv} ) to the setting with a cooperation facilitator. 
Our Theorem~\ref{thm-sw-cht-conv} continues to hold 
with $M_1$ and $M_2$ replaced by $L M_1$ and $L M_2$ 
in \eqref{eq:beta1star} and \eqref{eq:beta2star}, respectively 
(\eqref{eq:beta3star} remains unchanged).
\end{remark}

\section*{Acknowledgment}
The authors would like to thank the anonymous reviewer for the especially careful review that is reflected in the final version.

\bibliographystyle{IEEEtran}
\bibliography{IEEEabrv,bibtex/bib/refs}

\begin{thebibliography}{10}
\providecommand{\url}[1]{#1}
\csname url@samestyle\endcsname
\providecommand{\newblock}{\relax}
\providecommand{\bibinfo}[2]{#2}
\providecommand{\BIBentrySTDinterwordspacing}{\spaceskip=0pt\relax}
\providecommand{\BIBentryALTinterwordstretchfactor}{4}
\providecommand{\BIBentryALTinterwordspacing}{\spaceskip=\fontdimen2\font plus
\BIBentryALTinterwordstretchfactor\fontdimen3\font minus
  \fontdimen4\font\relax}
\providecommand{\BIBforeignlanguage}[2]{{%
\expandafter\ifx\csname l@#1\endcsname\relax
\typeout{** WARNING: IEEEtran.bst: No hyphenation pattern has been}%
\typeout{** loaded for the language `#1'. Using the pattern for}%
\typeout{** the default language instead.}%
\else
\language=\csname l@#1\endcsname
\fi
#2}}
\providecommand{\BIBdecl}{\relax}
\BIBdecl

\bibitem{chen-e-k}
S.~Chen, M.~Effros, and V.~Kostina, ``Lossless source coding in the
  point-to-point, multiple access, and random access scenarios,'' in
  \emph{Proc. {IEEE} Int. Symp. Inf. Theory (ISIT)}, Jul. 2019, pp. 1692--1696.

\bibitem{spectre}
\BIBentryALTinterwordspacing
{SPECTRE}: Short packet communication toolbox. [Online]. Available:
  \url{https://github.com/yp-mit/spectre/tree/master/lossless-sc}
\BIBentrySTDinterwordspacing

\bibitem{shannon}
C.~E. Shannon, ``A mathematical theory of communication,'' \emph{The Bell
  System Technical Journal}, vol.~27, no.~2, pp. 379--423 and 623--656, July
  and Oct. 1948.

\bibitem{strassen}
V.~Strassen, ``Asymptotische absch{\"a}zungen in shannons
  informationstheorie,'' in \emph{Proc. Trans. Third Prague Conf. Inf. Theory,
  Statist., Decision Funct., Random Process.}, 1964, pp. 689--723.

\bibitem{pol-poo-ver}
Y.~Polyanskiy, H.~V. Poor, and S.~Verd\'{u}, ``Channel coding rate in the
  finite blocklength regime,'' \emph{{IEEE} Trans. Inf. Theory}, vol.~56,
  no.~5, pp. 2307--2359, May 2010.

\bibitem{kontoyiannis-verdu}
I.~Kontoyiannis and S.~Verd\'{u}, ``Optimal lossless data compression:
  Non-asymptotics and asymptotics,'' \emph{{IEEE} Trans. Inf. Theory}, vol.~60,
  no.~2, pp. 777--795, Feb. 2014.

\bibitem{kostina-pol-verdu}
V.~Kostina, Y.~Polyanskiy, and S.~Verd\'{u}, ``Variable-length compression
  allowing errors,'' \emph{{IEEE} Trans. Inf. Theory}, vol.~61, no.~8, pp.
  4316--4330, Aug. 2015.

\bibitem{yushkevich}
A.~A. Yushkevich, ``On limit theorems connected with the concept of entropy of
  {Markov} chains,'' \emph{Uspekhi Mat. Nauk}, vol.~8, no. 5(57), pp. 177--180,
  1953.

\bibitem{han}
T.~S. Han, \emph{Information-Spectrum Methods in Information Theory}.\hskip 1em
  plus 0.5em minus 0.4em\relax Springer-Verlag Berlin Heidelberg, 2003.

\bibitem{hayashi}
M.~Hayashi, ``Second-order asymptotics in fixed-length source coding and
  intrinsic randomness,'' \emph{{IEEE} Trans. Inf. Theory}, vol.~54, no.~10,
  pp. 4619--4637, Aug. 2008.

\bibitem{slepian-wolf}
D.~Slepian and J.~K. Wolf, ``Noiseless coding of correlated information
  sources,'' \emph{{IEEE} Trans. Inf. Theory}, vol.~19, no.~4, pp. 471--480,
  Jul. 1973.

\bibitem{miyake-kanaya}
S.~Miyake and F.~Kanaya, ``Coding theorems on correlated general sources,''
  \emph{IEICE Trans. on Fundamentals of Electronics}, vol. E78-A, no.~9, pp.
  1063--1070, Sep. 1995.

\bibitem{tan-kosut}
V.~Y.~F. Tan and O.~Kosut, ``On the dispersions of three network information
  theory problems,'' \emph{{IEEE} Trans. Inf. Theory}, vol.~60, no.~2, pp.
  881--903, Feb. 2014.

\bibitem{nomura-han}
R.~Nomura and T.~S. Han, ``Second-order {Slepian-Wolf} coding theorems for
  non-mixed and mixed sources,'' \emph{{IEEE} Trans. Inf. Theory}, vol.~60,
  no.~9, pp. 5553--5572, Sep. 2014.

\bibitem{jose-k}
S.~T. {Jose} and A.~A. {Kulkarni}, ``Improved finite blocklength converses for
  {Slepian--Wolf} coding via linear programming,'' \emph{{IEEE} Trans. Inf.
  Theory}, vol.~65, no.~4, pp. 2423--2441, Apr. 2019.

\bibitem{minero-tse}
P.~Minero, M.~Franceschetti, and D.~N.~C. Tse, ``Random access: An
  information-theoretic perspective,'' \emph{{IEEE} Trans. Inf. Theory},
  vol.~58, no.~2, pp. 909--930, Feb. 2012.

\bibitem{ra-polyanskiy}
Y.~Polyanskiy, ``A perspective on massive random-access,'' in \emph{Proc.
  {IEEE} Int. Symp. Inf. Theory (ISIT)}, Jun. 2017, pp. 2523--2527.

\bibitem{recep}
M.~Effros, V.~Kostina, and R.~C. Yavas, ``Random access channel coding in the
  finite blocklength regime,'' in \emph{Proc. {IEEE} Int. Symp. Inf. Theory
  (ISIT)}, Jun. 2018, pp. 1261--1265.

\bibitem{burnashev}
M.~V. Burnashev, ``Data transmission over a discrete channel with feedback:
  Random transmission time,'' \emph{Problems of Information Transmission},
  vol.~12, no.~4, pp. 10--30, Aug. 1976.

\bibitem{tt}
A.~Tchamkerten and I.~E. Telatar, ``Variable length coding over an unknown
  channel,'' \emph{{IEEE} Trans. Inf. Theory}, vol.~52, no.~5, pp. 2126--2145,
  May 2006.

\bibitem{ppv2}
Y.~Polyanskiy, H.~V. Poor, and S.~Verd{\'u}, ``Feedback in the non-asymptotic
  regime,'' \emph{{IEEE} Trans. Inf. Theory}, vol.~57, no.~8, pp. 4903--4925,
  Aug. 2011.

\bibitem{draper}
S.~Draper, ``Universal incremental {Slepian-Wolf} coding,'' in \emph{Proc. 42th
  Annual Allerton Conference on Communications, Control, and Computing}, Sep.
  2004.

\bibitem{kostina-verdu}
V.~Kostina and S.~Verd\'{u}, ``Fixed-length lossy compression in the finite
  blocklength regime,'' \emph{{IEEE} Trans. Inf. Theory}, vol.~58, no.~6, pp.
  3309--3338, Jun. 2012.

\bibitem{verdu-notes}
S.~Verd\'{u}. Ee528 - information theory, lecture notes. Princeton University,
  Princeton, NJ, 2007.

\bibitem{polyanskiy-notes}
\BIBentryALTinterwordspacing
Y.~Polyanskiy and Y.~Wu. Lecture notes on information theory. MIT (6.441), UIUC
  (ECE 563), Yale (STAT 664), 2012-2017. [Online]. Available:
  \url{http://people.lids.mit.edu/yp/homepage/data/itlectures_v5.pdf}
\BIBentrySTDinterwordspacing

\bibitem{feller}
W.~Feller, \emph{An Introduction to Probability Theory and its Applications},
  2nd~ed.\hskip 1em plus 0.5em minus 0.4em\relax John Wiley \& Sons, 1971,
  vol.~II.

\bibitem{shevtsova2013absolute}
I.~G. Shevtsova, ``On the absolute constants in the {B}erry-{E}sseen inequality
  and its structural and nonuniform improvements,'' \emph{Informatika i Ee
  Primeneniya [Informatics and its Applications]}, vol.~7, no.~1, pp. 124--125,
  2013.

\bibitem{nagaoka2005strong}
H.~Nagaoka, ``Strong converse theorems in quantum information theory,'' in
  \emph{Asymptotic Theory of Quantum Statistical Inference: Selected
  Papers}.\hskip 1em plus 0.5em minus 0.4em\relax World Scientific, 2005, pp.
  64--65.

\bibitem{hayashi2006quantum}
M.~Hayashi, \emph{Quantum Information: An introduction}.\hskip 1em plus 0.5em
  minus 0.4em\relax Berlin Heidelberg: Springer, Apr. 2006.

\bibitem{campo}
A.~T. {Campo}, G.~{Vazquez-Vilar}, A.~G. i.~{F{\`a}bregas}, and A.~{Martinez},
  ``Converse bounds for finite-length joint source-channel coding,'' in
  \emph{Proc. 50th Annual Allerton Conference on Communications, Control, and
  Computing}, Oct. 2012.

\bibitem{kostina-jscc}
V.~Kostina and S.~Verd\'{u}, ``Lossy joint source-channel coding in the finite
  blocklength regime,'' \emph{{IEEE} Trans. Inf. Theory}, vol.~59, no.~5, pp.
  2545--2575, May 2013.

\bibitem{huang-moulin}
Y.~Huang and P.~Moulin, ``Strong large deviations for composite hypothesis
  testing,'' in \emph{Proc. {IEEE} Int. Symp. Inf. Theory (ISIT)}, Jun. 2014,
  pp. 556--560.

\bibitem{recep-arxiv}
M.~Effros, V.~Kostina, and R.~C. Yavas, ``Random access channel coding in the
  finite blocklength regime,'' \emph{arXiv:1801.09018v3 [cs.IT]}, Jul. 2019.

\bibitem{elkayam-feder}
N.~Elkayam and M.~Feder, ``On the calculation of the minimax-converse of the
  channel coding problem,'' in \emph{Proc. {IEEE} Int. Symp. Inf. Theory
  (ISIT)}, Jun. 2017, pp. 1247--1251.

\bibitem{bentkus}
V.~Bentkus, ``On the dependence of the {Berry-Esseen} bound on dimension,''
  \emph{J. Stat. Planning and Inference}, vol. 113, pp. 385--402, May 2003.

\bibitem{noorzad-ho}
P.~Noorzad, M.~Effros, M.~Langberg, and T.~Ho, ``On the power of cooperation:
  Can a little help a lot?'' in \emph{Proc. {IEEE} Int. Symp. Inf. Theory
  (ISIT)}, June/July 2014, pp. 3132--3136.

\bibitem{langberg}
M.~Langberg and M.~Effros, ``Network coding: Is zero error always possible?''
  in \emph{Proc. 49th Annual Allerton Conference on Communications, Control,
  and Computing}, Sep. 2011.

\bibitem{oohama-han}
Y.~Oohama and T.~S. Han, ``Universal coding for the {Slepian-Wolf} data
  compression system and the strong converse theorem,'' \emph{{IEEE} Trans.
  Inf. Theory}, vol.~40, no.~6, pp. 1908--1919, Nov. 1994.

\bibitem{csiszar-korner}
I.~Csisz\'{a}r and J.~K{\"{o}}rner, ``Towards a general theory of source
  networks,'' \emph{{IEEE} Trans. Inf. Theory}, vol.~26, no.~2, pp. 155--165,
  Mar. 1980.

\bibitem{oohama}
Y.~Oohama, ``Universal coding for correlated sources with linked encoders,''
  \emph{{IEEE} Trans. Inf. Theory}, vol.~42, no.~3, pp. 837--847, May 1996.

\bibitem{jaggi-effros}
S.~Jaggi and M.~Effros, ``Universal linked multiple access source codes,'' in
  \emph{Proc. {IEEE} Int. Symp. Inf. Theory (ISIT)}, Jun. 2002, p.~95.

\bibitem{yang-yeung}
E.~Yang, D.~He, T.~Uyematsu, and R.~W. Yeung, ``Universal multiterminal source
  coding algorithms with asymptotically zero feedback: Fixed database case,''
  \emph{{IEEE} Trans. Inf. Theory}, vol.~54, no.~12, pp. 5575--5590, Dec. 2008.

\bibitem{sarvotham-baraniuk}
S.~Sarvotham, D.~Baron, and R.~G. Baraniuk, ``Variable-rate universal
  {Slepian-Wolf} coding with feedback,'' in \emph{Proc. 39th Asilomar Conf. on
  Signals, Systems, and Computers}, Oct. 2005, pp. 8--12.

\bibitem{eckford-yu}
A.~W. Eckford and W.~Yu, ``Rateless {Slepian-Wolf} codes,'' in \emph{Proc. 39th
  Asilomar Conf. on Signals, Systems, and Computers}, Oct. 2005, pp.
  1757--1761.

\end{thebibliography}

\begin{IEEEbiographynophoto}{Shuqing Chen}(Graduate Student Member, IEEE) received the B.S. degree from Rice University in 2017 and the M.S. (2020) from California Institute of Technology, in electrical engineering. She worked with the Data Compression Laboratory, Caltech, from 2017 to 2019. Her research involves non-asymptotic information theory, composite hypothesis testing and multi-terminal source coding.
\end{IEEEbiographynophoto}

\begin{IEEEbiographynophoto}{Michelle Effros}(Fellow, IEEE)
received the B.S. (Hons.), M.S., and Ph.D. degrees in
electrical engineering from Stanford University,    
in 1989, 1990, and 1994, respectively. 

In 1994, she joined the faculty at the
California Institute of Technology, where she is currently the George
Van Osdol Professor of electrical engineering. Her research interests include
information theory, network coding, data compression, and communications.
She received Stanford's Frederick Emmons Terman Engineering Scholastic
Award (for excellence in engineering) in 1989, the Hughes Masters Full-Study
Fellowship in 1989, the National Science Foundation Graduate Fellowship in 1990,
the AT\&T Ph.D. Scholarship in 1993, the NSF CAREER Award in 1995, the Charles
Lee Powell Foundation Award in 1997, the Richard Feynman-Hughes Fellowship in
1997, and an Okawa Research Grant in 2000. She was cited by Technology Review as
one of the world's top young innovators in 2002. She and her co-authors
received the Communications Society and Information Theory Society Joint Paper
Award in 2009. She became a fellow of the IEEE in 2009.
She is a member of Tau Beta Pi, Phi Beta Kappa, and Sigma Xi. She served as the
Editor of the IEEE Information Theory Society Newsletter 1995 to 1998 and as a
Member of the Board of Governors of the IEEE Information Theory Society from
1998 to 2003 and from 2008 to 2017. She also served as President of the
Information Theory Society in 2015. She was a member of the Advisory Committee and
the Committee of Visitors for the Computer and Information Science and
Engineering (CISE) Directorate at the National Science Foundation from 2009 to
2012 and in 2014, respectively. She served on the IEEE Signal Processing Society
Image and Multi-Dimensional Signal Processing (IMDSP) Technical Committee from
2001 to 2007 and on ISAT from 2006 to 2009. She served as Associate Editor for
the joint special issue on Networking and Information Theory in the {\em IEEE
Transactions on Information Theory} and the {\em IEEE Transactions on Networking/ACM Transactions on Networking} 
and as Associate Editor for Source Coding for the
{\em IEEE Transactions on Information Theory} from 2004 to 2007. She has served on
numerous technical program committees and review boards, including serving as
general co-chair for the 2009 Network Coding Workshop and technical program
committee co-chair for the 2012 IEEE International Symposium on Information
Theory.
\end{IEEEbiographynophoto}

\begin{IEEEbiographynophoto}{Victoria Kostina}(Member, IEEE)
received the bachelor's degree from the Moscow institute of Physics and Technology in 2004, the master?s degree 
from the University of Ottawa in 2006, and the Ph.D. degree from Princeton University in 2013. She was affiliated with the Institute for Information Transmission Problems, Russian Academy of Sciences. In 2014, she joined Caltech, where she is currently a Professor of electrical engineering. Her research spans information theory, coding, control, learning, and communications. She received the Natural Sciences and Engineering Research Council of Canada master's scholarship in 2009, the Princeton Electrical Engineering Best Dissertation Award in 2013, the Simons-Berkeley Research Fellowship in 2015, and the NSF CAREER Award in 2017.
\end{IEEEbiographynophoto}

\end{document}